\documentclass[11pt,oneside]{article}

\usepackage[round,colon,authoryear]{natbib} 

\usepackage{amsmath}
\usepackage{amsfonts}
\usepackage{amsthm}
\usepackage{calrsfs}
\DeclareMathAlphabet{\pazocal}{OMS}{zplm}{m}{n}
\usepackage[mathscr]{euscript}
\usepackage{amssymb}
\usepackage[shortlabels]{enumitem}
\usepackage{stmaryrd}
\usepackage{bm}

\usepackage{bbm}					

\numberwithin{equation}{section}

\newtheorem{thm}{Theorem}[section]
\newtheorem{lem}[thm]{Lemma}
\newtheorem{prop}[thm]{Proposition}
\newtheorem{cor}[thm]{Corollary}

\theoremstyle{definition}
\newtheorem{defn}[thm]{Definition}

\theoremstyle{remark}
\newtheorem{example}[thm]{Example}

\theoremstyle{remark}
\newtheorem{rem}[thm]{Remark}

\title{
Open Markets
}

\vspace{2mm}

\author{  
\textsc{Donghan Kim} \thanks{ 
Department of Mathematics, Columbia University, New York, NY 10027 (E-mail: {\it dk2571@columbia.edu}).
}
}

\begin{document}

\maketitle

\bigskip

\begin{abstract}
\noindent
An open market is a subset of entire equity market composed of a certain fixed number of top-capitalization stocks. Though the number of stocks in the open market is fixed, the constituents of the market change over time as each company's rank by its market capitalization fluctuates. When one is allowed to invest also in money market, the open market resembles the entire `closed' equity market in the sense that the equivalence of market viability (lack of arbitrage) and the existence of num\'eraire portfolio (portfolio which cannot be outperformed) holds. When the access to the money market is prohibited, some topics such as Capital Asset Pricing Model~(CAPM), construction of functionally-generated portfolios, and the concept of universal portfolio are presented in open market setting.
\end{abstract}

\smallskip

\smallskip

\input amssym.def
\input amssym

\smallskip

\section{Introduction}
\label{sec: intro}

Equity markets are conventionally assumed to be \textit{closed}, in the sense that they are almost universally assumed to consist of a given, fixed number of stocks at all times. However, this assumption fails to represent most real markets, where new stocks enter and some others exit due to privatization, bankruptcy, or simply bad luck.

\smallskip

The number of companies in the U.S. stock market has undergone wide fluctuations. In 1975, there were around 4,800 U.S. domiciled firms listed on the NYSE, Amex, and Nasdaq. This number reached a peak of 7,500 listed firms in 1997, and then decreased by more than half to 3,600 firms 20 years later in 2017.

\smallskip

To mitigate the assumption of a fixed, immutable collection of companies, and to model stock markets more realistically, we study here markets that are \textit{open}. These are constructed by restricting our `investing space' from the entire market to the subset composed of a certain fixed number $n$ of top-capitalization stocks. More precisely, within the entire stock market, we keep track of the price dynamics of all stocks, rank them by order of market capitalization, consider an open market consisting of the top $n$ stocks, and only invest in stocks that belong to this open market. High-capitalization indexes, such as the S\&P 500 index, where one invests only in the $n = 500$ highest-capitalization companies, and any given stock is replaced by another one when its capitalization falls, are of this type.

\smallskip

In this paper, we present some results from \textit{closed} markets which remain valid also in \textit{open} markets. The main result of this type involves the concept of ``market viability'', which is understood as ``lack of a certain egregious form of arbitrage'', the condition which prohibits financing non-trivial liabilities starting from arbitrarily small initial capital. The result shows that in an open stock market, with an access to the money market, viability is equivalent to any one of the following conditions: (i) a portfolio with the local martingale num\'eraire property exists; (ii) a local martingale deflator exists; (iii) the market has locally finite maximal growth; (iv) a growth-optimal portfolio exists; and (v) a portfolio with the log-optimality property exists. We provide precise definitions for these terms, and show that this equivalence can be formulated in terms of the drifts and covariations of the underlying stock prices, modeled by continuous semimartingales.

\smallskip

When the access to the money market is forbidden and one can only allowed to invest in fixed number $n$ of top-capitalization stocks, the class of eligible portfolios diminishes significantly as portfolios must satisfy `self-financing' condition. Under this extra condition, we provide a connection of the above viability theory to the Capital Asset Pricing Model~(CAPM), develop a way for constructing functionally-generated portfolios, and discuss universal portfolio in an open market.

\medskip

\noindent
\textit{Preview} : Section~\ref{sec: setup} defines open markets, investment strategies, and portfolios, as well as other notions which will be needed throughout this paper. Section~\ref{sec: main} develops arbitrage theory in open markets, along with the concepts of market viability and num\'eraire. We provide definitions and properties for these concepts, then state and prove the main result. Section~\ref{sec: stock portfolio and FGP} explores stock portfolios, CAPM theory, functional generation of portfolios, and universal portfolio in open markets. Section~\ref{sec: conclusion} provides some concluding remarks.

\bigskip

\section{Portfolios in Open Markets}
\label{sec: setup}

Let us suppose that the ``whole equity market universe" is composed of $N$ stocks and that we are only interested in investing in the top $n$ largest capitalization stocks, for some fixed $1 \leq n < N$. For example, when our investing universe is the entire U.S. stock market, by setting $n=500$ we are investing in those large companies which form the S\&P 500 index. In order to invest in these top $n$ stocks, we must keep track of the rank of each stock's capitalization at all times, and put together a portfolio composed of the $n$ stocks with the largest capitalizations.

\smallskip

Throughout this paper, we fix two positive integers $n$ and $N$ satisfying $1 \leq n < N$ as above. We suppose that trading is continuous in time, with no transaction costs or taxes, and that shares are infinitely divisible. We also assume without loss of generality that each stock has a single share outstanding, and the price of a stock is equal to its capitalization; thus, we use the terms `price of stock' and `capitalization of stock' interchangeably. We also assume that stock prices are discounted by the money market, and adopt the convention that the money market pays and charges zero interest.

\bigskip

\subsection{Stock prices and their ranks}

We first present the following definition of price process in the market described above.

\smallskip

\begin{defn} [Price process]	\label{Def : price process by name}
	For an $N$-dimensional vector $S \equiv (S_1, \cdots, S_n, \cdots, S_N)$ of continuous semimartingales on a filtered probability space $(\Omega, \pazocal{F}, \pazocal{F}(\cdot), \mathbb{P})$, we call $S$ a \textit{vector of price processes} if each component is strictly positive, i.e., the inequalities $S_i(t) > 0$ hold for all $i \in \{1, \cdots, N\}$ at any time $t \geq 0$. The component processes of $S$ represent the stock prices, or the capitalizations, of $N$ companies.
\end{defn}

\smallskip

We need now to clarify some notation regarding ranks. Given the vector $S$ of price processes, we define the $k$-th ranked process $S_{(k)}(\cdot)$ of $S_1, \cdots, S_N$ by
\begin{equation}	\label{def : rank}
	S_{(k)}(\cdot) := \max_{1 \leq i_1 < \cdots < i_k \leq N} \min\{S_{i_1}(\cdot), \cdots, S_{i_k}(\cdot)\}.
\end{equation}
To be more specific, for any $t \geq 0$, we have
\begin{equation}		\label{Eq : rank}
	\max_{i=1, \cdots, N} S_i(t) = S_{(1)}(t) \geq S_{(2)}(t) \geq \cdots \geq S_{(N-1)}(t) \geq S_{(N)}(t) = \min_{i=1, \cdots, N} S_i(t);
\end{equation}
that is, we rank the components of the vector process $S = (S_1, \cdots, S_N)$ in descending order, with the lexicographic rule for breaking ties that always assigns a higher rank~(i.e., a smaller $(k)$) to the smallest index $i$.

\smallskip

\begin{defn} [Price process by rank]	\label{Def : price process by rank}
	For the vector $S$ of price processes in Definition~\ref{Def : price process by name}, we call the $N$-dimensional vector process
	\begin{equation}		\label{def : tilde S}
		\bm{S}(t) \equiv \big(S_{(1)}(t), \cdots, S_{(N)}(t)\big), \quad t \geq 0
	\end{equation}
	where each component is defined via \eqref{def : rank}, the \textit{vector of price processes by rank}. In particular, the $k$-th component $\bm{S}_k(t) = S_{(k)}(t)$ of the vector $\bm{S}(t)$ represents the price of $k$-th ranked stock among $N$ companies at time $t$.
\end{defn}

\smallskip

Each component of the vector process $\bm{S}(\cdot)$ is also a continuous semimartingale, from the results in \cite{Banner:Ghomrasni}. Along with the notation \eqref{def : rank}, we define a process $\{1, \cdots, N\} \times [0, \infty) \ni (i, t) \mapsto u_{i}(t) \in \{1, \cdots, N\}$, such that each $u_i(\cdot)$ is predictable and satisfies
\begin{equation}		\label{def : u}
	S_i(t) = S_{\big(u_i(t)\big)}(t), \quad \forall~ t \geq 0,
\end{equation}
for every $i = 1, \cdots, N$. In other words, $u_i(t)$ is the rank of the $i$-th stock $S_i(t)$ at time $t$, for any given index $i = 1, \cdots, N$. Note that for every fixed $t \geq 0$, the function $u_{\cdot}(t) : \{1, \cdots, N\} \rightarrow \{1, \cdots, N\}$ is a bijection, because we break ties using the lexicographic rule when defining \eqref{def : rank}.

\bigskip

\subsection{Cumulative return processes}	\label{subsec : cumulative return}

In this subsection, we present the notion of cumulative returns of the market. We first define the stochastic logarithm $\pazocal{L}(Y)$ of a positive continuous semimartingale $Y$ with $Y(0) = 1$ by
\begin{equation}		\label{def : stochastic logarithm}
	\pazocal{L}(Y) := \int _0^{\cdot}\frac{dY(t)}{Y(t)}
\end{equation}
and consider the vector $R \equiv (R_1, \cdots, R_N)$, whose every component is the stochastic logarithm of the corresponding normalized component of $S$ in Definition~\ref{Def : price process by name}:
\begin{equation}		\label{def : R}
	R_i
	:= \pazocal{L}\bigg(\frac{S_i}{S_i(0)}\bigg),
	\qquad i = 1, \cdots, N.
\end{equation}
Each component process $R_i$ is again a semimartingale, and represents the \textit{cumulative returns of the $i$-th stock}, since its dynamic is represented as
\begin{equation}		\label{Eq : tilde R dynamics}
	dR_i(t) = \frac{dS_i(t)}{S_i(t)}, \quad t \geq 0, \quad \text{and} \quad R_i(0) = 0 \quad \text{for} \quad i = 1, \cdots, N.
\end{equation}

\bigskip

We posit the semimartingale decomposition
\begin{equation}	\label{Eq : decomposition of R}
	R_i = A_i + M_i, \quad i = 1, \cdots, N,
\end{equation}
for each component of the vector $R = (R_1, \cdots, R_N)$. Here, the component $A_i$ of the vector process $A \equiv (A_1, \cdots, A_N)$ with $A_i(0) = 0$ is adapted, continuous and of finite variation on compact time intervals; whereas each component $M_i$ of the vector process $M \equiv (M_1, \cdots, M_N)$ is a continuous local martingale with $M_i(0) = 0$, for $i = 1, \cdots, N$. We think of the finite variation processes $A_i$ as the `drift components', and of the local martinagles $M_i$ as the `noise components', of $R$.

\smallskip

We define next the continuous, nondecreasing scalar process
\begin{equation}		\label{def : O}
	O := \sum_{i=1}^N \bigg( \int_0^{\cdot} |dA_i(t)| + d[M_i, M_i](t) \bigg),
\end{equation}
where $\int_0^T|dA_i(t)|$ denotes the total variation of $A_i$ on the interval $[0, T]$ for $T \geq 0$ and $[M_i, M_j]$ represents the covariation process of the continuous semimartingales $M_i$ and $M_j$ for $1 \leq i, j \leq N$. Here, we note that $[R_i, R_j] = [M_i, M_j]$ holds from \eqref{Eq : decomposition of R}. This scalar process $O$ plays the role of an ``operational clock'' for the vector $R$. All processes $A_i$ and $[M_i, M_j]$ for $1 \leq i, j \leq N$ are absolutely continuous with respect to this clock, and thus, by the Radon-Nikod\'ym Theorem, there exist two predictable processes
\begin{equation}		\label{def : alpha and c}
	\alpha \equiv (\alpha_i)_{1 \leq i \leq N} \quad \text{and} \quad c \equiv (c_{i, j})_{1 \leq i, j \leq N},
\end{equation}
vector-valued and matrix-valued, respectively, such that
\begin{equation}		\label{def : A and C}
	A = \int_0^{\cdot} \alpha(t)dO(t), \quad \text{and} \quad C \equiv [M, M] = \int_0^{\cdot} c(t)dO(t).
\end{equation}
Here and in what follows, we write $C \equiv (C_{i, j})_{1 \leq i, j \leq N}$ for the nonnegative-definite, matrix-valued process of covariations
\begin{equation}		\label{def : C}
	C_{i, j} := [M_i, M_j] = [R_i, R_j], \quad \text{for} \quad 1 \leq i, j \leq N.
\end{equation}
The component $\alpha_i$ in \eqref{def : alpha and c} represents the \textit{local rate of return} of the $i$-th stock in the market; whereas the entry $c_{i, j}$ stands for the \textit{local covariation rate} of the $i$-th and $j$-th stocks. We call the collection of local rates $\alpha, c$ in \eqref{def : alpha and c} the \textit{local characteristics} of the market, and these rates are measured with respect to the operational clock $O$ in \eqref{def : O}.

\bigskip

For a continuous vector-valued semimartingale $Y = (Y_1, \cdots, Y_N)$, we denote by $\pazocal{I}(Y)$ the class of predictable vector processes $\pi = (\pi_1, \cdots, \pi_N)$ which are integrable with respect to the vector $Y$. In particular, for the collection $\pazocal{I}(R)$ of the vector $R$ in \eqref{def : R} and \eqref{Eq : decomposition of R}, we have a very convenient characterization: A predictable vector process $\pi = (\pi_1, \cdots, \pi_N)$ belongs to $\pazocal{I}(R)$, if and only if
\begin{equation}		\label{con : integrable w.r.t R}
	\int_0^T \Big( |\pi'(t)\alpha(t)| + \pi'(t)c(t)\pi(t)\Big) dO(t) < \infty, \quad \text{for any} \quad T \geq 0.
\end{equation}
We denote then by
\begin{equation*}
	\int_0^{\cdot} \sum_{i=1}^N \pi_i(t) dR_i(t) \equiv \int_0^{\cdot} \pi'(t)dR(t) = \int_0^{\cdot} \pi'(t)dA(t) + \int_0^{\cdot} \pi'(t)dM(t),
\end{equation*}
the stochastic integral of $\pi \in \pazocal{I}(R)$, with respect to the vector semimartingale $R$.

\bigskip

\subsection{Investment strategies and portfolios}

Along with the $N$-dimensional vector $S$ of Definition~\ref{Def : price process by name}, representing the stock prices of the market, we introduce the following notions.

\smallskip

\begin{defn} [Investment strategy, wealth process, and num\'eraire]	\label{Def : investment}
	We call an $N$-dimensional vector of predictable process $\vartheta \equiv (\vartheta_1, \cdots, \vartheta_N)$ \textit{investment strategy}, if it is integrable with respect to the price vector $S$, i.e., $\vartheta \in \pazocal{I}(S)$.
	For any nonnegative real number $x$, we call
	\begin{equation}		\label{def : investment}
		X(\cdot; x, \vartheta) := x + \int_0^{\cdot} \vartheta'(t)dS(t) \equiv x + \int_0^{\cdot} \sum_{i=1}^N \vartheta_i(t)dS_i(t)
	\end{equation}
	the \textit{wealth process} generated by $\vartheta$ with initial capital $x$. We call the wealth process \textit{num\'eraire}, if $X(\cdot; 1, \vartheta) > 0$ holds for the normalized initial capital $x=1$. The collection of all num\'eraires is denoted by $\pazocal{X}$.
\end{defn}

\smallskip

The $i$-th component $\vartheta_i(t)$ represents the units of investment (or number of shares) held in the $i$-th stock at time $t$, and plays the role of integrand with integrator $dS_i(t)$ in the stochastic integral of \eqref{def : investment}. The requirement $X(0) = x = 1$ in defining num\'eraires is a simple normalization, because $X(\cdot;cx, c\vartheta) = cX(\cdot; x, \vartheta)$ holds for any positive real number $c$. Since we consider investment only in the top $n$ stocks, we need a similar definition of investment strategy for this particular case.

\smallskip

\begin{defn} [Investment strategy among the top \textit{n} stocks]	\label{Def : investment among n}
	We call an investment strategy $\vartheta \in \pazocal{I}(S)$ an \textit{investment strategy among the top $n$ stocks}, if the ``sensoring'' equalities
	\begin{equation}		\label{con : investment among n}
		\vartheta_i(t)\bm{1}_{\{u_i(t) > n\}} = 0, \quad \text{for} \quad i = 1, \cdots, N, \quad t \geq 0,
	\end{equation}
	hold with the notation \eqref{def : u}.
	
	The wealth process and the num\'eraire associated with this investment strategy $\vartheta$ among the top $n$ stocks, are defined in the same manner as in Definition~\ref{Def : investment}. We denote the collection of $N$-dimensional predictable processes $\vartheta$ satisfying the condition \eqref{con : investment among n} by $\pazocal{T}(n)$, and the collection of investment strategies among the top $n$ stocks by $\pazocal{I}(S) \cap \pazocal{T}(n)$.
	
	The collection of all num\'eraires generated by investment strategies $\vartheta \in \pazocal{I}(S) \cap \pazocal{T}(n)$ among the top $n$ stocks is denoted by $\pazocal{X}^n$.
\end{defn}

\smallskip

The condition \eqref{con : investment among n} prohibits the strategy $\vartheta$ from investing in the $i$-th stock at time $t \geq 0$, if this stock fails to rank at that time among the top $n$ stocks in terms of capitalization. We present another definition, that of a portfolio rule, which plays the role of integrand with respect to the integrator $dR_i(t)$ of \eqref{def : R}.

\smallskip

\begin{defn} [Portfolio] \label{Def : portfolio}
	We call an $N$-dimensional predictable, vector-valued process $\pi \equiv (\pi_1, \cdots, \pi_N) \in \pazocal{I}(R)$ a \textit{portfolio}, if it is integrable with respect to the cumulative return vector $R$ of \eqref{def : R}. We call a portfolio $\pi \in \pazocal{I}(R)$ a \textit{portfolio among the top $n$ stocks}, if the equalities 
	\begin{equation}		\label{con : portfolio among n}
		\pi_i(t)\bm{1}_{\{u_i(t) > n\}} = 0, \quad \text{for} \quad i = 1, \cdots, N, \quad t \geq 0,
	\end{equation}
	hold with the notation \eqref{def : u}. We denote the collection of portfolios among the top $n$ stocks by $\pazocal{I}(R) \cap \pazocal{T}(n)$.
\end{defn}

\smallskip 
 
Since the function $u_{\cdot}(t) : \{1, \cdots, N\} \rightarrow \{1, \cdots, N\}$ is bijective for every $t \geq 0$, the collection $\big\{\bm{1}_{\{u_i(t) = k\}} \big\}_{k=1, \cdots, N}$ constitutes a partition of unity for any given $i=1, \cdots, N$, $t \geq 0$, and the conditions \eqref{con : investment among n}, \eqref{con : portfolio among n} can also be formulated respectively as
\begin{equation}		\label{con : equiv investment among n}
\vartheta_i(t)
= \sum_{k=1}^n \vartheta_i(t)\bm{1}_{\{u_i(t) = k\}}
= \vartheta_i(t)\bm{1}_{\{u_i(t) \leq n\}}, \quad \text{for} \quad i = 1, \cdots, N, \quad t \geq 0,
\end{equation}
\begin{equation}		\label{con : equiv portfolio among n}
	\pi_i(t)
	= \sum_{k=1}^n \pi_i(t)\bm{1}_{\{u_i(t) = k\}}
	= \pi_i(t)\bm{1}_{\{u_i(t) \leq n\}}, \quad \text{for} \quad i = 1, \cdots, N, \quad t \geq 0.
\end{equation}

\bigskip

We present next the connection between investment strategies $\vartheta$ and portfolios $\pi$. For any scalar continuous semimartingale $Z$ with $Z(0) = 0$, we denote the stochastic exponential of $Z$ by
\begin{equation}		\label{def : stochastic exponential}
	\pazocal{E}(Z) := \exp \Big( Z - \frac{1}{2}[Z, Z] \Big).
\end{equation}
It can be shown that this is also the unique process satisfying the linear stochastic integral equation
\begin{equation}		\label{Eq : exponential SDE}
	\pazocal{E}(Z) = 1 + \int_0^{\cdot} \pazocal{E}(Z)(t) dZ(t).
\end{equation}
It is straightforward to check that the stochastic logarithm operator $\pazocal{L}(\cdot)$ in \eqref{def : stochastic logarithm}, is the inverse of the stochastic exponential operator $\pazocal{E}(\cdot)$ in \eqref{def : stochastic exponential}.

\smallskip

We introduce now the \textit{cumulative returns process of a portfolio} $\pi$ as in Definition~\ref{Def : portfolio}, via the vector stochastic integral
\begin{equation}		\label{def : rpi}
	R_{\pi} := \int_0^{\cdot} \pi'(t) dR(t) = \int_0^{\cdot} \sum_{i=1}^N \pi_i(t) dR_i(t),
\end{equation}
and consider its stochastic exponential
\begin{equation}		\label{def : wealth}
	X_{\pi} := \pazocal{E}(R_{\pi}) = \pazocal{E}\Big( \int_0^{\cdot} \sum_{i=1}^N \pi_i(t) dR_i(t) \Big).
\end{equation}
In particular, we note that $X_{\pi}$ is positive. Then, from \eqref{def : wealth}, \eqref{def : rpi}, and \eqref{Eq : tilde R dynamics}, we obtain the dynamics
\begin{equation}		\label{Eq : dynamics of wealth}
	\frac{dX_{\pi}(t)}{X_{\pi}(t)} = dR_{\pi}(t) = \sum_{i=1}^N \pi_i(t) dR_i(t) = \sum_{i=1}^N \pi_i(t) \frac{dS_i(t)}{S_i(t)}, \quad X_{\pi}(0) = 1.
\end{equation}
By setting
\begin{equation}	\label{Eq : relationship}
	\vartheta_i := \frac{X_{\pi}\pi_i}{S_i} \quad \text{for} \quad i = 1, \cdots, N,
\end{equation}
we arrive at the equation \eqref{def : investment} with $X(\cdot; 1, \vartheta)$ replaced by $X_{\pi}(\cdot)$. Thus, from the portfolio $\pi \in \pazocal{I}(R)$, we can obtain the corresponding investment strategy $\vartheta$ and its num\'eraire $X(\cdot; 1, \vartheta)$, via the recipe \eqref{Eq : relationship}. Here, we denote the num\'eraire generated by the portfolio $\pi$ by $X(\cdot; 1, \vartheta) := X_{\pi}$, as in \eqref{def : wealth}.

\bigskip

Conversely, for a given investment strategy $\vartheta$ generating a positive wealth process, i.e., the num\'eraire $X(\cdot; 1, \vartheta)$, we define a predictable, vector-valued process $\pi \equiv (\pi_1, \cdots, \pi_N)$ as
\begin{equation}		\label{Eq : relationship2}
	\pi_i := \frac{S_i\vartheta_i}{X(\cdot; 1, \vartheta)} \quad \text{for} \quad i = 1, \cdots, N.
\end{equation}
It can be easily checked that $\pi$ is indeed a portfolio, i.e., $R$-integrable and \eqref{def : investment} can be written as
\begin{equation*}
	X(\cdot; 1, \vartheta) = 1 + \int_0^{\cdot} X(t; 1, \vartheta) \sum_{i=1}^N \pi_i(t)dR_i(t),
\end{equation*}
with the help of \eqref{Eq : tilde R dynamics}. This last equation gives the dynamics in \eqref{Eq : dynamics of wealth} with $X_{\pi}(\cdot) \equiv X(\cdot; 1, \vartheta)$.

\bigskip

Thus, whether we start from an investment strategy $\vartheta$ (generating a num\'eraire) or from a portfolio $\pi$, the counterpart can always be obtained via \eqref{Eq : relationship2} or \eqref{Eq : relationship}, respectively, and we will denote the corresponding num\'eraire $X(\cdot;, 1, \vartheta)$ in \eqref{def : investment} by $X_\pi$ as in \eqref{def : wealth}.

\smallskip

In the relationship \eqref{Eq : relationship2}, the product $S_i(t)\vartheta_i(t)$ represents the amount of wealth invested in $i$-th stock at time $t$, thus $\pi_i(t)$ can be interpreted as the proportion of current wealth invested in $i$-th stock at time $t$. The remaining proportion of wealth
\begin{equation}		\label{def : money proportion}
	\pi_0 := 1 - \sum_{i=1}^N \pi_i
\end{equation}
is then considered to be placed in the money market.

\bigskip

We present now a few more concepts regarding portfolios. For any two portfolios $\pi, \rho$ in $\pazocal{I}(R)$, we consider the covariation process between the cumulative returns $R_{\pi}, R_{\rho}$ in \eqref{def : rpi}, namely
\begin{equation}	\label{def : C pi rho}
	C_{\pi\rho} := [R_{\pi}, R_{\rho}] = \int_0^{\cdot} c_{\pi\rho}(t)dO(t), \quad \text{with} \quad c_{\pi\rho} := \pi'c\rho = \sum_{i=1}^N \sum_{j=1}^N \pi_i c_{i, j}\rho_{j}.
\end{equation}
Here, we recall the definitions of the matrix-valued processes $c$ and $C$ in \eqref{def : alpha and c}, \eqref{def : A and C} and note the notational consistency with \eqref{def : C pi rho}. In particular, when the portfolio is given as the unit vector $e^i$ of $\mathbb{R}^N$ for some $i = 1, \cdots, N$, we use the subscript `$i$' instead of `$e^i$' to write $C_{i\rho} \equiv C_{e^i\rho}$ and $c_{i\rho} \equiv c_{e^i\rho}$ in order to ease notation. This convention is consistent with the actual equalities $C_{i, j} = C_{e^ie^j}$ and $c_{i, j} = c_{e^ie^j}$ for $1 \leq i, j \leq N$.

\smallskip

By recalling the wealth process $X_{\pi}$ generated by the portfolio $\pi$ as in \eqref{def : wealth}, \eqref{Eq : dynamics of wealth}, we can express the logarithm of $X_{\pi}$ as
\begin{equation}		\label{def : log of wealth}
	\log X_{\pi} = R_{\pi} - \frac{1}{2}C_{\pi\pi}
	= \int_0^{\cdot} \pi'(t)dA(t) - \frac{1}{2}C_{\pi\pi} + \int_0^{\cdot} \pi'(t)dM(t).
\end{equation}
We call the finite-variation part of $\log X_{\pi}$ the \textit{cumulative growth of the portfolio} $\pi$, and denote it by
\begin{equation}		\label{def : Gamma pi}
	\Gamma_{\pi} := A_{\pi} - \frac{1}{2}C_{\pi\pi}, \quad \text{where} \quad A_{\pi} := \int_0^{\cdot} \pi'(t)dA(t).
\end{equation}
In a similar manner, the local martingale part of the decomposition in \eqref{def : log of wealth} is denoted by
\begin{equation}		\label{def : M pi}
	M_{\pi} := \int_0^{\cdot} \pi'(t)dM(t).
\end{equation}
In particular, the cumulative return $R_{\pi}$ in \eqref{def : wealth} is the stochastic logarithm $\pazocal{L}(X_{\pi})$ of $X_{\pi}$, and has `drift' component $A_{\pi}$ as in \eqref{def : Gamma pi}, from \eqref{def : rpi} and \eqref{Eq : decomposition of R}; whereas the natural logarithm $\log X_{\pi}$ in \eqref{def : log of wealth} of $X_{\pi}$ has `drift' term $\Gamma_{\pi}$. 

\smallskip

We further define the predictable processes
\begin{equation}		\label{def : alpha pi, gamma pi}
	\alpha_{\pi} := \pi'\alpha, \qquad
	\gamma_{\pi} := \pi'\alpha-\frac{1}{2}\pi' c \pi = \alpha_{\pi} - \frac{1}{2} c_{\pi\pi},
\end{equation}
and call $\alpha_{\pi}$ the \textit{rate of return}, and $\gamma_{\pi}$ the \textit{growth rate}, of the portfolio $\pi$. The `drift parts' $A_{\pi}$ and $\Gamma_{\pi}$, of $\pazocal{L}(X_{\pi})$ and $\log X_{\pi}$, respectively, are then represented as the integrals of these rates with respect to the `operational clock' in \eqref{def : O}:
\begin{equation}		\label{def : A pi, Gamma pi}
	A_{\pi} = \int_0^{\cdot} \alpha_{\pi}(t)dO(t), \qquad \Gamma_{\pi} = \int_0^{\cdot} \gamma_{\pi}(t)dO(t).
\end{equation}

\bigskip

\subsection{Portfolios among the top \textit{n} stocks}		\label{subsec : portfolio among the top n stocks}

In this subsection we provide definitions, similar to those introduced in the previous subsections, for portfolios that invest only among the top $n$ stocks.

\bigskip
 
For $\vartheta \in \pazocal{I}(S) \cap \pazocal{T}(n)$ and $\pi \in \pazocal{I}(R) \cap \pazocal{T}(n)$, representing a strategy that invests only among the top $n$ stocks and a portfolio among the top $n$ stocks, respectively, the equations \eqref{def : rpi}-\eqref{def : money proportion} can be used in the same manner. In particular, the bidirectional connections \eqref{Eq : relationship} and \eqref{Eq : relationship2} between $\vartheta \in \pazocal{I}(S) \cap \pazocal{T}(n)$ and $\pi \in \pazocal{I}(R) \cap \pazocal{T}(n)$ still hold, because of the similarity in the conditions \eqref{con : investment among n} and \eqref{con : portfolio among n}.

\bigskip

We define next a new $N$-dimensional vector $\widetilde{R} \equiv (\widetilde{R}_1, \cdots, \widetilde{R}_N)$ by
\begin{equation}			\label{def : tilde R}
\widetilde{R}_i(t) := \int_0^t \bm{1}_{\{u_i(s) \leq n\}} dR_i(s), \quad \text{for} \quad i = 1, \cdots, N, \quad t \geq 0.
\end{equation}
Each component $\widetilde{R}_i(t)$ represents the cumulative return of the $i$-th stock, accumulated over $[0, t]$ but only at times when the stock ranks among the top $n$ by capitalization. Then, for $\pi \in \pazocal{I}(R) \cap \pazocal{T}(n)$, \eqref{def : rpi} can be also cast as
\begin{equation}	\label{Eq : R pi in tilde}
	R_{\pi}
	= \int_0^{\cdot} \sum_{i=1}^N \pi_i(t) dR_i(t)
	= \int_0^{\cdot} \sum_{i=1}^N \pi_i(t) \bm{1}_{\{u_i(t) \leq n\}} dR_i(t)
	= \int_0^{\cdot} \sum_{i=1}^N \pi_i(t) d\widetilde{R}_i(t),
\end{equation}
where the second equality follows from \eqref{con : equiv portfolio among n}. We then have the semimartingale decomposition
\begin{equation}	\label{Eq : decomposition of tilde R}
	\widetilde{R}_i = \widetilde{A}_i + \widetilde{M}_i, \qquad i = 1, \cdots, N,
\end{equation}
where
\begin{equation}		\label{def : tilde A and M}
\widetilde{A}_i(t) := \int_0^t \bm{1}_{\{u_i(s) \leq n\}} dA_i(s), \qquad
\widetilde{M}_i(t) := \int_0^t \bm{1}_{\{u_i(s) \leq n\}} dM_i(s), \qquad i = 1, \cdots, N,
\end{equation}
from \eqref{Eq : decomposition of R}. In the decomposition $R_{\pi} = A_{\pi} + M_{\pi}$, with $A_{\pi}$ as in \eqref{def : Gamma pi} and $M_{\pi}$ as in \eqref{def : M pi}, we note that $A_{\pi}$ and $M_{\pi}$ can be expressed in terms of the components of $\widetilde{A}$ and $\widetilde{M}$, respectively, as
\begin{equation}	\label{Eq : M pi in tilde}
	A_{\pi}
	= \int_0^{\cdot} \sum_{i=1}^N \pi_i(t) dA_i(t)
	= \int_0^{\cdot} \sum_{i=1}^N \pi_i(t) d\widetilde{A}_i(t), \quad
	M_{\pi}
	= \int_0^{\cdot} \sum_{i=1}^N \pi_i(t) dM_i(t)
	= \int_0^{\cdot} \sum_{i=1}^N \pi_i(t) d\widetilde{M}_i(t)
\end{equation}
by analogy with \eqref{Eq : R pi in tilde}. Also in a manner similar to \eqref{def : C}, we define
\begin{equation}		\label{def : tilde C}
\widetilde{C}_{i, j} := [\widetilde{M}_i, \widetilde{M}_j] = [\widetilde{R}_i, \widetilde{R}_j], \quad \text{for} \quad 1 \leq i, j \leq N.
\end{equation}
Note the relationship
\begin{equation}		\label{Eq : tilde C and C}
d\widetilde{C}_{i, j}(t)
= d[\widetilde{R}_i, \widetilde{R}_j](t)
= \bm{1}_{\{u_i(t) \leq n\}}\bm{1}_{\{u_j(t) \leq n\}} d[R_i, R_j](t)
= \bm{1}_{\{u_i(t) \leq n\}}\bm{1}_{\{u_j(t) \leq n\}} dC_{i, j}(t)
\end{equation}
between $\widetilde{C}$ and $C$. We further define  a vector-valued process $\widetilde{\alpha} \equiv (\widetilde{\alpha}_1, \cdots, \widetilde{\alpha}_N)$ and a matrix-valued process $\widetilde{c} \equiv (\widetilde{c}_{i, j})_{1 \leq i, j \leq N}$ as
\begin{equation}		\label{def : tilde alpha}
	\widetilde{\alpha}_i(t) := \bm{1}_{\{u_i(t) \leq n\}}\alpha_i(t), \quad \qquad \qquad \qquad i = 1, \cdots, N,
\end{equation}
\begin{equation}		\label{def : tilde c}
	\widetilde{c}_{i, j}(t) := \bm{1}_{\{u_i(t) \leq n\}}\bm{1}_{\{u_j(t) \leq n\}}c_{i, j}(t), \qquad 1 \leq i, j \leq N;
\end{equation}
then it is straightforward to obtain the relationships
\begin{equation}		\label{Eq : tilde A and tilde C}
	\widetilde{A} = \int_0^{\cdot} \widetilde{\alpha}(t)dO(t), \quad \text{and} \quad \widetilde{C} \equiv [\widetilde{M}, \widetilde{M}] = \int_0^{\cdot} \widetilde{c}(t)dO(t)
\end{equation}
in accordance with \eqref{def : A and C}, where the vector-valued and matrix-valued processes $\widetilde{A} \equiv (\widetilde{A}_1, \cdots, \widetilde{A}_N)$ and $\widetilde{C} \equiv (\widetilde{C}_{i, j})_{1 \leq i, j \leq N}$, respectively, are as in \eqref{def : tilde A and M}, \eqref{def : tilde C}.

\bigskip

The definition of $C_{\pi\rho}$ in \eqref{def : C pi rho} can be also invoked when $\pi, \rho \in \pazocal{I}(R) \cap \pazocal{T}(n)$, but we have also with the help of \eqref{Eq : R pi in tilde} and \eqref{def : tilde C} the alternative representation
\begin{equation}		\label{Eq : C pi rho}
	C_{\pi\rho} = [R_{\pi}, R_{\rho}]
	= \bigg[\int_0^{\cdot} \sum_{i=1}^N \pi_i(t) d\widetilde{R}_i(t), \int_0^{\cdot} \sum_{j=1}^N \rho_j(t) d\widetilde{R}_j(t)\bigg]
	= \int_0^{\cdot} \sum_{i=1}^N \sum_{j=1}^N \pi_i(t) \rho_j(t) d\widetilde{C}_{i, j}(t).
\end{equation}

\smallskip

In particular, consider the portfolio $\pi$ among the top $n$ stocks, defined as
\begin{equation}		\label{def : tilde unit portfolio}
	\pi(\cdot) := \bm{1}_{\{u_i(\cdot) \leq n\}} e^i
\end{equation}
for a fixed $i = 1, \cdots, N$. This portfolio $\pi$ invests all wealth in the $i$-th stock, when this stock ranks among the top $n$; otherwise, it puts all wealth in the money market. From \eqref{Eq : R pi in tilde}, the identity
\begin{equation}			\label{Eq : R unit vector}
	R_{\pi} = \widetilde{R}_i	
\end{equation}
holds, and we shall use the subscript `$\widetilde{i}$' instead of `$\pi$' to write $X_{\pi} \equiv X_{\widetilde{i}}$ and 
\begin{equation}			\label{Eq : C tilde i rho}
	C\,_{\widetilde{i}\rho} \equiv C_{\pi\rho} = \int_0^{\cdot} \sum_{j=1}^N \rho_j(t) d\widetilde{C}_{i, j}(t), \qquad \text{as well as} \qquad c\,_{\widetilde{i}\rho} \equiv c_{\pi\rho} = \sum_{j=1}^N \rho_j(t) \widetilde{c}_{i, j}(t),
\end{equation}
in order to ease notation for the specific $\pi$ in \eqref{def : tilde unit portfolio}. This convention is consistent with the equalities $C\,_{\widetilde{i}, \widetilde{j}} = [\widetilde{R}_i, \widetilde{R}_j] = \widetilde{C}_{i, j}$ for $1 \leq i, j \leq N$.

\bigskip

It is useful to write succinctly the above relationships in this subsection, between symbols with tilde and corresponding symbols without tilde, in matrix notation. We do this by introducing the predictable matrix-valued process $D \equiv (D_{i, j})_{1 \leq i, j \leq N}$ with entries
\begin{equation}		\label{def : diagonal matrix D}
D_{i, j}(t) := 
\begin{cases}
\bm{1}_{\{u_i(t) \leq n\}} \qquad & i = j,
\\
~~~~~0 & i \ne j,
\end{cases}	
\end{equation}
for each $t \geq 0$. Here, we note that $D(t)$ is a diagonal, idempotent matrix, whose $(i, i)$-th entry is $1$ if the $i$-th stock belongs to the top $n$ stocks at time $t \geq 0$, otherwise it is zero. Because at least $N-n$ diagonal entries of $D(t)$ are zero, $D(\cdot)$ is always singular. Then, any $N$-dimensional predictable process $\nu$ in $\pazocal{T}(n)$ as in Definition~\ref{Def : investment among n}, satisfies $D\nu = \nu$; in particular,
\begin{equation}			\label{Eq : vartheta, pi with D}
	D\vartheta = \vartheta, \qquad D\pi = \pi,
\end{equation}
hold for all $\vartheta \in \pazocal{I}(S) \cap \pazocal{T}(n)$ and $\pi \in \pazocal{I}(R) \cap \pazocal{T}(n)$ from the conditions \eqref{con : equiv investment among n} and \eqref{con : equiv portfolio among n}. Also, the identities \eqref{def : tilde R}, \eqref{def : tilde A and M}, \eqref{Eq : tilde C and C}, \eqref{def : tilde alpha}, and \eqref{def : tilde c} can be reformulated as
\begin{equation}		\label{Eq : tilde R with D}
	d\widetilde{R}(t) = D(t)dR(t), \quad d\widetilde{A}(t) = D(t)dA(t), \quad d\widetilde{M}(t) = D(t)dM(t),
\end{equation}
\begin{equation*}
	d\widetilde{C}(t) = D(t) dC(t) D(t),
\end{equation*}
as well as
\begin{equation}		\label{Eq : tilde alpha, c with D}
	\widetilde{\alpha} = D\alpha, \quad \widetilde{c} = DcD.
\end{equation}
Moreover, we have another expression of the type \eqref{def : alpha pi, gamma pi} for $\pi \in \pazocal{I}(R) \cap \pazocal{T}(n)$: using the property \eqref{Eq : vartheta, pi with D}, we write
\begin{equation}		\label{Eq : alpha pi, gamma pi}
	\alpha_{\pi} = \pi'\alpha = \pi'D\alpha = \pi'\widetilde{\alpha}, \qquad 
	\gamma_{\pi} = \pi'\alpha-\frac{1}{2}\pi' c \pi = \pi'\widetilde{\alpha}-\frac{1}{2}\pi'DcD\pi = \pi'\widetilde{\alpha}-\frac{1}{2}\pi'\widetilde{c}\pi.
\end{equation}

\bigskip

We present now the following results regarding the integrability condition with respect to $R$ (or $\widetilde{R}$), which will be used in the next section.

\smallskip

\begin{lem}	[Null portfolio]	\label{lem : null portfolio}
	For an $N$-dimensional predictable process $\eta \in \pazocal{T}(n)$, suppose that $\eta'\widetilde{\alpha} = 0$ and $\widetilde{c}\eta = 0$ hold in the $(\mathbb{P}\otimes O)$-a.e. sense.
	
	Then $\eta$ is a portfolio, i.e., $\eta \in \pazocal{I}(R) \cap \pazocal{T}(n)$, and the identity $R_{\eta} = \int_0^{\cdot} \eta'(t) d\widetilde{R}(t) \equiv 0$ holds. In this case, we call $\eta$ a \textit{null portfolio}.
\end{lem}

\smallskip

\begin{proof}
	As $\eta \in \pazocal{T}(n)$, we have $D\eta = \eta$, or $\eta' = \eta'D$. Recalling \eqref{Eq : tilde R with D}, \eqref{Eq : decomposition of tilde R}, and \eqref{Eq : tilde A and tilde C}, we have
	\begin{align}
		\int_0^{\cdot} \eta'(t) dR(t)
		&=\int_0^{\cdot} \eta'(t) D(t) dR(t)
		= \int_0^{\cdot} \eta'(t) d\widetilde{R}(t)		\nonumber
		\\
		&= \int_0^{\cdot} \eta'(t) d\widetilde{A}(t) + \int_0^{\cdot} \eta'(t) d\widetilde{M}(t)
		= \int_0^{\cdot} \eta'(t) \widetilde{\alpha}(t) dO(t) + \int_0^{\cdot} \eta'(t) d\widetilde{M}(t).			\label{Eq : null portfolio}
	\end{align}
	The first integral on the right-hand side of \eqref{Eq : null portfolio} vanishes, thanks to the assumption $\eta'\widetilde{\alpha} = 0$. The second integral $\int_0^{\cdot} \eta'(t) d\widetilde{M}(t)$ is a continuous local martingale, and has quadratic variation $\int_0^{\cdot} \eta'(t) \widetilde{c}(t) \eta(t) dO(t)$ from \eqref{Eq : tilde A and tilde C}. This quadratic variation also vanishes on account of the assumption $\widetilde{c}\eta = 0$, and the result follows.
\end{proof}

\bigskip

\begin{lem} [Integrability condition with respect to $R$]		\label{lem : integrability condition}
	An $N$-dimensional predictable vector process $\pi \in \pazocal{T}(n)$ belongs to $\pazocal{I}(R)$, if and only if
	\begin{equation}		\label{con : integrable w.r.t tilde R}
		\int_0^T \Big( |\pi'(t)\widetilde{\alpha}(t)| + \pi'(t)\widetilde{c}(t)\pi(t)\Big) dO(t) < \infty, \quad \text{for any} \quad T \geq 0.
	\end{equation}
\end{lem}

\smallskip

\begin{proof}
	From the assumption $\pi \in \pazocal{T}(n)$, we have $D\pi = \pi$, and $\pi'D = \pi'$. The condition \eqref{con : integrable w.r.t R} can be rewritten with the help of \eqref{Eq : tilde alpha, c with D} as
	\begin{align*}
		&~~~\int_0^T \Big( |\pi'(t)\alpha(t)| + \pi'(t)c(t)\pi(t)\Big) dO(t)
		=\int_0^T \Big( |\pi'(t)D(t)\alpha(t)| + \pi'(t)D(t)c(t)D(t)\pi(t)\Big) dO(t)
		\\
		&=\int_0^T \Big( |\pi'(t)\widetilde{\alpha}(t)| + \pi'(t)\widetilde{c}(t)\pi(t)\Big) dO(t)
		< \infty.
	\end{align*}
\end{proof}

\bigskip

\section{Num\'eraires and Market Viability}
\label{sec: main}

This section presents the fundamental result in arbitrage theory of equity market, in open market context. Before we state and prove the result, we explain several necessary concepts one after another.

\subsection{Auxiliary market}
Consider a portfolio $\rho \in \pazocal{I}(R)$ which generates the num\'eraire $X_{\rho}$ as in Definition~\ref{Def : portfolio} and \eqref{def : wealth}, and fix $\rho$ throughout this subsection. We regard this portfolio $\rho$ as a `baseline', in the sense that want to compare the relative performance of any other portfolio $\pi \in \pazocal{I}(R)$ with respect to $\rho$, by understanding the relative wealth process
\begin{equation}	\label{def : relative wealth}
X^{\rho}_{\pi} := \frac{X_{\pi}}{X_{\rho}}.
\end{equation} 
As the wealth $X_{\pi}$ is denominated relative to $X_{\rho}$ in \eqref{def : relative wealth}, we consider an \textit{auxiliary market}, in which all the components of the price vector $S$ in Definition~\ref{Def : price process by name} are denominated in units of $X_{\rho}$:
\begin{equation}		\label{def : denominated stock price}
	S^{\rho}_i := \frac{S_i}{X_{\rho}}, \quad i = 1, \cdots, N.
\end{equation}
Here, we also consider the money market $S_0 \equiv 1$, with $S^{\rho}_0 := 1/X_{\rho}$, as we assume that the money market pays and charges zero interest in the introductory part of Section~\ref{sec: setup}. Since $S^{\rho}_0$ is no longer trivial, we will consider the $(N+1)$-dimensional vector $S^{\rho} \equiv (S^{\rho}_0, S^{\rho}_1, \cdots S^{\rho}_N)$ as the price process vector in this auxiliary market.

\bigskip

Recalling the notation \eqref{def : rpi} and \eqref{def : C pi rho}, we define two $(N+1)$-dimensional vectors of semimartingales $R^{\rho} \equiv (R^{\rho}_0, \cdots, R^{\rho}_N)$, and $\widetilde{R}^{\rho} \equiv (\widetilde{R}^{\rho}_0, \cdots, \widetilde{R}^{\rho}_N)$ with components
\begin{equation}		\label{def : cumulative return in auxiliary market}
	R^{\rho}_0 := C_{\rho\rho}-R_{\rho}, \quad \text{and} \quad R^{\rho}_i := R^{\rho}_0 + (R_i - C_{i\rho}), \quad \text{for} \quad i = 1, \cdots, N.
\end{equation}
\begin{equation}		\label{def : tilde cumulative return in auxiliary market}
	\widetilde{R}^{\rho}_0 := R^{\rho}_0 = C_{\rho\rho}-R_{\rho}, \quad \text{and} \quad \widetilde{R}^{\rho}_i := \widetilde{R}^{\rho}_0 + (\widetilde{R}_i - C_{\widetilde{i}\rho}), \quad \text{for} \quad i = 1, \cdots, N.
\end{equation}
The following result, which builds on Proposition~1.29 of \cite{KK2}, shows that the vectors $R^{\rho}$, $\widetilde{R}^{\rho}$ play the role of \textit{cumulative returns} in the auxiliary market. We also recall the `money market proportion' $\pi_0$ of a portfolio $\pi$ in \eqref{def : money proportion}.

\smallskip

\begin{prop}		\label{prop : relative wealth}
	For any two portfolios $\rho, \pi \in \pazocal{I}(R)$, the relative wealth process $X^{\rho}_{\pi}$ of \eqref{def : relative wealth} admits the representation
	\begin{equation}		\label{def : relative cumulative return}
	X^{\rho}_{\pi} = \pazocal{E}(R^{\rho}_{\pi}), \quad \text{where} \quad R^{\rho}_{\pi} := R_{\pi-\rho} - C_{\pi-\rho, \rho} = \int_0^{\cdot} \sum_{i=0}^N \pi_i(t)dR^{\rho}_i(t).
	\end{equation}
	In particular, for any two portfolios $\rho, \pi \in \pazocal{I}(R) \cap \pazocal{T}(n)$ among the top $n$ stocks, the process $R^{\rho}_{\pi}$ in \eqref{def : relative cumulative return} admits the additional representation
	\begin{equation}		\label{Eq : R rho pi}
		R^{\rho}_{\pi} = \int_0^{\cdot} \sum_{i=0}^N \pi_i(t)d\widetilde{R}^{\rho}_i(t).
	\end{equation}
\end{prop}

\smallskip

\begin{proof}
	The first part is exactly Proposition 1.29 of \cite{KK2}. Thus, it is enough to show $\int_0^{\cdot} \sum_{i=0}^N \pi_i(t)dR^{\rho}_i(t) = \int_0^{\cdot} \sum_{i=0}^N \pi_i(t)d\widetilde{R}^{\rho}_i(t)$ when $\rho, \pi$ belong to $\pazocal{I}(R) \cap \pazocal{T}(n)$.
	Since $\widetilde{R}^{\rho}_0 = R^{\rho}_0$ in \eqref{def : cumulative return in auxiliary market} and \eqref{def : tilde cumulative return in auxiliary market}, this reduces to showing
	\begin{equation*}
		\int_0^{\cdot} \sum_{i=1}^N \pi_i(t)d\big(R_i-C_{i\rho}\big)(t) = \int_0^{\cdot} \sum_{i=1}^N \pi_i(t)d\big(\widetilde{R}_i-C\,_{\widetilde{i}\rho}\big)(t).
	\end{equation*}
	Thanks to the condition \eqref{con : equiv portfolio among n} and the definition \eqref{def : tilde R}, this can be easily checked:
	\begin{align*}
		\int_0^{\cdot} \sum_{i=1}^N \pi_i(t)d\big(R_i-C_{i\rho}\big)(t)
		&= \int_0^{\cdot} \sum_{i=1}^N \pi_i(t) \bm{1}_{\{u_i(t) \leq n\}} d\big(R_i-C_{i\rho}\big)(t)
		\\
		& = \int_0^{\cdot} \sum_{i=1}^N \pi_i(t)d\big(\widetilde{R}_i-C\,_{\widetilde{i}\rho}\big)(t)
	\end{align*}
	where, in the last equality, we used the string of identities
	\begin{equation}		\label{Eq : C i rho}
		\bm{1}_{\{u_i(t) \leq n\}}dC_{i\rho}(t) = \bm{1}_{\{u_i(t) \leq n\}}d[R_i, R_{\rho}](t) = d[\widetilde{R}_i, R_{\rho}](t) = dC\,_{\widetilde{i}\rho}.
	\end{equation}
\end{proof}

\smallskip

In the special case $\pi \equiv e^i$, that is, when the portfolio $\pi$ invests all wealth in the $i$-th stock at all times, the relative wealth process $X^{\rho}_{\pi}$ and its stochastic logarithm $R^{\rho}_{\pi}$ in \eqref{def : relative wealth}, \eqref{def : relative cumulative return} become
\begin{equation*}
	X^{\rho}_{\pi} = \frac{S_i}{X_{\rho}} = S^{\rho}_i, \qquad R^{\rho}_{\pi} = R^{\rho}_i,
\end{equation*}
and Proposition~\ref{prop : relative wealth} yields 
\begin{equation}		\label{Eq : S rho}
	S^{\rho}_i = \pazocal{E}(R^{\rho}_i)
\end{equation}
for any given $i = 1, \cdots, N$. Therefore, the component $R^{\rho}_i$ of \eqref{def : cumulative return in auxiliary market} is the stochastic logarithm of the $i$-th component of the price vector $S^{\rho}$ in the auxiliary market, and the vector $R^{\rho}$ plays the role of cumulative returns in the auxiliary market.

\smallskip

By analogy with \eqref{Eq : dynamics of wealth}, we also have
\begin{equation}		\label{Eq : dynamics of relative wealth}
	\frac{dX^{\rho}_{\pi}(t)}{X^{\rho}_{\pi}(t)} = dR^{\rho}_{\pi}(t) = \sum_{i=0}^N \pi_i(t) dR^{\rho}_i(t) = \sum_{i=0}^N \pi_i(t) \frac{dS^{\rho}_i(t)}{S^{\rho}_i(t)}, \quad X^{\rho}_{\pi}(0) = 1,
\end{equation}
for $\rho, \pi$ in $\pazocal{I}(R)$, from \eqref{def : relative cumulative return}, \eqref{Eq : S rho}. It is very important that the summation in \eqref{Eq : dynamics of relative wealth} should include the index $i=0$, as indeed it does, in contrast to the summation in \eqref{Eq : dynamics of wealth}.

\bigskip

\subsection{Supermartingale Num\'eraire and local martingale Num\'eraire}

We introduce now the notions of supermartingale num\'eraire and local martingale num\'eraire.

\begin{defn} [Supermartingale num\'eraire and local martingale num\'eraire]		\label{Def : supermatingale numeraire}
	A given portfolio $\rho \in \pazocal{I}(R)$ is called \textit{supermatingale num\'eraire portfolio (local martingale num\'eraire portfolio) in the whole market}, if the relative wealth process $X^{\rho}_{\pi} = X_{\pi}/X_{\rho}$ of \eqref{def : relative wealth} is a supermartingale (local martingale) for every portfolio $\pi \in \pazocal{I}(R)$ in the market. In this case, the wealth process $X_{\rho}$ is called a \textit{supermartingale num\'eraire (local martingale num\'eraire, respectively) in the whole market}.
	
	\smallskip
	
	Similarly, a given portfolio $\rho \in \pazocal{I}(R) \cap \pazocal{T}(n)$ among the top $n$ stocks is called \textit{supermatingale num\'eraire portfolio (local martingale num\'eraire portfolio) among the top $n$ stocks}, if the relative wealth process $X^{\rho}_{\pi}$ is a supermartingale (local martingale) for every portfolio $\pi \in \pazocal{I}(R) \cap \pazocal{T}(n)$ among the top $n$ stocks. In this case, the wealth process $X_{\rho}$ is called \textit{supermartingale num\'eraire (local martingale num\'eraire, respectively) among the top $n$ stocks}.
\end{defn}

\smallskip

By Fatou's lemma, every nonnegative local martingale is a supermartingale; thus, every local martingale num\'eraire is in particular a supermatingale num\'eraire. We also have the following uniqueness result for supermartingale (local martingale) num\'eraires (respectively, among the top $n$ stocks). 

\smallskip

\begin{lem}
	There is a unique supermartingale (local martingale) num\'eraire portfolio in the entire market (respectively, among the top $n$ stocks).
\end{lem}

\smallskip

\begin{proof}
	Suppose that there are two local martingale (or two supermartingale) num\'eraire portfolios $\rho$ and $\nu$ with the same initial wealth $X_{\rho}(0) = X_{\nu}(0)$. Then, the relative wealth process $X_{\rho}/X_{\nu}$ and its reciprocal $X_{\nu}/X_{\rho}$ are positive supermartingales. From the Doob-Meyer decomposition of semimartingales, it is easy to show that a continuous, positive supermartingale $Y$ is almost everywhere constant, if its reciprocal is also a supermartingale. Thus, $X_{\rho}/X_{\nu} \equiv 1$ almost everywhere, and the two portfolios $\rho$ and $\nu$ generate the same wealth process.
\end{proof}

\bigskip

It can be shown that the supermartingale num\'eraire is actually the local martingale num\'eraire, thus the two num\'eraires are equivalent, in the whole market where no constraint is imposed on portfolios. This is Proposition 2.4 of \cite{KK2}, which we repeat here for the convenience of the reader.

\smallskip

\begin{prop}		\label{prop : general supermartingale numeraire}
	For a portfolio $\rho \in \pazocal{I}(R)$, the following statements are equivalent:
	\begin{itemize}
		\item [$(1)$] $\rho$ is a supermartingale num\'eraire portfolio in the whole market.
		\item [$(2)$] $\rho$ is a local martingale num\'eraire portfolio in the whole market.
		\item [$(3)$] The equality $A_i = C_{i\rho}$ holds for all $i = 1, \cdots, N$.
	\end{itemize}
\end{prop}

\smallskip

The statement \textit{(3)} gives a very simple structural condition, derived from the cumulative return process of the market, which characterizes this equivalence. It is no surprise that the result also holds for the portfolios among the top $n$ stocks; but in this case, the cumulative return process vector $R$ in \eqref{def : R} should be replaced by $\widetilde{R}$ of \eqref{def : tilde R} instead.

\smallskip

\begin{prop}		\label{prop : supermartingale numeraire among the top n}
	For a portfolio $\rho \in \pazocal{I}(R) \cap \pazocal{T}(n)$, the following statements are equivalent:
	\begin{itemize}
		\item [$\widetilde{(1)}$] $\rho$ is a supermartingale num\'eraire portfolio among the top $n$ stocks.
		\item [$\widetilde{(2)}$] $\rho$ is a local martingale num\'eraire portfolio among the top $n$ stocks.
		\item [$\widetilde{(3)}$] The equality $\widetilde{A}_i = C_{\widetilde{i}\rho}$ holds for all $i = 1, \cdots, N$.
	\end{itemize}
\end{prop}

\smallskip

\begin{proof}
	The proof follows the same general outline as that for Proposition 2.4 in \cite{KK2}. We first assume statement \textit{(3)}, which is equivalent to the requirement that $\widetilde{R}_i - C_{\widetilde{i}\rho} = \widetilde{M}_i$ is a local martingale for all $i=1, \cdots, N$ from \eqref{Eq : decomposition of tilde R}. Recalling the notation of \eqref{def : tilde cumulative return in auxiliary market}, \eqref{def : C pi rho} with the identities \eqref{con : equiv portfolio among n}, \eqref{Eq : R pi in tilde} and \eqref{Eq : C i rho}, we obtain that the process
	\begin{align}
		\widetilde{R}^{\rho}_0 &= C_{\rho\rho}-R_{\rho}
		= \int_0^{\cdot} \sum_{i=1}^N \rho_i(t)dC_{i\rho}(t) - \int_0^{\cdot} \sum_{i=1}^N \rho_i(t) d\widetilde{R}_i(t)										\nonumber
		\\
		&= \int_0^{\cdot} \sum_{i=1}^N \rho_i(t) \bm{1}_{\{u_i(t) \leq n\}} dC_{i\rho}(t) - \int_0^{\cdot} \sum_{i=1}^N \rho_i(t) d\widetilde{R}_i(t)			\nonumber
		\\
		&= \int_0^{\cdot} \sum_{i=1}^N \rho_i(t) dC_{\widetilde{i}\rho}(t) - \int_0^{\cdot} \sum_{i=1}^N \rho_i(t) d\widetilde{R}_i(t)
		= - \int_0^{\cdot} \sum_{i=1}^N \rho_i(t) d\widetilde{M}_i(t)		\label{Eq : tilde R rho 0}
	\end{align}
	is then also a local martingale. This in turn implies that all the components $\widetilde{R}^{\rho}_i = \widetilde{R}^{\rho}_0 + (\widetilde{R}_i - C_{\widetilde{i}\rho})$ for $i=1, \cdots, N$ in \eqref{def : tilde cumulative return in auxiliary market} are local martingales as well. Moreover, from Proposition~\ref{prop : relative wealth}, the processes $R^{\rho}_{\pi}$ and $X^{\rho}_{\pi}$ are also local martingales for every portfolio $\pi \in \pazocal{I}(R) \cap \pazocal{T}(n)$ among the top $n$ stocks, so the implication \textit{(3)} $\Rightarrow$ \textit{(2)} has been proved.
	
	\smallskip
	
	Since statement \textit{(2)} trivially implies statement \textit{(1)}, it remains to establish the implication \textit{(1)} $\Rightarrow$ \textit{(3)}. Assuming statement \textit{(1)}, we first fix any $i$ in $\{1, \cdots, N\}$, consider a specific portfolio $\pi$ among the top $n$ stocks defined as in \eqref{def : tilde unit portfolio}, and recall the notation $X_{\pi} \equiv X_{\widetilde{i}}$ as well as $R_{\pi} \equiv \widetilde{R}_i$. Then, the processes
	\begin{equation*}
		X^{\rho}_{\rho+\widetilde{i}} = \frac{X_{\rho+\widetilde{i}}}{X_{\rho}}, \qquad X^{\rho}_{\rho-\widetilde{i}} = \frac{X_{\rho-\widetilde{i}}}{X_{\rho}},
	\end{equation*}
	are supermartingales. In view of Proposition~\ref{prop : relative wealth} along with \eqref{Eq : R unit vector}, all processes
	\begin{equation*}
		\pazocal{L}(X^{\rho}_{\rho+\widetilde{i}}) = R^{\rho}_{\rho+\widetilde{i}} = \widetilde{R}_{i}-C_{\widetilde{i}\rho}, \qquad 
		\pazocal{L}(X^{\rho}_{\rho-\widetilde{i}}) = R^{\rho}_{\rho-\widetilde{i}} = -(\widetilde{R}_{i}-C_{\widetilde{i}\rho}),
	\end{equation*}
	are local supermartingales, implying that $\widetilde{R}_{i}-C_{\widetilde{i}\rho}$ is a local martingale. Since $i \in \{1, \cdots, N\}$ can be chosen arbitrarily, we arrive at statement \textit{(3)}.
\end{proof}

\bigskip

\begin{rem} [Representation of wealth relative to the supermartingale num\'eraire]
	When $\rho \in \pazocal{I}(R) \cap \pazocal{T}(n)$ is a supermartingale num\'eraire portfolio among the top $n$ stocks, statement \textit{(3)} of Proposition~\ref{prop : supermartingale numeraire among the top n} implies that $\widetilde{R}_i - C_{\widetilde{i}\rho} = \widetilde{M}_i$ is a local martingale for all $i=1, \cdots, N$. Then, Proposition~\ref{prop : relative wealth} with the notation \eqref{def : tilde cumulative return in auxiliary market} yields the following representation of the relative wealth process $X^{\rho}_{\pi}$ for any portfolio $\pi \in \pazocal{I}(R) \cap \pazocal{T}(n)$ among the top $n$ stocks, namely,
	\begin{align*}
		X^{\rho}_{\pi} 
		&= \pazocal{E}\bigg(\int_0^{\cdot} \sum_{i=0}^N \pi_i(t)d\widetilde{R}^{\rho}_i(t)\bigg)
		= \pazocal{E}\bigg(\widetilde{R}^{\rho}_0 + \int_0^{\cdot} \sum_{i=1}^N \pi_i(t)d\widetilde{M}_i(t)\bigg)
		\\
		& = \pazocal{E}\bigg(\int_0^{\cdot} \sum_{i=1}^N \big(\pi_i(t)-\rho_i(t)\big) d\widetilde{M}_i(t)\bigg)
		= 1 + \int_0^{\cdot} X^{\rho}_{\pi}(t)\sum_{i=1}^N \big(\pi_i(t)-\rho_i(t)\big) d\widetilde{M}_i(t),
	\end{align*}
	where the second-last equality is from \eqref{Eq : tilde R rho 0}. Thus, the relative wealth process $X^{\rho}_{\pi}$ is a stochastic integral with respect to the local martingale vector $\widetilde{M}$, defined in \eqref{def : tilde A and M}.
\end{rem}

\bigskip

\begin{rem} [Equivalent condition of statement \textit{(3)}]	\label{rem : equiv statement (3)}
	The statement \textit{(3)} of Proposition~\ref{prop : supermartingale numeraire among the top n} can be reformulated using the `rate processes' $\widetilde{\alpha}$, $\widetilde{c}$ of \eqref{Eq : tilde A and tilde C}, namely,
	\begin{equation*}
		\int_0^{\cdot} \widetilde{\alpha}_i(t) dO(t)
		= \widetilde{A}_i 
		= C_{\widetilde{i}\rho}
		= \int_0^{\cdot} \sum_{j=1}^N \rho_j(t) d\widetilde{C}_{i, j}(t)
		= \int_0^{\cdot} \sum_{j=1}^N \rho_j(t) \widetilde{c}_{i, j}(t) dO(t),
	\end{equation*}
	with the help of \eqref{Eq : C tilde i rho}. Thus, we have the following statement $\widetilde{(3)}$ of Proposition~\ref{prop : supermartingale numeraire among the top n}, namely
	\begin{flalign}		\label{Eq : equiv statement (3)}
		\widetilde{(3)}' \qquad \widetilde{\alpha} = \widetilde{c}\rho, \qquad (\mathbb{P}\otimes O) -\text{a.e.} &&
	\end{flalign}
	in matrix notation. In the same manner, the statement $(3)$ of Proposition~\ref{prop : general supermartingale numeraire} also has the equivalent formulation:
	\begin{flalign}		\label{Eq : general equiv statement (3)}
		(3)' \qquad \alpha = c\rho, \qquad (\mathbb{P}\otimes O) -\text{a.e.} &&
	\end{flalign}
\end{rem}

\bigskip

\subsection{Structural conditions}

In this subsection, we present another equivalent requirement for statement \textit{(3)} of Proposition~\ref{prop : supermartingale numeraire among the top n}, in the form of what we call `structural conditions'. First, we note that $\widetilde{c}$ of \eqref{Eq : tilde alpha, c with D} is a singular symmetric matrix, thus not invertible, from the fact that $D$ is singular. Before proceeding to the next result, we need the following definition of `pseudo-inverse' for the matrix-valued process $\widetilde{c}$ of \eqref{def : tilde c}:
\begin{equation}		\label{def : tilde c dagger}
	\widetilde{c}\,^{\dagger} := \lim_{m \rightarrow \infty} \Big( \big( \widetilde{c}+ \frac{1}{m}I \big)^{-2}~ \widetilde{c} \Big),
\end{equation}
where $I$ is the identity operator on $\mathbb{R}^N$. This process $\widetilde{c}\,^{\dagger}$ will play the role of `pseudo-inverse' for $\widetilde{c}$, because it is easily checked that
\begin{enumerate}[(a)]
	\item $\widetilde{c}\,^{\dagger}$ is the inverse of $\widetilde{c}$ when restricted on $\textbf{range}(\widetilde{c})$,
	\item $\widetilde{c}\,\widetilde{c}\,^{\dagger}$ coincides with the projection operator of $\mathbb{R}^N$ onto $\textbf{range}(\widetilde{c})$,
	\item $\widetilde{c}\,^{\dagger}$ is predictable, since matrix inversion is a continuous operation when restricted to strictly positive-definite matrices.
\end{enumerate}
We are now ready to present the structural conditions.

\smallskip

\begin{prop}		\label{prop : structural condition}
	The existence of the supermartingale num\'eraire portfolio among the top $n$ stocks, is equivalent to the conjunction of the conditions:
	\begin{align}
		(i) ~~& ~\widetilde{\alpha} \in \textbf{range}(\widetilde{c}), \qquad \qquad \qquad \quad (\mathbb{P} \otimes O)-a.e.,				\label{con : structural1}
		\\
		(ii) ~~&\int_0^T \widetilde{\alpha}'(t)\widetilde{c}\,^{\dagger}(t)\widetilde{\alpha}(t) dO(t) < \infty, \quad \text{for any} \quad T \geq 0.		\label{con : structural2}
	\end{align}
\end{prop}

\smallskip

\begin{proof}
	First, we assume that the supermartingale num\'eraire portfolio $\rho$ among the top $n$ stocks exists; then, from statement \textit{(3')} of \eqref{Eq : equiv statement (3)}, the identity $\widetilde{\alpha} = \widetilde{c}\rho$ holds. The condition \textit{(i)} follows immediately, and we obtain $\widetilde{c}\,\widetilde{c}\,^{\dagger}\widetilde{\alpha} = \widetilde{\alpha}$ from the property (b) above. This also implies that the set $\{\widetilde{\alpha} \in \textbf{range}(\widetilde{c})\}$ is predictable. We set the predictable process
	\begin{equation}		\label{def : superMG candidate}
		\nu := \widetilde{c}\,^{\dagger}\widetilde{\alpha},
	\end{equation}
	which is $\textbf{range}(\widetilde{c})$-valued in the $(\mathbb{P} \otimes O)$-a.e. sense, and satisfies $\widetilde{\alpha} = \widetilde{c}\,\nu$. Then, every supermartingale num\'eraire portfolio among the top $n$ stocks should be of the form
	\begin{equation}			\label{Eq : superMG candidate}
		\rho = \nu + \eta = \widetilde{c}\,^{\dagger}\widetilde{\alpha} + \eta,
	\end{equation}
	for a suitable predictable process $\eta$ which is in $\textbf{ker}(\widetilde{c})$, the kernel of $\widetilde{c}$, $(\mathbb{P}\otimes O)$-a.e. We have $\widetilde{c}\eta = 0$ and $\eta'\widetilde{\alpha} = 0$, thus $\eta$ is a null portfolio in the sense of Lemma~\ref{lem : null portfolio}.
	
	\smallskip
	
	On the other hand, the assumption that the supermartingale num\'eraire portfolio among the top $n$ stocks exists, implies that some $N$-dimensional process of the form $\rho = \widetilde{c}\,^{\dagger}\widetilde{\alpha}+\eta$ of the form \eqref{Eq : superMG candidate} should be a portfolio, i.e., $R$-integrable. The integrability condition \eqref{con : integrable w.r.t tilde R} in Lemma~\ref{lem : integrability condition} with the observation
	\begin{equation*}
		\rho'\widetilde{c}\rho = \rho'\widetilde{c}\,(\widetilde{c}\,^{\dagger}\widetilde{\alpha}+\eta) = \rho'\widetilde{\alpha} = \widetilde{\alpha}'\rho = \widetilde{\alpha}'\widetilde{c}\,^{\dagger}\widetilde{\alpha},
	\end{equation*}
	gives the condition \text{(ii)}.
	
	\bigskip
	
	We next assume the conjunction of conditions \textit{(i), (ii)} and find the supermartingale num\'eraire portfolio among the top $n$ stocks. We define the two predictable processes
	\begin{equation}		\label{def : nu and rho}
		\nu := \widetilde{c}\,^{\dagger}\widetilde{\alpha}, \qquad \text{and} \qquad \rho := D\nu = D\widetilde{c}\,^{\dagger}\widetilde{\alpha},
	\end{equation}
	and claim that $\rho$ is the supermartingale num\'eraire portfolio among the top $n$ stocks.
	Thanks to the condition \textit{(i)}, we obtain the identity $\widetilde{c}\nu = \widetilde{c}\,\widetilde{c}\,^{\dagger}\widetilde{\alpha} = \widetilde{\alpha}$, $(\mathbb{P} \otimes O)$-a.e. Then, we observe the series of identities
	\begin{equation*}
		\nu'\widetilde{c}\nu
		= \nu'\widetilde{\alpha}
		= \widetilde{\alpha}'\nu
		= \widetilde{\alpha}'\widetilde{c}\,^{\dagger}\widetilde{\alpha}, \qquad (\mathbb{P} \otimes O)-\text{a.e.},
	\end{equation*}
	as well as
	\begin{equation}		\label{Eq : rho in terms of alpha}
		\rho'\widetilde{c}\rho
		= \nu'D\widetilde{c}D\nu
		= \nu'\widetilde{c}\nu
		= \widetilde{\alpha}'\widetilde{c}\,^{\dagger}\widetilde{\alpha}, \qquad
		\rho'\widetilde{\alpha}
		= \nu'D\widetilde{\alpha}
		= \nu'\widetilde{\alpha}
		= \widetilde{\alpha}'\widetilde{c}\,^{\dagger}\widetilde{\alpha}, \qquad (\mathbb{P} \otimes O)-\text{a.e.}
	\end{equation}
	Here, we used the identities $D\widetilde{\alpha} = \widetilde{\alpha}$, and $D\widetilde{c}D = \widetilde{c}$ which can be obtained from \eqref{Eq : tilde alpha, c with D}. Combining equations of \eqref{Eq : rho in terms of alpha} with the condition \textit{(ii)} yields the integrability condition \eqref{con : integrable w.r.t tilde R} for $\rho \equiv \pi$ in Lemma~\ref{lem : integrability condition}, i.e., $\rho \in \pazocal{I}(R)$.  Also, from the construction \eqref{def : nu and rho}, we have $D\rho = DD\nu = D\nu = \rho$, thus $\rho \in \pazocal{T}(n)$. Therefore, we have shown that $\rho$ is a portfolio among the top $n$ stocks, i.e., $\rho \in \pazocal{I}(R) \cap \pazocal{T}(n)$.
	
	\smallskip
	
	Furthermore, we deduce
	\begin{equation}			\label{Eq : rho superMG property}
		\widetilde{c}\rho = \widetilde{c}D\nu = \widetilde{c}\nu = \widetilde{c}\widetilde{c}\,^{\dagger}\widetilde{\alpha} = \widetilde{\alpha}, \qquad (\mathbb{P} \otimes O)-\text{a.e.,}
	\end{equation}
	where the second equation uses the identity $\widetilde{c}D = \widetilde{c}$, a consequence of \eqref{Eq : tilde alpha, c with D} and of the fact that $D$ is idempotent. Thus, we have obtained the condition \eqref{Eq : equiv statement (3)}, which is equivalent to statement \textit{(3)} of Proposition~\ref{prop : supermartingale numeraire among the top n}, and $\rho$ is indeed the supermartingale num\'eraire portfolio among the top $n$ stocks.
\end{proof}

\bigskip

The conjunction of the two conditions in Proposition~\ref{prop : structural condition} can be formulated as one equivalent condition, as follows. We first recall the `growth rate' $\gamma_{\pi}$ of the portfolio $\pi \in \pazocal{I}(R) \cap \pazocal{T}(n)$ among the top $n$ stocks in \eqref{Eq : alpha pi, gamma pi}. We denote $\mathbb{R}^N \cap \pazocal{T}(n)$ the collection of elements in $\mathbb{R}^N$ such that at most $n$ components are nonzero; then $\pi(t)$ takes values in $\mathbb{R}^N \cap \pazocal{T}(n)$ for each $t \geq 0$, by the property \eqref{con : portfolio among n}. Let us define the $[0, \infty]$-valued process
\begin{equation}			\label{def : maximal growth rate}
	\widetilde{g} := \sup_{p \in \mathbb{R}^N} \Big( p' \widetilde{\alpha} - \frac{1}{2}p'\widetilde{c}p \Big)
	=\sup_{p \in \mathbb{R}^N \cap \pazocal{T}(n)} \Big( p' \widetilde{\alpha} - \frac{1}{2}p'\widetilde{c}p \Big).
\end{equation}
The last equality follows because of the identities
\begin{equation*}
	p'\widetilde{\alpha}-\frac{1}{2}p'\widetilde{c}p
	= p'D\widetilde{\alpha}-\frac{1}{2}p'D\widetilde{c}Dp 
	= \widetilde{p}'\widetilde{\alpha}-\frac{1}{2}\widetilde{p}'\widetilde{c}\widetilde{p},
\end{equation*}
valid for any $p \in \mathbb{R}^N$, where $\widetilde{p} := Dp \in \mathbb{R}^N \cap \pazocal{T}(n)$, by recalling the properties $\widetilde{\alpha} = D\widetilde{\alpha}$ and $\widetilde{c} = D\widetilde{c}D$ which can be deduced from \eqref{Eq : tilde alpha, c with D}. This process $\widetilde{g}$ can be interpreted as the \textit{maximal growth rate} achievable for all portfolios among the top $n$ stocks. Note that $\widetilde{g}$ is predictable, because the supremum can be restricted over a countable, dense subset of $\mathbb{R}^N$. We then easily rewrite the process $\widetilde{g}$ in the form
\begin{equation}			\label{Eq : maximal growth rate}
	\widetilde{g} = \frac{1}{2}\big(\widetilde{\alpha}'\widetilde{c}\,^{\dagger}\widetilde{\alpha}\big)\bm{1}_{\{\widetilde{\alpha} \in \textbf{range}(\widetilde{c})\}} + \infty \bm{1}_{\{\widetilde{\alpha} \notin \textbf{range}(\widetilde{c})\}},
\end{equation}
and the supremum of \eqref{def : maximal growth rate} is attained if and only if $\widetilde{g} < \infty$, at $p \equiv \rho := D\widetilde{c}\,^{\dagger}\widetilde{\alpha}$ as in \eqref{def : nu and rho} and \eqref{Eq : rho in terms of alpha}. Then, the conjunction of conditions \textit{(i) + (ii)} in Proposition~\ref{prop : structural condition} becomes simply
\begin{equation}		\label{con : structural condition}
	\widetilde{G}(T) < \infty, \qquad \text{for all} \quad T \geq 0,
\end{equation}
where $\widetilde{G}$ is an adapted nondecreasing process
\begin{equation}			\label{def : aggregate maximal growth}
	\widetilde{G}:= \int_0^{\cdot} \widetilde{g}(t) dO(t).
\end{equation}
We call this $\widetilde{G}$ the \textit{aggregate maximal growth from portfolios among the top $n$ stocks}; and say that the market consisting of the top $n$ stocks has \textit{locally finite growth}, if the process $\widetilde{G}$ satisfies the condition \eqref{con : structural condition}. We formalize this argument into the next proposition.

\smallskip

\begin{prop}	\label{prop : locally finite growth}
	The requirement of \eqref{con : structural condition} of locally finite growth among the top $n$ stocks, is equivalent to the conjunction of the two conditions \textit{(i) + (ii)} of Proposition~\ref{prop : structural condition}, thus sufficient and necessary for a supermartinagle num\'eraire portfolio among the top $n$ stocks to exist. In this case, we have
	\begin{equation*}
		\widetilde{G}= \Gamma_{\rho},
	\end{equation*}
	where $\rho$ is a supermartingale num\'eraire portfolio among the top $n$ stocks.
\end{prop}

\smallskip

We present the following results which will be used later.

\smallskip

\begin{lem}		\label{lem : locally finite growth}
	Suppose the market has locally finite growth among the top $n$ stocks, i.e., that \eqref{con : structural condition} holds, and let $\rho$ be the supermartingale num\'eraire portfolio among the top $n$ stocks. Recalling \eqref{Eq : alpha pi, gamma pi}, \eqref{def : A pi, Gamma pi}, \eqref{def : C pi rho}, and \eqref{Eq : M pi in tilde}, we have
	\begin{equation*}
		\widetilde{G} = \Gamma_{\rho} = \frac{1}{2}C_{\rho\rho},
	\end{equation*}
	as well as the representation
	\begin{equation}		\label{Eq : reciprocal of supermartingale numeraire portfolio}
		\frac{1}{X_{\rho}} = \pazocal{E}(-M_{\rho}).
	\end{equation}
\end{lem}

\begin{proof}
	As with \eqref{def : nu and rho} in the proof of Proposition~\ref{prop : structural condition}, the supermartingale num\'eraire portfolio $\rho$ among the top $n$ stocks is of the form $D\widetilde{c}\,^{\dagger}\widetilde{\alpha}$. With \eqref{Eq : rho in terms of alpha}, the claim $\widetilde{G} = \Gamma_{\rho}$ is easily obtained. Furthermore, again by \eqref{Eq : rho in terms of alpha} with \eqref{Eq : alpha pi, gamma pi}, we have
	\begin{equation}		\label{Eq : gamma of superMG portfolio}
		\gamma_{\rho} = \rho'\widetilde{\alpha} - \frac{1}{2}\rho'\widetilde{c}\rho
		= \frac{1}{2}\rho'\widetilde{c}\rho
		= \frac{1}{2}c_{\rho\rho}
		= \widetilde{g}
	\end{equation}
	thus $\Gamma_{\rho} = \frac{1}{2}C_{\rho\rho}$, as well as $A_{\rho} = C_{\rho\rho}$. We then write \eqref{def : relative cumulative return}, \eqref{Eq : R rho pi} with $\pi \equiv (0, \cdots, 0) \in \pazocal{I}(R) \cap \pazocal{T}(n)$:
	\begin{equation*}
		\frac{1}{X_{\rho}} = X^{\rho}_{\pi} = \pazocal{E}(\widetilde{R}^{\rho}_0)
		= \pazocal{E}(C_{\rho\rho}-R_{\rho}) = \pazocal{E}(C_{\rho\rho}-A_{\rho}-M_{\rho})
		= \pazocal{E}(-M_{\rho}).
	\end{equation*}
\end{proof}

\smallskip

\begin{lem}		\label{lem : X / X rho}
	Let $\rho$ be the supermartingale num\'eraire portfolio among the top $n$ stocks. For any investment strategy $\vartheta \in \pazocal{I}(S) \cap \pazocal{T}(n)$ among the top $n$ stocks, and for any initial capital $x \geq 0$, let us recall the wealth process $X \equiv X(\cdot;x, \vartheta)$ generated by $\vartheta$ and $x$ in the manner of \eqref{def : investment}. Then there exists a process $\eta = (\eta_1, \cdots, \eta_N) \in \pazocal{I}(\widetilde{M}) \cap \pazocal{T}(n)$, such that
	\begin{equation}		\label{Eq : relative wealth from strategy}
		\frac{X}{X_{\rho}} = x + \int_0^{\cdot} \sum_{i=1}^N \eta_i(t) d\widetilde{M}_i(t).
	\end{equation}
	Conversely, for any $x \geq 0$ and $\eta \in \pazocal{I}(M) \cap \pazocal{T}(n)$, there exists a process $\vartheta \in \pazocal{I}(S) \cap \pazocal{T}(n)$ such that \eqref{Eq : relative wealth from strategy} holds.
\end{lem}

\smallskip

\begin{proof}
	From \eqref{Eq : reciprocal of supermartingale numeraire portfolio} and \eqref{Eq : M pi in tilde}, we have $d\big(1/X_{\rho}(t)\big) = \big(1/X_{\rho}(t)\big)\sum_{i=1}^N \big(-\rho_i(t)\big)d\widetilde{M}_i(t)$, as well as the dynamics
	\begin{equation*}
		dX(t) = \sum_{i=1}^N \vartheta_i(t) dS_i(t) = \sum_{i=1}^N \vartheta_i(t) S_i(t) d\widetilde{R}_i(t) = \sum_{i=1}^N \vartheta_i(t) S_i(t) \big( d\widetilde{A}_i(t) + d\widetilde{M}_i(t) \big),
	\end{equation*}
	from \eqref{Eq : decomposition of tilde R}. Combining two equations via It\^o's formula, we obtain
	\begin{align*}
		d\big(X(t)/X_{\rho}(t)\big) 
		&= \sum_{i=1}^N \frac{\vartheta_i(t)S_i(t)}{X_{\rho}(t)} \big( d\widetilde{A}_i(t) + d\widetilde{M}_i(t) \big) + \frac{X(t)}{X_{\rho}(t)} \sum_{i=1}^N \big(-\rho_i(t)\big)d\widetilde{M}_i(t)
		\\
		&~~~+ \sum_{i=1}^N \sum_{j=1}^N \frac{\vartheta_i(t)S_i(t)}{X_{\rho}(t)} \big(-\rho_j(t)\big) d[\widetilde{M}_i, \widetilde{M}_j](t).
	\end{align*}
	Here, the finite variation terms vanish because of the relationship $d\widetilde{A}_i(t) = \sum_{j=1}^N \rho_j(t)d[\widetilde{M}_i, \widetilde{M}_j](t)$ for $i = 1, \cdots, N$, which is valid on the strength of condition $\widetilde{(3)}$ in Proposition~\ref{prop : supermartingale numeraire among the top n}. Thus, by setting 
	\begin{equation*}
		\eta_i(t) := \frac{\vartheta_i(t)S_i(t) - X(t)\rho_i(t)}{X_{\rho}(t)}, \quad i = 1, \cdots, N,
	\end{equation*}
	it is straightforward to check $\eta \in \pazocal{T}(n)$, and the result follows. The converse can be easily shown by reversing the above procedure.
\end{proof}
\bigskip

\subsection{Local martingale deflator and market viability}

\begin{defn} [Local martingale deflator] \label{Def : deflator}
	We call an adapted, right-continuous with left-limited process $Y$, a \textit{local martingale deflator among the top $n$ stocks}, if it satisfies $Y(0) = 1$, $Y > 0$, and the process $YX$ is a local martingale for every $X \in \pazocal{X}^n$ of Definition~\ref{Def : investment among n}. We denote by $\pazocal{Y}^n$ the collection of all local martingale deflators among the top $n$ stocks.
\end{defn}

\smallskip

Since $X \equiv 1 \in \pazocal{X}^n$, every deflator $Y \in \pazocal{Y}^n$ is in particular local martingale.

\smallskip

\begin{defn} [cumulative withdrawal stream]		\label{Def : withdrawal}
	We denote by $\pazocal{K}$ the collection of all nondecreasing, adapted and right-continuous processes $K$ with $K(0) = 0$. Any element $K$ of $\pazocal{K}$ is called \textit{cumulative withdrawal process}, and $K(t)$ represents for the cumulative capital withdrawn up to time $t \geq 0$; actual withdrawals in each infinitesimal interval $(t, t+dt]$ are represented as $dK(t)$. We say that $K \in \pazocal{K}$ is nonzero, if $\mathbb{P}\big( K(\infty) > 0 \big) > 0$.
	
	\smallskip
	
	For $x \geq 0$, $\vartheta \in \pazocal{I}(S)$ or $\vartheta \in \pazocal{I}(S) \cap \pazocal{T}(n)$, the wealth process $X(\cdot; x, \vartheta)$ defined in \eqref{def : investment} is said to \textit{finance} a given cumulative withdrawal process $K \in \pazocal{K}$, if $X \geq K$ holds. In this case, we say the process $K$ is \textit{financeable from the initial capital $x \geq 0$ with the investment strategy $\vartheta$}.
	
	\smallskip
	
	We denote by $\pazocal{K}(x)$, $\pazocal{K}^n(x)$ the subset of $\pazocal{K}$ consisting of cumulative capital withdrawal processes financeable from initial capital $x$; namely:
	\begin{equation}
		\pazocal{K}(x) := \{ K \in \pazocal{K} ~|~ \exists ~ \vartheta \in \pazocal{I}(S) \text{ such that } X(\cdot; x, \vartheta) \geq K \},
	\end{equation}
	\begin{equation}
		\pazocal{K}^n(x) := \{ K \in \pazocal{K} ~|~ \exists ~ \vartheta \in \pazocal{I}(S) \cap \pazocal{T}(n) \text{ such that } X(\cdot; x, \vartheta) \geq K \},
	\end{equation}
	
	\smallskip
	
	We introduce also the collection of cumulative withdrawal processes in $\pazocal{K}$ which can be financed starting from any positive initial capital:
	\begin{equation}
		\pazocal{K}(0+) := \bigcap_{x > 0} \pazocal{K}(x) \subset \pazocal{K}, \qquad 
		\pazocal{K}^n(0+) := \bigcap_{x > 0} \pazocal{K}^n(x) \subset \pazocal{K}.
	\end{equation}
\end{defn}

\smallskip

\begin{defn} [Superhedging capital]	\label{Def : superheding}
	For any cumulative withdrawal process $K \in \pazocal{K}$, we call the quantities
	\begin{equation}
		x(K) := \inf \{ x \geq 0 ~|~ K \in \pazocal{K}(x) \} = \inf \{ x \geq 0 ~|~ \exists ~\vartheta \in \pazocal{I}(S) \text{ such that } X(\cdot; x, \vartheta) \geq K \},
	\end{equation}
	\begin{equation}		\label{def : superhedging}
	x^n(K) := \inf \{ x \geq 0 ~|~ K \in \pazocal{K}^n(x) \} = \inf \{ x \geq 0 ~|~ \exists ~\vartheta \in \pazocal{I}(S) \cap \pazocal{T}(n) \text{ such that } X(\cdot; x, \vartheta) \geq K \}
	\end{equation}
	the \textit{superhedging capital} associated with the withdrawal stream $K$ in the entire market, and in the market consisting of the top $n$ stocks, respectively. We follow here the standard convention that the infimum of an empty set is equal to infinity.
\end{defn}

\smallskip

\begin{lem}			\label{lem : x^n(K)}
	Suppose that $\pazocal{Y}^n$ is nonempty. For a fixed cumulative withdrawal process $K \in \pazocal{K}$, we assume that it is financeable from the initial capital $x \geq 0$ with investment strategy $\vartheta \in \pazocal{I}(S) \cap \pazocal{T}(n)$, i.e.,
	\begin{equation*}
		X \equiv X(\cdot; x, \vartheta) = x + \int_0^{\cdot} \sum_{i=1}^N \vartheta_i(t)dS_i(t) \geq K.
	\end{equation*}
	Then, the process
	\begin{equation*}
		Y(X-K) + \int_0^{\cdot} Y(t-)dK(t)
	\end{equation*}
	is a nonnegative local martingale, thus also a supermartingale, for every local martingale deflator $Y \in \pazocal{Y}^n$ among the top $n$ stocks. In particular, $Y(X-K)$ is nonnegative supermartingale, for every $Y \in \pazocal{Y}^n$. Furthermore, for the quantity $x^n(K)$ of \eqref{def : superhedging} we have the inequality
	\begin{equation}		\label{Eq : x^n(K)}
		x^n(K) \geq \sup_{Y \in \pazocal{Y}^n} \mathbb{E}^{\mathbb{P}} \bigg[ \int_0^{\infty} Y(t-)dK(t) \bigg].
	\end{equation}
\end{lem}

\smallskip

\begin{proof}
	For every $Y \in \pazocal{Y}^n$, integration by parts gives
	\begin{equation*}
		Y(X-K) = YX - \int_0^{\cdot} Y(t-)dK(t) - \int_0^{\cdot} K(t-)dY(t),
	\end{equation*}
	thus
	\begin{equation}			\label{Eq : Y(X-K)}
		Y(X-K) + \int_0^{\cdot} Y(t-)dK(t) = YX - \int_0^{\cdot} K(t-)dY(t).
	\end{equation}
	Both terms on the right-hand side of \eqref{Eq : Y(X-K)} are local martingales, and the terms on the left hand side of \eqref{Eq : Y(X-K)} are nonnegative; thus the first claim follows. Also, the process $\int_0^{\cdot} Y(t-)dK(t)$ is nondecreasing, therefore $Y(X-K)$ is nonnegative supermartingale. We denote the left hand side of \eqref{Eq : Y(X-K)} by $Q := Y(X-K) + \int_0^{\cdot} Y(t-)dK(t)$, then we obtain
	\begin{equation*}
		Q(0) = x \geq \mathbb{E}^{\mathbb{P}}\Big[Q(\infty)\Big] \geq \mathbb{E}^{\mathbb{P}}\Big[\int_0^{\infty} Y(t-)dK(t)\Big].
	\end{equation*}
	By taking the supremum over $Y \in \pazocal{Y}^n$ and then the infimum over the initial capital $x \geq 0$, the last claim follows. 
\end{proof}

\smallskip

\begin{defn} [Viability] \label{Def : viability}
	We say that the entire market is \textit{viable} if, whenever $x(K) = 0$ holds for some cumulative withdrawal process $K \in \pazocal{K}$, we have $K \equiv 0$.
	
	In the same manner, we say the market consisting of the top $n$ stocks is \textit{viable}, if whenever $x^n(K) = 0$ holds for some cumulative withdrawal process $K \in \pazocal{K}$, we have $K \equiv 0$.
\end{defn}

\smallskip

The viability of the market consisting of the top $n$ stocks, is actually equivalent to the identity
\begin{equation}
	\pazocal{K}^n(0+) = \{0\};
\end{equation}
whereas the failure of such viability implies the strict inclusion $\pazocal{K}^n(0+) \supset \{0\}$. When the viability of the market consisting of the top $n$ stocks fails, there exists a nonzero cumulative withdrawal process $K \in \pazocal{K}$ which is financeable from any initial capital $x > 0$, no matter how minuscule; or equivalently, there exists an investment strategy $\vartheta_m \in \pazocal{I}(R) \cap \pazocal{T}(n)$ for each $m \in \mathbb{N}$, such that
\begin{equation}
	X(\cdot; \frac{1}{m}, \vartheta_m) \geq K.
\end{equation}
We further present the following lemma; it can be proven in the same manner as Exercise~2.22 of \cite{KK2}.

\smallskip

\begin{lem}				\label{lem : not viable}
	The market consisting of the top $n$ stocks fails to be viable if, and only if, there exist a real number $T \geq 0$ and a nonnegative $\pazocal{F}(T)$-measurable random variable $h$ with $\mathbb{P}[h > 0] > 0$ such that for every $m \in \mathbb{N}$, there exists an $X^m \in \pazocal{X}^n$ with $X^m(T) \geq mh$.
\end{lem}

\smallskip

The following result presents another equivalent characterization of viability for the market consisting of the top $n$ stocks.

\smallskip

\begin{prop} [Boundedness in probability] 	\label{prop : boundedness in prob}
	The market consisting of the top $n$ stocks is viable if, and only if,
	\begin{equation}		\label{Eq : boundedness in prob}
		\lim_{m \rightarrow \infty} \sup_{X \in \pazocal{X}^n} \mathbb{P} [X(T) > m] = 0, \qquad \forall~T \geq 0.
	\end{equation}
\end{prop}

\smallskip

\begin{proof}
	We first assume that the market consisting of the top $n$ stocks is not viable. Then, from Lemma~\ref{lem : not viable}, there exist a real number $T \geq 0$, a nonnegative $\pazocal{F}(T)$-measurable random variable $h$ with $\mathbb{P}[h > 0] > 0$, and a sequence $(X^m)_{m \in \mathbb{N}}$ of wealth processes $X^m \in \pazocal{X}^n$ satisfying $X^m(T) \geq mh$. Pick $\epsilon > 0$ sufficiently small, so that $\mathbb{P}[h > \epsilon] > \epsilon$ holds. We then have
	\begin{equation*}
		\liminf_{m \rightarrow \infty} \mathbb{P}[X^m(T) > \epsilon m]
		\geq \liminf_{m \rightarrow \infty} \mathbb{P}[X^m(T) > mh, h > \epsilon]
		\geq \epsilon,
	\end{equation*}
	thus the condition \eqref{Eq : boundedness in prob} is violated.
	
	\smallskip
	
	Conversely, we assume that for some $T \geq 0$, there exist $\epsilon > 0$ and a sequence $(X^m)_{m \in \mathbb{N}} \subset \pazocal{X}^n$ such that $\mathbb{P}[X^m(T) > m2^m] > \epsilon$ hold for all $m \in \mathbb{N}$. Consider the set
	\begin{equation*}
		H := \bigcap_{m=1}^{\infty} \bigcup_{k=m}^{\infty} \big\{ X^k(T) > k2^k \big\} \in \pazocal{F}(T),
	\end{equation*}
	and note that $\mathbb{P}(H) \geq \epsilon$. For every $m \in \mathbb{N}$, the inclusion
	\begin{equation*}
		H \subseteq \bigcup_{k=m+1}^{\infty}\{ X^k(T) > k2^k \}
	\end{equation*}
	holds, so there exists a sufficiently large number $K_m > m$ such that the set
	\begin{equation*}
		H_m := H \cap \Big( \bigcup_{k=m+1}^{K_m}\{ X^k(T) > k2^k \} \Big) \in \pazocal{F}(T)
	\end{equation*}
	satisfies $\mathbb{P}[H \setminus H_m] \leq \frac{\mathbb{P}[H]}{2^{m+1}}$. Then, the countable intersection 
	\begin{equation*}
		E := \bigcap_{m=1}^{\infty} H_m \in \pazocal{F}(T)
	\end{equation*}
	is a subset of $H$, and we have
	\begin{equation*}
		\mathbb{P}[H \setminus E] = \mathbb{P}\Big[\bigcup_{m=1}^{\infty}(H \setminus H_m)\Big]
		\leq \sum_{m=1}^{\infty} \frac{\mathbb{P}[H]}{2^{m+1}} = \frac{\mathbb{P}[H]}{2},
	\end{equation*}
	thus, $\mathbb{P}[E] \geq \frac{\mathbb{P}[H]}{2}$ and $\mathbb{P}[E] \geq \frac{\epsilon}{2} > 0$. Let us define a sequence of num\'eraires $(\Xi^m)_{m \in \mathbb{N}}$
	\begin{equation*}
		\Xi^m := \sum_{k=m+1}^{K_m} 2^{-(k-m)}X^k, \qquad \text{for each } m \in \mathbb{N},
	\end{equation*}
	and it is straightforward that $\Xi^m \in \pazocal{X}^n$, as all $X^k \in \pazocal{X}^n$ for $k \in \mathbb{N}$. Furthermore, for every $m \in \mathbb{N}$, we have $E \subseteq H_m \subseteq \{ \Xi^m(T) > m \}$, from which $\Xi^m(T) \geq m\bm{1}_E$ follows. Set $h := \bm{1}_E \in \pazocal{F}(T)$, then
	\begin{equation*}
		\mathbb{P}[h > 0] = \mathbb{P}[E] \geq \frac{\epsilon}{2} > 0.
	\end{equation*}
	Lemma~\ref{lem : not viable} yields that the market consisting of the top $n$ stocks is not viable.
\end{proof}

\smallskip

We are now ready to state and prove the main result of this section.

\smallskip

\begin{thm} 			\label{thm : main result}
	The following statements are equivalent:
	\begin{enumerate} [(1)]
		\item The market consisting of the top $n$ stocks is viable.
		\item There exists a local martingale deflator among the top $n$ stocks, i.e., $\pazocal{Y}^n \neq \emptyset$.
		\item The supermartingale num\'eraire among the top $n$ stocks exists.
		\item The market consisting of the top $n$ stocks has locally finite growth; namely, the condition \eqref{con : structural condition} of the aggregate maximal growth process $\widetilde{G}$ among the top $n$ stocks of \eqref{def : aggregate maximal growth} holds.
	\end{enumerate}
\end{thm}

\smallskip

\begin{proof}
	The implication \textit{(4)} $\Rightarrow$ \textit{(3)} follows from Proposition~\ref{prop : locally finite growth}. The implication \textit{(3)} $\Rightarrow$ \textit{(2)} also follows easily, because the supermartingale num\'eraire among the top $n$ stocks is a local martingale num\'eraire among the top $n$ stocks from Proposition~\ref{prop : supermartingale numeraire among the top n}, and the reciprocal of the local martingale num\'eraire among the top $n$ stocks is a local martingale deflator among the top $n$ stocks.
	
	\smallskip
	
	In order to prove \textit{(2)} $\Rightarrow$ \textit{(1)}, let $Y \in \pazocal{Y}^n$ be a local martingale deflator and pick a cumulative withdrawal process $K \in \pazocal{K}$ such that $x^n(K) = 0$. From \eqref{Eq : x^n(K)} of Lemma~\ref{lem : x^n(K)}, we have
	\begin{equation*}
		\mathbb{E}^{\mathbb{P}} \bigg[ \int_0^{\infty} Y(t-)dK(t) \bigg] = 0.
	\end{equation*}
	Since $Y$ is strictly positive and $K$ is nondecreasing with $K(0) = 0$, it follows that $K(\infty) = 0$ holds $\mathbb{P}$-a.e., which is equivalent to $K \equiv 0$. The market consisting of the top $n$ stocks is then viable.
	
	\smallskip
	
	The remaining part is to show the implication \textit{(1)} $\Rightarrow$ \textit{(4)}, which is quite technical. Suppose that the market fails to have locally finite growth among the top $n$ stocks, i.e., one of the structural conditions \eqref{con : structural1}, \eqref{con : structural2} is violated. Thus, we need to consider two cases:
	\begin{enumerate}[(A)]
		\item the set $\{\widetilde{\alpha} \not\in \textbf{range}(\widetilde{c}) \}$ fails to be $(\mathbb{P} \otimes O)$-null,
		\item the set $\{\widetilde{\alpha} \not\in \textbf{range}(\widetilde{c}) \}$ is $(\mathbb{P} \otimes O)$-null, but $\mathbb{P}[\widetilde{G}(T) = \infty] > 0$ holds for some $T > 0$.
	\end{enumerate}
	We shall show that the market is not viable in each of the cases \textit{(A)} and \textit{(B)} below.
	
	\smallskip
	
	\noindent
	$\ast$ \textbf{Case (A)}. Recalling the notation \eqref{def : tilde c dagger} with its properties (a)-(c), we first note that the predictable process
	\begin{equation}		\label{def : varphi example}
		\varphi := \frac{1}{||\widetilde{\alpha}-\widetilde{c}\widetilde{c}\,^{\dagger}\widetilde{\alpha}||^2} \big(\widetilde{\alpha}-\widetilde{c}\widetilde{c}\,^{\dagger}\widetilde{\alpha}\big) \bm{1}_{\{\widetilde{\alpha} \not\in \textbf{range}(\widetilde{c})\}},
	\end{equation}
	is well-defined, because $\widetilde{\alpha} \not\in \textbf{range}(\widetilde{c})$ holds if and only if $\widetilde{c}\widetilde{c}\,^{\dagger}\widetilde{\alpha} \neq \widetilde{\alpha}$. Note that $D\varphi = \varphi$, thus $\varphi \in \pazocal{T}(n)$, thanks to the properties $D\widetilde{\alpha} = \widetilde{\alpha}$, $D\widetilde{c} = \widetilde{c}$ from \eqref{Eq : tilde alpha, c with D}. Since the process $\widetilde{\alpha}-\widetilde{c}\widetilde{c}\,^{\dagger}\widetilde{\alpha}$ is orthogonal to $\textbf{range}(\widetilde{c})$, we have $\widetilde{c}\varphi = 0$. Furthermore, we have $\varphi' \widetilde{\alpha} = \bm{1}_{\{\widetilde{\alpha} \not\in \textbf{range}(\widetilde{c})\}}$, because
	\begin{equation*}
		(\widetilde{\alpha}-\widetilde{c}\widetilde{c}\,^{\dagger}\widetilde{\alpha})'\widetilde{\alpha}
		= ||\widetilde{\alpha}-\widetilde{c}\widetilde{c}\,^{\dagger}\widetilde{\alpha}||^2 + (\widetilde{\alpha}-\widetilde{c}\widetilde{c}\,^{\dagger}\widetilde{\alpha})'(\widetilde{c}\widetilde{c}\,^{\dagger}\widetilde{\alpha}) = ||\widetilde{\alpha}-\widetilde{c}\widetilde{c}\,^{\dagger}\widetilde{\alpha}||^2.
	\end{equation*}
	Thus, from Lemma~\ref{lem : integrability condition}, $\varphi$ is a portfolio among the top $n$ stocks, i.e., $\varphi \in \pazocal{I}(R) \cap \pazocal{T}(n)$. Also, the local martingale vanishes: $\int_0^{\cdot} \varphi'(t)d\widetilde{M}(t) \equiv 0$, because its quadratic variation process vanishes
	\begin{equation}		\label{Eq : vanishing martingale}
		\Big[\int_0^{\cdot} \varphi'(t)d\widetilde{M}(t)\Big] = \int_0^{\cdot} \varphi'(t)\widetilde{c}(t)\varphi(t)dO(t)
		\equiv 0.
	\end{equation}
	Thus, 
	\begin{equation*}
		\int_0^{\cdot} \varphi'(t)d\widetilde{R}(t)
		= \int_0^{\cdot} \varphi'(t)d\widetilde{A}(t)
		= \int_0^{\cdot} \varphi'(t)\widetilde{\alpha}(t)dO(t)
		= \int_0^{\cdot} \bm{1}_{\{\widetilde{\alpha} \not\in \textbf{range}(\widetilde{c})\}}(t)dO(t)
		=: K.
	\end{equation*}
	We define the vector process $\vartheta \equiv (\vartheta_1, \cdots, \vartheta_N)$ with components given by $\vartheta_i = \varphi_i/S_i$ for $i = 1, \cdots, N$. It is then easy to check that $m\vartheta$ is an investment strategy among the top $n$ stocks, i.e., $m\vartheta \in \pazocal{I}(S) \cap \pazocal{T}(n)$, for any $m \in \mathbb{N}$, and
	\begin{equation*}
		X(\cdot; 0, m\vartheta)
		= \int_0^{\cdot}m\vartheta'(t)dS(t)
		= m\int_0^{\cdot}\varphi'(t)dR(t) 
		= m\int_0^{\cdot}\varphi'(t)d\widetilde{R}(t)
		= mK.
	\end{equation*}
	In other words, for any $m \in \mathbb{N}$, the wealth process generated by the investment strategy $m\vartheta$ among the top $n$ stocks has vanishing local martingale part, and is equal to the non-trivial, nondecreasing part $mK$ of finite variation. This process $mK$ can be arbitrarily scaled by the multiplicative constant $m \in \mathbb{N}$, and thus $x^n(K) = 0$, by recalling \eqref{def : superhedging}. We conclude that the market consisting of the top $n$ stocks is not viable.
	
	\smallskip
	
	\noindent
	$\ast$ \textbf{Case (B)}. We assume that the set $\{\widetilde{\alpha} \not\in \textbf{range}(\widetilde{c}) \}$ is $(\mathbb{P} \otimes O)$-null, but $\mathbb{P}[\widetilde{G}(T) = \infty] > 0$ holds for some $T > 0$. In this case, the aggregate maximal growth process $\widetilde{G}$ of \eqref{def : aggregate maximal growth} becomes
	\begin{equation}	\label{Eq : G in case B}
		\widetilde{G} = \frac{1}{2} \int_0^{\cdot} \widetilde{\alpha}'(t)\widetilde{c}\,^{\dagger}(t)\widetilde{\alpha}(t) dO(t).
	\end{equation}
	We consider first the portfolio $\rho := D\widetilde{c}\,^{\dagger}\widetilde{\alpha} \in \pazocal{T}(n)$ as in \eqref{def : nu and rho}, and also set $\rho^m := \rho \bm{1}_{\{ ||\rho|| \leq m \}} \in \pazocal{I}(R) \cap \pazocal{T}(n)$. The log-wealth process of \eqref{def : log of wealth} can be represented, with the help of \eqref{Eq : M pi in tilde} and \eqref{Eq : rho in terms of alpha}, as
	\begin{equation}		\label{Eq : log rho^m}
		\log X_{\rho^m}
		= \frac{1}{2} \int_0^{\cdot} \bm{1}_{\{ ||\rho(t)|| \leq m \}} \rho'(t)\widetilde{c}(t)\rho(t)dO(t)
		+ \int_0^{\cdot} \bm{1}_{\{ ||\rho(t)|| \leq m \}} \rho'(t)d\widetilde{M}(t).
	\end{equation}
	Note that the first integral on the right-hand side of \eqref{Eq : log rho^m}, namely
	\begin{equation*}
		2G^m := \int_0^{\cdot} \bm{1}_{\{ ||\rho(t)|| \leq m \}} \rho'(t)\widetilde{c}(t)\rho(t)dO(t),
	\end{equation*}
	is the quadratic variation of the local martingale $\int_0^{\cdot} \bm{1}_{\{ ||\rho(t)|| \leq m \}} \rho'(t)d\widetilde{M}(t)$, which is the second integral on the right-hand side of \eqref{Eq : log rho^m}. The Dambis-Dubins-Schwarz representation (cf. Theorem~3.4.6 and Problem~3.4.7 of \cite{KS1}), with the scaling property of Brownian motion, implies that there exists a Brownian motion $W^m$, on a possibly enlarged filtered probability space, such that
	\begin{equation}		\label{Eq : DDS representation}
		\log X_{\rho^m} = G^m + \sqrt{2}W^m(G^m),
	\end{equation}
	for every $m \in \mathbb{N}$. The sequence $\{ G^m(T) \}_{m \in \mathbb{N}}$ is nondecreasing and converges to 
	\begin{equation*}
		\frac{1}{2} \int_0^{\cdot} \rho'(t)\widetilde{c}(t)\rho(t)dO(t)
		= \frac{1}{2} \int_0^{\cdot} \widetilde{\alpha}'(t)\widetilde{c}\,^{\dagger}(t)\widetilde{\alpha}(t)dO(t)
		= \widetilde{G}(T),		
	\end{equation*}
	as in \eqref{Eq : G in case B}, again with the help of \eqref{Eq : rho in terms of alpha}. The strong law of large numbers for Brownian motion gives
	\begin{equation*}
		\lim_{m \rightarrow \infty} \mathbb{P} \bigg[ \frac{W^m\big(G^m(T)\big)}{G^m(T)} \leq -\frac{1}{2\sqrt{2}}, \quad \widetilde{G}(T) = \infty \bigg] = 0.
	\end{equation*}
	From the representation \eqref{Eq : DDS representation}, we obtain
	\begin{equation*}
		\lim_{m \rightarrow \infty} \mathbb{P} \bigg[ \frac{\log X_{\rho^m}(T)}{G^m(T)} \leq \frac{1}{2}, \quad \widetilde{G}(T) = \infty \bigg] = 0.
	\end{equation*}
	Therefore, in case \textit{(B)}, the collection of random variables $\{ X_{\rho^m}(T) ~|~m \in \mathbb{N} \} \subseteq \{ X(T) ~|~X \in \pazocal{X}^n \}$ fails to be bounded in probability, and Proposition~\ref{prop : boundedness in prob} concludes that the market consisting of the top $n$ stocks is not viable.
\end{proof}

\bigskip

\subsection{Growth optimality and relative log-optimality}

The results in the previous subsection characterize the supermartingale num\'eraire portfolio among the top $n$ stocks, via the `structural condition', in terms of $\widetilde{\alpha}$ and $\widetilde{c}$. More specifically, in the argument leading to Proposition~\ref{prop : structural condition} and in the proof of Lemma~\ref{lem : locally finite growth}, the \textit{maximal growth rate among the top $n$ stocks} $\widetilde{g}$ of \eqref{def : maximal growth rate} is attained when the portfolio is the supermartingale num\'eraire portfolio among the top $n$ stocks, as in \eqref{Eq : gamma of superMG portfolio}. In this subsection, we reformulate this property and show that the supermartingale num\'eraire portfolio is `optimal' in some sense among portfolios of top $n$ stocks.

\smallskip

\begin{defn} [Relative growth and growth optimality]		\label{Def : relative growth}
	We define the \textit{relative growth} of a given portfolio $\pi \in \pazocal{I}(R)$ with respect to another portfolio $\rho \in \pazocal{I}(R)$ as
	\begin{equation}		\label{def : relative growth}
	\Gamma^{\rho}_{\pi} := \Gamma_{\pi} - \Gamma_{\rho},
	\end{equation}
	namely, the difference between the finite variation process of the log-relative wealth process $\log X^{\rho}_{\pi} = \log (X_{\pi}/X_{\rho})$ from \eqref{def : Gamma pi}, \eqref{def : relative wealth}.
	
	\smallskip
	
	We call a portfolio $\rho \in \pazocal{I}(R) \cap \pazocal{T}(n)$ \textit{growth-optimal among the top $n$ stocks}, if for every portfolio $\pi \in \pazocal{I}(R) \cap \pazocal{T}(n)$ the process $\Gamma^{\rho}_{\pi} = \Gamma_{\pi}-\Gamma_{\rho}$ is non-increasing.
\end{defn}

\smallskip

\begin{prop}		\label{prop : growth-optimal}
	A portfolio is growth-optimal among the top $n$ stocks, if and only if it is a supermartingale num\'eraire portfolio among the top $n$ stocks.
\end{prop}

\begin{proof}
	(i) Let us first assume that $\rho \in \pazocal{I}(R) \cap \pazocal{T}(n)$ is the supermartingale num\'eraire portfolio among the top $n$ stocks. From Proposition~\ref{prop : structural condition} and \eqref{Eq : equiv statement (3)}, we know that $\widetilde{\alpha} \in \textbf{range}(\widetilde{c})$ and $\widetilde{\alpha} = \widetilde{c}\rho$ hold $(\mathbb{P} \otimes O)-a.e.$ Recalling \eqref{Eq : alpha pi, gamma pi}, \eqref{def : maximal growth rate} and the fact that the supremum of $g$ is attained at the supermartingale num\'eraire portfolio among the top $n$ stocks, $\gamma_{\rho} = \widetilde{g} \geq \gamma_{\pi}$ holds $(\mathbb{P} \otimes O)-a.e.$ for every $\pi \in \pazocal{I}(R) \cap \pazocal{T}(n)$. Thus, $\rho$ is growth-optimal.
	
	\smallskip
	
	(ii) Next, we assume that $\nu \in \pazocal{I}(R) \cap \pazocal{T}(n)$ is a growth-optimal portfolio among the top $n$ stocks. We pick a portfolio $\varphi \in \pazocal{I}(R) \cap \pazocal{T}(n)$ satisfying $\widetilde{c}\varphi = 0$ and $\varphi'\widetilde{\alpha}=1$ on the set $\{\widetilde{\alpha} \not\in \textbf{range}(\widetilde{c})\}$ (for example, as in \eqref{def : varphi example} in the proof of Theorem~\ref{thm : main result}). We then have $\gamma_{\nu+\varphi} = \gamma_{\nu}+1$ on $\{\widetilde{\alpha} \not\in \textbf{range}(\widetilde{c})\}$ from \eqref{Eq : alpha pi, gamma pi}, violating the growth-optimality of $\nu$. This implies that the latter set is $(\mathbb{P} \otimes O)$-null. In particular, $\widetilde{g} < \infty$ in the $(\mathbb{P} \otimes O)-a.e.$ sense, from \eqref{Eq : maximal growth rate}.
	
	On the other hand, we let $\rho := D\widetilde{c}\,^{\dagger}\widetilde{\alpha} \in \pazocal{T}(n)$ and define $\rho^m := \rho\bm{1}_{\{||\rho|| \leq m\}} \in \pazocal{I}(R) \cap \pazocal{T}(n)$ for $m \in \mathbb{N}$. The equation \eqref{Eq : gamma of superMG portfolio} yields $\gamma_{\nu} \geq \gamma_{\rho^m} = \widetilde{g}\bm{1}_{\{||\rho|| \leq m\}}$, and thus $\gamma_{\nu} \geq \widetilde{g}$ holds $(\mathbb{P} \otimes O)-a.e.$ by taking the limit $m \rightarrow \infty$. We conclude that $\nu$ is also a supermartingale num\'eraire portfolio among the top $n$ stocks.
\end{proof}

\bigskip

The supermartingale num\'eraire portfolio among the top $n$ stocks is `optimal' also in another sense, as follows.

\smallskip

\begin{defn}		\label{Def : relatively log-optimal}
	A portfolio $\rho \in \pazocal{I}(R) \cap \pazocal{T}(n)$ is called \textit{relatively log-optimal among the top $n$ stocks}, if for all portfolios $\pi \in \pazocal{I}(R) \cap \pazocal{T}(n)$ and for all stopping times $\tau$ of $\pazocal{F}$, we have
	\begin{equation}		\label{con : relatively log-optimal}
	\mathbb{E}^{\mathbb{P}} \big[ \big(\log X^{\rho}_{\pi}(\tau)\big)^+ \big] < \infty,
	\qquad \text{and} \qquad 
	\mathbb{E}^{\mathbb{P}} \big[ \log X^{\rho}_{\pi}(\tau) \big] \leq 0.
	\end{equation}	
\end{defn}

\smallskip

\begin{prop}		\label{prop : relatively log-optimal}
	A portfolio is relatively log-optimal among the top $n$ stocks, if and only if it is a supermartingale num\'eraire portfolio among the top $n$ stocks.
\end{prop}

\begin{proof}
	(i) We first suppose that $\rho \in \pazocal{I}(R) \cap \pazocal{T}(n)$ is the supermartingale num\'eraire portfolio among the top $n$ stocks. Then, we obtain
	\begin{equation*}
	\mathbb{E}^{\mathbb{P}} \big[ \big(\log X^{\rho}_{\pi}(\tau)\big)^+ \big]
	= \int_0^{\infty} \mathbb{P}(X^{\rho}_{\pi}(\tau) > e^t) dt
	\leq \int_0^{\infty} \mathbb{P}(X^{\rho}_{\pi}(\tau) > t) dt
	\leq \mathbb{E}^{\mathbb{P}} [ X^{\rho}_{\pi}(\tau) ]
	\leq 1,
	\end{equation*}
	where the last inequality is from the Optional Sampling Theorem. By applying Jensen's inequality to this last inequality, the second condition of \eqref{con : relatively log-optimal} also holds, and we conclude that $\rho$ is relatively log-optimal among the top $n$ stocks.
	
	\medskip
	
	(ii) For the converse implication, we assume that $\nu \in \pazocal{I}(R) \cap \pazocal{T}(n)$ is relatively log-optimal among the top $n$ stocks. As in the proof of Proposition~\ref{prop : growth-optimal}, we pick a portfolio $\varphi \in \pazocal{I}(R) \cap \pazocal{T}(n)$ as in \eqref{def : varphi example}, satisfying $\widetilde{c}\varphi = 0$ and $\varphi'\widetilde{\alpha}=1$ on the set $\{\widetilde{\alpha} \not\in \textbf{range}(\widetilde{c})\}$. By recalling \eqref{def : log of wealth}, \eqref{def : Gamma pi}, \eqref{def : A pi, Gamma pi}, \eqref{Eq : M pi in tilde}, and \eqref{Eq : alpha pi, gamma pi}, straightforward computations show
	\begin{align}
	\log X^{\nu}_{\nu+\varphi} &= \log X_{\nu+\varphi} - \log X_{\nu}
	= \int_0^{\cdot} \big(\gamma_{\nu+\varphi}(t) - \gamma_{\nu}(t) \big) dO(t) + \int_0^{\cdot} \varphi'(t) d\widetilde{M}(t)		\label{Eq : log relative wealth}
	\\
	&= \int_0^{\cdot} \bm{1}_{\{\widetilde{\alpha} \not\in \textbf{range}(\widetilde{c})\}}(t)dO(t).	\nonumber
	\end{align}
	Here, the last integral on the right-hand side of \eqref{Eq : log relative wealth} vanishes, because of the equation \eqref{Eq : vanishing martingale} above. The relative log-optimality of $\nu$ implies that the set $\{\widetilde{\alpha} \not\in \textbf{range}(\widetilde{c})\}$ is $(\mathbb{P} \otimes O)$-null. We then consider a process $\rho := D\widetilde{c}\,^{\dagger}\widetilde{\alpha} \in \pazocal{T}(n)$ of \eqref{def : nu and rho}, as in the proof of Proposition~\ref{prop : growth-optimal}. Note that $\widetilde{\alpha} = \widetilde{c}\rho$ holds $(\mathbb{P} \otimes O)-a.e.$ from \eqref{Eq : rho superMG property}, or equivalently, $\widetilde{A}_i = C_{\widetilde{i}\rho}$ hold for $i = 1, \cdots, N$, from Remark~\ref{rem : equiv statement (3)}. This last requirement implies that $A_{\pi} = C_{\pi\rho}$, thus $R_{\pi}-C_{\pi\rho}$ is local martingale for every $\pi \in \pazocal{I}(R) \cap \pazocal{T}(n)$. We further define
	\begin{equation*}
	\nu^m := \nu \bm{1}_{\{ \widetilde{\alpha} = \widetilde{c}\nu \}} + \nu \bm{1}_{\{ \widetilde{\alpha} \neq \widetilde{c}\nu \}} \bm{1}_{\{ ||\rho|| > m \}} + \rho \bm{1}_{\{ \widetilde{\alpha} \neq \widetilde{c}\nu \}} \bm{1}_{\{ ||\rho|| \leq m \}}, \qquad \text{for} \quad m \in \mathbb{N},
	\end{equation*}
	and it is easy to check that $\nu^m \in \pazocal{I}(R) \cap \pazocal{T}(n)$ for all $m \in \mathbb{N}$.
	
	\smallskip
	
	We now claim that the ratio $X_{\nu} / X_{\nu^m}$ for every $m \in \mathbb{N}$ is a local martingale. Proposition~\ref{prop : relative wealth} implies that it is sufficient to show $R^{\nu^m}_{\nu} = R_{\pi} - C_{\pi \nu^m} =:Q$ is a local martingale, where we set $\pi := \nu - \nu^m \in \pazocal{I}(R) \cap \pazocal{T}(n)$. On the set $\zeta := \{ \widetilde{\alpha} \neq \widetilde{c}\nu, ||\rho|| \leq m \}$, we have $\nu^m = \rho$, thus $Q$ is local martingale. On the complement set $\zeta^c$, we have $\pi = \nu - \nu^m = 0$, thus $Q = 0$. In other words, we showed that
	\begin{equation*}
	Q = \int_0^{\cdot} \bm{1}_{\zeta}(t) dQ(t) = \int_0^{\cdot} \bm{1}_{\zeta}(t) d(R_{\pi}-C_{\pi\rho})(t)
	\end{equation*}
	is local martingale, verifying our claim that $X_{\nu} / X_{\nu^m}$ is a local martingale for every $m \in \mathbb{N}$. As the ratio is positive, $X_{\nu} / X_{\nu^m}$ is also a supermartingale.
	
	\smallskip
	
	If we assume that $\mathbb{P} [X_{\nu}(T) \neq X_{\nu^m}(T)] > 0$ were true for some $T > 0$, we obtain
	\begin{equation*}
	\mathbb{E}^{\mathbb{P}} \bigg[ \log \frac{X^{\nu}(T)}{X^{\nu^m}(T)} \bigg] 
	< \log \mathbb{E}^{\mathbb{P}} \bigg[ \frac{X^{\nu}(T)}{X^{\nu^m}(T)} \bigg]
	\leq 0,
	\end{equation*} 
	contradicting the relative log-optimality of $\nu$. Thus, we conclude that $X_{\nu} = X_{\nu^m}$, from the continuity of $X_{\nu}/X_{\nu^m}$, and $\nu - \nu^m$ is a null portfolio in the sense of Lemma~\ref{lem : null portfolio}. We then have $\widetilde{c}{\nu^m} = \widetilde{c}{\nu} = \widetilde{c}{\rho} = \widetilde{\alpha}$, $(\mathbb{P} \otimes O)-a.e$, on the set $\zeta = \{ \widetilde{\alpha} \neq \widetilde{c}\nu, ||\rho|| \leq m \}$ defined above, which implies that $\zeta$ is $(\mathbb{P} \otimes O)$-null. Since this property is true for every $m \in \mathbb{N}$, the identity $\widetilde{\alpha} = \widetilde{c}\nu$ is valid $(\mathbb{P} \otimes O)-a.e$, thus $\nu$ is the supermartingale num\'eraire portfolio among the top $n$ stocks.
\end{proof}

\bigskip

In part (ii) of the proofs of both Proposition~\ref{prop : growth-optimal} and Proposition~\ref{prop : relatively log-optimal}, we did not assume the existence of supermartingale num\'eraire portfolio among the top $n$ stocks. Thus, the existence of growth-optimal or relatively log-optimal portfolio among the top $n$ stocks is equivalent to the existence of the supermartingale num\'eraire portfolio among the top $n$ stocks, and we can add the following two statements to the list of equivalences in Theorem~\ref{thm : main result}: 

\begin{enumerate} [(1)]
	\setcounter{enumi}{4}
	\item A growth-optimal portfolio among the top $n$ stocks exists.
	\item A relatively log-optimal portfolio among the top $n$ stocks exists.
\end{enumerate}

\bigskip

\subsection{The optional decomposition}

Suppose that we are given a nonnegative, adapted process with RCLL paths and $X(0) = x \geq 0$. In this subsection, we characterize the condition when $X$ belongs to $\pazocal{X}^n$ of Definition~\ref{Def : investment among n}, i.e., when $X$ is the wealth process generated by an investment strategy that invests in the top $n$ stocks of the market, and study how can we construct this strategy from $X$. The following Theorem~\ref{thm : optional decomposition} (or Corollary~\ref{cor : optional decomposition}), which we call the Optional Decomposition Theorem, gives the answer to this question.

\smallskip

We first present the following result, originally from Theorem~1 of \cite{Schweizer:1995}. See also Propositions~2.3 and 3.2 of \cite{Lar:Zit:2007}. We recall for this purpose the semimartingale vector $\widetilde{M}$ defined in \eqref{def : tilde A and M} and  write $\pazocal{M}^{\perp}_{loc}(\widetilde{M})$ the collection of scalar local martingales $L$ with RCLL paths, satisfying $L(0) = 0$ and the orthogonality $[L, \widetilde{M}_i] = 0$ for all $i = 1, \cdots, N$.

\smallskip

\begin{lem}		\label{lem : deflator representation}
	If the supermartingale num\'eraire portfolio $\rho$ among the top $n$ stocks exists, then the collections $\pazocal{Y}^n$ of local martingale deflators among the top $n$ stocks, defined in Definition~\ref{Def : deflator}, admits the representation:
	\begin{equation}		\label{Eq : local martingale deflator representation}
		\pazocal{Y}^n = \Big\{ \frac{1}{X_{\rho}}\pazocal{E}(L) : L \in \pazocal{M}^{\perp}_{loc}(\widetilde{M}) \quad \text{with} \quad \Delta L > -1 \Big\}.
	\end{equation}
\end{lem}

\smallskip
In order to simplify the proof of the Optional Decomposition Theorem, we shall work under the following assumption. The general case of the Theorem can be proven as in the Subsection 3.1.3 of \cite{KK2}.

\medskip

\noindent
\textbf{Assumption} $\star$ : All local martingales on the filtered probability space $(\Omega, \pazocal{F}, \pazocal{F}(\cdot), \mathbb{P})$ have continuous paths.

\smallskip

\begin{thm} [Optional Decomposition]		\label{thm : optional decomposition}
	Suppose that the market consisting of the top $n$ stocks is viable. For a nonnegative, adapted process $X$ with RCLL paths satisfying $X(0) = x \geq 0$, the following statements are equivalent:
	\begin{enumerate} [(1)]
		\item The process $YX$ is a supermartingale, for every $Y \in \pazocal{Y}^n$.
		\item There exist an investment strategy $\vartheta \in \pazocal{I}(S) \cap \pazocal{T}(n)$ among the top $n$ stocks, and a cumulative withdrawal process $K \in \pazocal{K}$, such that
		\begin{equation}		\label{Eq : optional decomposition}
			X = x+\int_0^{\cdot}\sum_{i=1}^N \vartheta_i(t)dS_i(t) - K.
		\end{equation}
	\end{enumerate}
\end{thm}

\smallskip

\begin{proof}
	We first show the implication \textit{(2)} $\Longrightarrow$ \textit{(1)}. For any $Y \in \pazocal{Y}^n$, we write $Y = \pazocal{E}(L)/X_{\rho}$ for some $L \in \pazocal{M}^{\perp}_{loc}(\widetilde{M})$ with $\Delta L > -1$ from Lemma~\ref{lem : deflator representation}, where we denote by $\rho$ the supermartingale num\'eraire portfolio among the top $n$ stocks. Then, we have from Lemma~\ref{lem : X / X rho},
	\begin{equation*}
		YX + YK = \frac{x+\int_0^{\cdot}\sum_{i=1}^N \vartheta_i(t)dS_i(t)}{X_{\rho}} \pazocal{E}(L)
		= \Big(x + \int_0^{\cdot} \sum_{i=1}^N \eta_i(t) d\widetilde{M}_i(t)\Big)\pazocal{E}(L),
	\end{equation*}
	for some process $\eta \in \pazocal{I}(\widetilde{M}) \cap \pazocal{T}(n)$. The last expression is a product of two nonnegative, orthogonal local martingales, thus it is a nonnegative local martingale. The claim that $YX$ is a supermartingale follows.
	
	\smallskip
	
	We now show the implication \textit{(1)} $\Longrightarrow$ \textit{(2)} which is more involved, under the above Assumption $\star$. We assume that \textit{(1)} holds and recall the collections $\pazocal{Y}^n$ and $\pazocal{M}^{\perp}_{loc}(\widetilde{M})$ of \eqref{Eq : local martingale deflator representation}. All processes in $\pazocal{M}^{\perp}_{loc}(\widetilde{M})$ have continuous paths under the Assumption $\star$. From Lemma~\ref{lem : deflator representation}, $(X/X_{\rho})\pazocal{E}(L)$ is a supermartingale for every $L \in \pazocal{M}^{\perp}_{loc}(\widetilde{M})$, and in particular, $X/X_{\rho}$ is a supermartingale itself. The Doob-Meyer and Kunita-Watanabe decompositions give
	\begin{equation*}
		\frac{X}{X_{\rho}} = x + M_{\eta} + L - B, \quad \text{where} \quad M_{\eta} := \int_0^{\cdot} \sum_{i=1}^N \eta_i(t) d\widetilde{M}_i(t).
	\end{equation*}
	Here, $\eta \equiv (\eta_1, \cdots, \eta_N) \in \pazocal{I}(\widetilde{M})$, $L \in \pazocal{M}^{\perp}_{loc}(\widetilde{M})$ and $B$ is an adapted, nondecreasing and right-continuous process with $B(0) = 0$, i.e., $B$ is a cumulative withdrawal process in $\pazocal{K}$. Recalling the diagonal matrix $D$ of \eqref{def : diagonal matrix D} with its property $Dd\widetilde{M}(t) = d\widetilde{M}(t)$, we further define $\widetilde{\eta} := D\eta \in \pazocal{I}(\widetilde{M}) \cap \pazocal{T}(n)$, and we have
	\begin{equation*}
		M_{\eta} = \int_0^{\cdot} \eta'(t) d\widetilde{M}(t) = \int_0^{\cdot} (D\eta)'(t) d\widetilde{M}(t) = M_{\widetilde{\eta}}.
	\end{equation*}
	Consequently, we obtain
	\begin{equation}		\label{Eq : Kunita-Watanabe decomp}
		\frac{X}{X_{\rho}} = x + M_{\widetilde{\eta}} + L - B, \quad \text{with} \quad \widetilde{\eta} \in \pazocal{I}(\widetilde{M}) \cap \pazocal{T}(n), ~~~L \in \pazocal{M}^{\perp}_{loc}(\widetilde{M}).
	\end{equation}
	
	\smallskip
	
	We next show that $L \equiv 0$ in \eqref{Eq : Kunita-Watanabe decomp}. Again from Lemma~\ref{lem : deflator representation}, $(1/X_{\rho})\pazocal{E}(mL)$ is a local martingale, thus $(X/X_{\rho})\pazocal{E}(mL)$ is a supermartingale for every $m \in \mathbb{N}$. Since $[\pazocal{E}(mL), \widetilde{M}_i] = 0$ for $i = 1, \cdots, N$, we have $[\pazocal{E}(mL), M_{\widetilde{\eta}}] = 0$ and consequently, $\pazocal{E}(mL)M_{\widetilde{\eta}}$ is a local martingale as a product of two orthogonal local martingales. Thus, from \eqref{Eq : Kunita-Watanabe decomp}, the process
	\begin{equation*}
		\pazocal{E}(mL)(L-B) = \pazocal{E}(mL)\frac{X}{X_{\rho}} - \pazocal{E}(mL)(x+M_{\widetilde{\eta}})
	\end{equation*}
	is a local supermartingale for every $m \in \mathbb{N}$. On the other hand, the integration by parts gives
	\begin{equation*}
		\pazocal{E}(mL)(L-B) = \int_0^{\cdot} (L-B)(t-)d\pazocal{E}(mL)(t) + \int_0^{\cdot} \pazocal{E}(mL)(t)dL(t) + \int_0^{\cdot} \pazocal{E}(mL)(t)d\big([mL, L]-B\big)(t).
	\end{equation*}
	Then, the last integrator $m[L, L]-B$ should be a local supermartingale for every $m \in \mathbb{N}$, which implies $[L, L] \equiv 0$, thus $L \equiv 0$.
	
	\smallskip
	
	As a result, the equation \eqref{Eq : Kunita-Watanabe decomp} becomes
	\begin{equation*}
		\frac{X}{X_{\rho}} = x + M_{\widetilde{\eta}} - B,
	\end{equation*}
	and we apply the product rule to obtain the decomposition of $X = X_{\rho}(X/X_{\rho})$:
	\begin{equation*}
		X = x + \int_0^{\cdot} X(t-)\rho'(t)dR(t) + \int_0^{\cdot} X_{\rho}(t)d(M_{\widetilde{\eta}}-B)(t) + \int_0^{\cdot} X_{\rho}(t)dC_{\widetilde{\eta}\rho}(t),
	\end{equation*}
	in conjunction with \eqref{Eq : dynamics of wealth} and \eqref{Eq : C pi rho}. Moreover, the condition $(\widetilde{3})$ of Proposition~\ref{prop : supermartingale numeraire among the top n} implies $C_{\widetilde{\eta}\rho} = A_{\widetilde{\eta}} = R_{\widetilde{\eta}} - M_{\widetilde{\eta}}$, and we deduce
	\begin{equation*}
		X = x + \int_0^{\cdot} \big( X(t-)\rho'(t) - X_{\rho}(t)\widetilde{\eta}'(t) \big) dR(t) - \int_0^{\cdot} X_{\rho}(t)dB(t).
	\end{equation*}
	Therefore, if we define
	\begin{equation*}
		\vartheta_i(t) := \frac{X(t-)\rho'(t) - X_{\rho}(t)\widetilde{\eta}'(t)}{S_i(t)}, \quad i = 1, \cdots, N, \qquad K := \int_0^{\cdot} X_{\rho}(t)dB(t),
	\end{equation*}
	then it is easy to check that $\vartheta \equiv (\vartheta_1, \cdots, \vartheta_N) \in \pazocal{I}(S) \cap \pazocal{T}(n)$ and $K \in \pazocal{K}$.		
\end{proof}

\smallskip

\begin{cor}		\label{cor : optional decomposition}
	Suppose that the market consisting of the top $n$ stocks is viable. For a nonnegative, adapted process $X$ with RCLL paths satisfying $X(0) = x \geq 0$, the following statements are then equivalent:
	\begin{enumerate} [(1)]
		\item The process $YX$ is a local martingale, for every $Y \in \pazocal{Y}^n$.
		\item There exists an investment strategy $\vartheta \in \pazocal{I}(S) \cap \pazocal{T}(n)$ among the top $n$ stocks, such that
		\begin{equation}
			X = x+\int_0^{\cdot}\sum_{i=1}^N \vartheta_i(t)dS_i(t).
		\end{equation}
	\end{enumerate}
\end{cor}

\smallskip

\begin{proof}
	We first assume \textit{(1)}; then $YX$ is a supermartingale for every $Y \in \pazocal{Y}^n$. From Theorem~\ref{thm : optional decomposition}, we have a decomposition \eqref{Eq : optional decomposition} for some $\vartheta \in \pazocal{I}(S) \cap \pazocal{T}(n)$ and $K \in \pazocal{K}$. In particular, if we take $Y = 1/X_{\rho}$, the reciprocal of the local martingale num\'eraire, we obtain
	\begin{equation*}
		YK = \frac{X(\cdot;x, \vartheta)}{X_{\rho}}-YX, 
	\end{equation*}
	with the notation in \eqref{def : investment}. Since the terms on the right-hand side are local martingales, $YK$ is a local martingale, and so is
	\begin{equation*}
		YK - \int_0^{\cdot}K(t-)dY(t) = \int_0^{\cdot} Y(t)dK(t).
	\end{equation*}
	However, the last integral is nondecreasing and is a supermartingale (as a non-negative local martingale), and therefore identically equal to zero. Thus, $K \equiv 0$ as $Y$ is positive, and the statement \textit{(2)} follows.
	
	\smallskip
	
	In order to show the reverse implication, we assume \textit{(2)}, then $X/X_{\rho}$ is a local martingale where $\rho$ is the local martingale num\'eraire portfolio among the top $n$ stocks, as before. From Lemma~\ref{lem : X / X rho}, $X/X_{\rho}$ can be cast as a stochastic integral with respect to the local martingale vector $\widetilde{M}$. Furthermore, from Lemma~\ref{lem : deflator representation}, every $Y \in \pazocal{Y}^n$ is of the form $Y = (1/X_{\rho})\pazocal{E}(L)$ for some local martingale $L$ satisfying $[L, \widetilde{M}_i]= 0$ for $i=1, \cdots, N$. Therefore, the product $YX = (X/X_{\rho})\pazocal{E}(L)$ of these two orthogonal local martingales is again a local martingale.
\end{proof}

\bigskip

\subsection{Entire market versus top $n$ market}

We present first the following result, which can be easily proven from the equivalence between the existence of supermartingale num\'eraire portfolio and the market viability.

\smallskip

\begin{thm}			\label{thm : superMG in whole market implies top market}
	The existence of a supermartingale num\'eraire portfolio in the whole market, implies the existence of supermartingale num\'eraire portfolio among the top $n$ stocks.
\end{thm}

\begin{proof}
	From Theorem 2.34 of \cite{KK2}, the existence of a supermartingale num\'eraire portfolio in the whole market, is equivalent to the viability of the whole market. The viability of the whole market implies the viability of the market consisting of the top $n$ stocks, thanks to the inequality $0 \leq x(K) \leq x^n(K)$ in Definition~\ref{Def : superheding}. We conclude that there exists a supermartingale num\'eraire portfolio among the top $n$ stocks, from Theorem~\ref{thm : main result}.
\end{proof}

\smallskip

Theorem~\ref{thm : superMG in whole market implies top market} shows that the viability of the entire market, composed of $N$ stocks, implies the viability of the `top $n$ market'. Thus, if the entire market is viable, there exist both a supermartingale num\'eraire portfolio for the whole market, and a supermartingale num\'eraire portfolio among the top $n$ stocks, and the former dominates the latter in the sense of growth-optimality. In the following proposition, we study this dominance by expressing the asymptotic behavior of log-relative wealth process between these two portfolios in terms of the `local characteristics' of the market. We first need the following definitions which are similar to those in \eqref{def : maximal growth rate}-\eqref{def : aggregate maximal growth}.

\smallskip

We call a $[0, \infty]$-valued, predictable process
\begin{equation}			\label{def : maximal growth rate of whole market}
g := \sup_{p \in \mathbb{R}^N} \Big( p' \alpha - \frac{1}{2}p'cp \Big)
\end{equation}
\textit{maximal growth rate} achievable in the whole market. This process can be rewritten in the form
\begin{equation}			\label{Eq : maximal growth rate of whole market}
g = \frac{1}{2}\big(\alpha'c^{\dagger}\alpha\big)\bm{1}_{\{\alpha \in \textbf{range}(c)\}} + \infty \bm{1}_{\{\alpha \notin \textbf{range}(c)\}},
\end{equation}
and the supremum of \eqref{def : maximal growth rate of whole market} is attained if and only if $g < \infty$, at $p \equiv \rho := c^{\dagger}\alpha$, i.e., when $\rho$ is the supermartingale num\'eraire portfolio of whole market. Here, $c^{\dagger}$ is the `pseudo-inverse' of $c$, defined as in \eqref{def : tilde c dagger}. Then, the viability of the whole market can be shown to be equivalent to the condition
\begin{equation}		\label{con : structural condition of whole market}
G(T) := \int_0^{T} g(t) dO(t) < \infty, \qquad \text{for all} \quad T \geq 0.
\end{equation}
Here, the adapted nondecreasing process $G$ is called as \textit{aggregate maximal growth of whole market}.

\smallskip

When the whole market is viable, the growth rates $\widetilde{g}$ of \eqref{Eq : maximal growth rate} and $g$ of \eqref{Eq : maximal growth rate of whole market} have simpler forms
\begin{equation}			\label{Eq : tilde g and g}
\widetilde{g} = \frac{1}{2}\widetilde{\alpha}'\widetilde{c}\,^{\dagger}\widetilde{\alpha} = \gamma_{\widetilde{\rho}}, \qquad
g = \frac{1}{2} \alpha' c^{\dagger} \alpha = \gamma_{\rho},
\end{equation}
respectively, as from \eqref{Eq : gamma of superMG portfolio}, with $\widetilde{\rho} := D\widetilde{c}\,^{\dagger}\widetilde{\alpha}$ the supermartingale num\'eraire portfolio among the top $n$ stocks, and $\rho := c^{\dagger}\alpha$ the supermartingale num\'eraire portfolio for the whole market. We denote the \textit{difference of aggregate maximal growth between the whole market and the top $n$ market} by
\begin{equation}		\label{def : diff of aggregate maximal growth}
\pazocal{G} := G - \widetilde{G}
= \int_0^{\cdot} \big( g(t)-\widetilde{g}(t) \big) dO(t) 
= \int_0^{\cdot} \big( \gamma_{\rho}(t)-\gamma_{\widetilde{\rho}}(t) \big) dO(t)
= \Gamma_{\rho} - \Gamma_{\widetilde{\rho}}.
\end{equation}
Since $\rho$ is also a growth-optimal portfolio (as a supermartingale num\'eraire portfolio) in the whole market, the relative growth $\Gamma^{\rho}_{\widetilde{\rho}} = \Gamma_{\widetilde{\rho}} - \Gamma_{\rho}$ of Definition~\ref{Def : relative growth} is non-increasing, from which we conclude that $\pazocal{G}$ is nondecreasing and nonnegative. 

\medskip

\begin{prop}
	Suppose that the whole market is viable and let $\rho$ and $\widetilde{\rho}$ be the supermartingale num\'eraire portfolio for the whole market and the supermartingale num\'eraire portfolio among the top $n$ stocks, respectively. Then, the asymptotic growth rate of the log-relative wealth process $\log X^{\rho}_{\widetilde{\rho}}$ is the same as $-\pazocal{G}$ of \eqref{def : diff of aggregate maximal growth}, namely:
	\begin{equation}		\label{Eq : asymptotic rate of log-relative wealth}
	\lim_{T \rightarrow \infty} \frac{1}{\pazocal{G}(T)}\log \bigg( \frac{X_{\widetilde{\rho}}(T)}{X_{\rho}(T)} \bigg) = -1 \quad \text{holds} \quad \mathbb{P}-\text{a.e. on the set} \quad \{ \lim_{T \rightarrow \infty} \pazocal{G}(T) = \infty \}.
	\end{equation}
\end{prop}

\begin{proof}
	We recall the notations \eqref{def : log of wealth}-\eqref{def : A pi, Gamma pi} and write for $T \geq 0$, 
	\begin{equation}		\label{Eq : relative log between two numeraires}
	\log X^{\rho}_{\widetilde{\rho}}(T) = \log X_{\widetilde{\rho}}(T) - \log X_{\rho}(T)
	= \int_0^{T} \big( \gamma_{\widetilde{\rho}}(t)-\gamma_{\rho}(t) \big) dO(t) + \int_0^{T} \big( \widetilde{\rho}(t)-\rho(t) \big)' dM(t).
	\end{equation}
	The first integral on the right-hand side is just $-\pazocal{G}(T)$ of \eqref{def : diff of aggregate maximal growth} and it can be rewritten as
	\begin{equation}			\label{Eq : - diff}
	-\pazocal{G}(T) = \int_0^{T} \big( \gamma_{\widetilde{\rho}}(t)-\gamma_{\rho}(t) \big) dO(t)
	= \frac{1}{2} \int_0^{T} \big( \widetilde{\alpha}'\widetilde{c}\,^{\dagger}\widetilde{\alpha} - \alpha'c^{\dagger}\alpha \big)(t) dO(t)
	\end{equation}
	from \eqref{Eq : tilde g and g}.
	On the other hand, from $\rho = c^{\dagger}\alpha$ and $\widetilde{\rho} = D\widetilde{c}\,^{\dagger}\widetilde{\alpha}$, we obtain series of equations as in \eqref{Eq : rho in terms of alpha}:
	\begin{equation*}
	(\widetilde{\rho}-\rho)'c(\widetilde{\rho}-\rho)
	= \widetilde{\rho}'c\widetilde{\rho} + \rho'c\rho - \widetilde{\rho}'c\rho - \rho'c\widetilde{\rho}
	= \widetilde{\alpha}'\widetilde{c}\,^{\dagger}\widetilde{\alpha} + \alpha'c^{\dagger}\alpha - 2\widetilde{\rho}'c\rho,
	\end{equation*}
	as well as
	\begin{equation*}
	\widetilde{\rho}'c\rho = \widetilde{\rho}'cc^{\dagger}{\alpha} = \widetilde{\rho}'\alpha 
	= \widetilde{\alpha}'(\widetilde{c}\,^{\dagger})'D\alpha = \widetilde{\alpha}'(\widetilde{c}\,^{\dagger})'\widetilde{\alpha}.
	\end{equation*}
	Combining these equations, we have
	\begin{equation*}
	(\widetilde{\rho}-\rho)'c(\widetilde{\rho}-\rho) = \alpha'c^{\dagger}\alpha - \widetilde{\alpha}'\widetilde{c}\,^{\dagger}\widetilde{\alpha}.
	\end{equation*}
	Thus, the quadratic variation of the last integral on the right hand-side of \eqref{Eq : relative log between two numeraires} is written as
	\begin{align*}
	\bigg[ \int_0^{T} \big( \widetilde{\rho}(t)-\rho(t) \big)' dM(t) \bigg]
	&= \int_0^{T} (\widetilde{\rho}-\rho)'c(\widetilde{\rho}-\rho)(t) dO(t)
	\\
	&= \int_0^{T} (\alpha'c^{\dagger}\alpha - \widetilde{\alpha}'\widetilde{c}\,^{\dagger}\widetilde{\alpha})(t) dO(t) = 2\pazocal{G}(T).
	\end{align*}
	The Dambis-Dubins-Schwarz representation (Theorem~3.4.6, Problem~3.4.7 of \cite{KS1}) with the scaling property of Brownian motion implies that there exists a Brownian motion $W$, on a possibly enlarged filtered probability space, such that
	\begin{equation}		\label{Eq : DDS representation2}
	\log X^{\rho}_{\widetilde{\rho}}(T) = -\pazocal{G}(T) + \sqrt{2}W\big( \pazocal{G}(T) \big).
	\end{equation}
	The strong law of large numbers for Brownian motion gives the result \eqref{Eq : asymptotic rate of log-relative wealth}.
\end{proof}

\smallskip

The expression \eqref{Eq : - diff} shows that the asymptotic growth rate of the log-relative wealth process $\log X^{\rho}_{\widetilde{\rho}}$ is expressed in terms of `local characteristics' of the market: $\alpha, \widetilde{\alpha}, c$, and $\widetilde{c}$.

\bigskip

\section{Stock Portfolios in Open Markets}
\label{sec: stock portfolio and FGP}

The open market described in the previous section consists of the top $n$ stocks and the money market. The existence of this money market gives us a flexibility to construct portfolios among top $n$ stocks. To be more specific, for any given portfolio $\pi \in \pazocal{I}(R)$, multiplying the diagonal matrix $D$ of \eqref{def : diagonal matrix D} transforms it into a new portfolio $D\pi$ among top $n$ stocks. The proportion of assets, which is supposed to be invested in `bottom' $N-n$ stocks by $\pi$, is now assigned to the money market by $D\pi$. In the absence of the money market, building portfolios among top $n$ stocks is more subtle, and this section focuses on these subtleties.

\subsection{Stock portfolios and the market portfolio}

An important subclass of portfolios in Definition~\ref{Def : portfolio} is the collection of portfolios $\pi$ satisfying $\sum_{i=1}^N \pi_i \equiv 1$, or $\pi_0 \equiv 0$ in \eqref{def : money proportion}. Such a portfolio never invests in the money market; and this condition can be formulated as $\pi \in \Delta^{N-1}$, where we denote
\begin{equation*}
	\Delta^{N-1} := \Big\{ (x_1, \cdots, x_N) \in \mathbb{R}^N ~\big|~ \sum_{i=1}^N x_i = 1 \Big\}.
\end{equation*}

\smallskip

\begin{defn} [Stock Portfolio]		\label{Def : stock portfolio}
	We call a portfolio $\pi \in \pazocal{I}(R)$ \textit{stock portfolio}, if it takes values in $\Delta^{N-1}$, i.e., satisfies $\sum_{i=1}^N \pi_i \equiv 1$. We denote the collection of stock portfolios by $\pazocal{I}(R) \cap \Delta^{N-1}$.
	
	\smallskip
	
	We call a stock portfolio $\pi$ \textit{stock portfolio among the top $n$ stocks}, if in addition it belongs to $\pazocal{T}(n)$, i.e., satisfies the condition \eqref{con : portfolio among n}, or equivalently, \eqref{con : equiv portfolio among n}. We denote the collection of stock portfolios among the top $n$ stocks by $\pazocal{I}(R) \cap \Delta^{N-1} \cap \pazocal{T}(n)$.
\end{defn}

\bigskip

\begin{rem}[Self-financibility of stock portfolios]
	For any stock portfolio $\pi$, we sum over \eqref{Eq : relationship2} for all indices $i = 1, \cdots, N$ to obtain
	\begin{equation*}
		1 \equiv \sum_{i=1}^N \pi_i(\cdot) = \frac{\sum_{i=1}^N S_i(\cdot)\vartheta_i(\cdot)}{X(\cdot; 1, \vartheta)},
	\end{equation*}
	and from \eqref{def : investment},
	\begin{equation*}
		X(\cdot; 1, \vartheta) 
		= 1 + \int_0^{\cdot} \sum_{i=1}^N \vartheta_i(t)dS_i(t) 
		= \sum_{i=1}^N \vartheta_i(\cdot)S_i(\cdot).
	\end{equation*}
	This last equation shows the `self-financing' property of the stock portfolios; (see Definition~2.1 of \cite{Karatzas:Ruf:2017}) the sum of product between the trading strategy $\vartheta_i$ and the stock price $S_i$ is equal to the sum of stochastic integrals of each trading strategy with respect to the corresponding stock price, along with the initial capital $1$, at any time $t \geq 0$. There are neither withdrawals nor infusions of capital from the money market; gains are re-invested, losses are absorbed.
\end{rem}

\bigskip

Before we present the most important example of stock portfolios, we introduce the notation
\begin{equation}		\label{def : Sigma}
	\Sigma := S_1 + \cdots + S_N,
\end{equation}
representing the total capitalization of whole equity market.

\smallskip

\begin{example} [Market portfolio]		\label{ex : market portfolio}
	Suppose that an investment strategy $\vartheta$ is given as $\vartheta \equiv \bm{1}/\Sigma(0) \equiv (1, 1, \cdots, 1)/\Sigma(0)$ with initial wealth $x = 1$. Then, its wealth process is just the total capitalization normalized by its initial value:
	\begin{equation}		\label{Eq : wealth of mu}
		X(\cdot; 1, \vartheta) = \frac{\Sigma(\cdot)}{\Sigma(0)}.
	\end{equation}
	Whereas, from \eqref{Eq : relationship2}, the corresponding portfolio $\pi \equiv \mu \equiv (\mu_1, \cdots, \mu_N)$ can be expressed as
	\begin{equation}		\label{def : market portfolio}
		\mu_i(\cdot) = \frac{S_i(\cdot)}{\Sigma(\cdot)} = \frac{S_i(\cdot)}{S_1(\cdot)+ \cdots +S_N(\cdot)}, \quad \text{for} \quad i = 1, \cdots, N.
	\end{equation}
	We call this special stock portfolio $\mu$ the \textit{market portfolio}, and its component processes in \eqref{def : market portfolio} \textit{market weights}; it is considered as the most important stock portfolio, as its wealth process gives the evolution of total market capitalization.
	
	\smallskip
	
	In an analogous manner, we define the \textit{top $n$ market portfolio}, which we denote by $\widetilde{\mu} \equiv (\widetilde{\mu}_1, \cdots, \widetilde{\mu}_N)$, with components
	\begin{equation}		\label{def : top n market portfolio}
		\widetilde{\mu}_i(\cdot) := \frac{\widetilde{S}_i(\cdot)}{\widetilde{\Sigma}(\cdot)}
		= \frac{\widetilde{S}_i(\cdot)}{\widetilde{S}_1(\cdot)+ \cdots +\widetilde{S}_N(\cdot)}, \quad \text{for} \quad i = 1, \cdots, N,
	\end{equation}
	where 
	\begin{equation}		\label{def : S tilde}
		\widetilde{\Sigma} := \sum_{i=1}^N \widetilde{S}_i = S_{(1)}+\cdots+S_{(n)}, \qquad \text{and} \qquad
		\widetilde{S}_i(\cdot) := \bm{1}_{\{u_i(\cdot) \leq n\}}S_i(\cdot), \qquad \text{for} \quad i = 1, \cdots, N.
	\end{equation}
	The denominator $\widetilde{\Sigma}$ of \eqref{def : top n market portfolio} represents the sum of the capitalizations of the top $n$ stocks; thus, $\widetilde{\mu}_i(t)$ is the proportion of the capitalization of stock $i$, if this stock belongs to the top $n$, to the total capitalization of the top $n$ stocks at time $t$. In other words, $\widetilde{\mu}_i$ can be interpreted as the `market weight' of $i$-th stock in the restricted market composed of the top $n$ stocks by capitalization. It is easy to check that $\widetilde{\mu}$ is a stock portfolio among the top $n$ stocks, i.e., $\widetilde{\mu} \in \pazocal{I}(R) \cap \Delta^{N-1} \cap \pazocal{T}(n)$.
\end{example}

\bigskip

\subsection{Capital Asset Pricing Model}

The Capital Asset Pricing Model assumes that individual stocks cannot systematically outperform the market. In our open market setting, this requirement can be cast as saying that each individual stock, whenever it belongs to the top $n$ stocks, cannot outperform the top $n$ market. In this subsection, we briefly discuss this model for the top $n$ market. Recalling the top $n$ stock portfolio $\widetilde{\mu}$ defined in \eqref{def : top n market portfolio}, we have the next definition.

\smallskip

\begin{defn} [CAPM]		\label{Def : CAPM}
	We say that the top $n$ market is in the realm of the \textit{Capital Asset Pricing Model~(CAPM)}, if
	\begin{equation}		\label{def : CAPM}
		\widetilde{R}_i = \int_0^{\cdot} \beta_i(t)dR_{\widetilde{\mu}}(t) + N_i, \qquad i = 1, \cdots, N,
	\end{equation}
	hold for appropriate processes $\beta_i \in \pazocal{I}(R_{\widetilde{\mu}})$, $i = 1, \cdots, N$, and for continuous local martingales $N_i$ with $N_i(0) = 0$ which are orthogonal to $R_{\widetilde{\mu}}$ for all $i = 1, \cdots, N$:
	\begin{equation*}
		[N_i, R_{\widetilde{\mu}}] \equiv 0.
	\end{equation*}
\end{defn}

\smallskip

The following proposition characterizes this property, in terms of the local characteristics of the top market introduced in Section~\ref{subsec : portfolio among the top n stocks}.

\smallskip

\begin{prop} [Characterization of CAPM]		\label{prop : CAPM}
	The top $n$ market is in the realm of the CAPM if, and only if, the following two conditions hold.
	\begin{enumerate} [(A)]
		\item There exists a scalar ``leverage'' predictable process $b$ such that
		\begin{equation}		\label{con : A-1}
			\sum_{i=1}^N \int_0^T |b(t)|\bm{1}_{\{c_{\widetilde{\mu}\widetilde{\mu}}>0\}}|dC\,_{\widetilde{i}{\widetilde{\mu}}}(t)| < \infty, \qquad \text{for} \quad T \geq 0,
		\end{equation}
		and the equalities hold $(\mathbb{P}\otimes O)$-a.e.:
		\begin{equation}		\label{con : A-2}
			\widetilde{\alpha}_i = bc\,_{\widetilde{i}{\widetilde{\mu}}}, \quad \text{on} \quad \{ c_{\widetilde{\mu}\widetilde{\mu}} > 0 \} \quad \text{for} \quad i = 1, \cdots, N.
		\end{equation}

		\item On the set $\{ c_{\widetilde{\mu}\widetilde{\mu}} = 0 \}$, we have $(\mathbb{P}\otimes O)$-a.e.:
		\begin{equation}
			\alpha_{\widetilde{\mu}} = 0 \quad \Longleftrightarrow \quad \widetilde{\alpha}_i = 0, \qquad i = 1, \cdots, N.
		\end{equation}
	\end{enumerate}
	When these conditions are satisfied, the process $b$ of \eqref{con : A-1} and the processes $\beta_i \in \pazocal{I}(R_{\widetilde{\mu}})$ of \eqref{def : CAPM} can be chosen, respectively, as
	\begin{equation}		\label{def : b}
		b = \frac{\alpha_{\widetilde{\mu}}}{c_{\widetilde{\mu}\widetilde{\mu}}}\bm{1}_{\{c_{\widetilde{\mu}\widetilde{\mu}}>0\}},
	\end{equation}
	\begin{equation}		\label{def : beta}
		\beta_i = \frac{c\,_{\widetilde{i}\widetilde{\mu}}}{c_{\widetilde{\mu}\widetilde{\mu}}}\bm{1}_{\{c_{\widetilde{\mu}\widetilde{\mu}}>0\}}+\frac{\widetilde{\alpha}_i}{\alpha_{\widetilde{\mu}}}\bm{1}_{\{c_{\widetilde{\mu}\widetilde{\mu}}=0,~ \alpha_{\widetilde{\mu}} \neq 0\}}, \qquad i = 1, \cdots, N.
	\end{equation}
\end{prop}

\smallskip

\begin{proof}
	Let us assume first that the top $n$ market is in the realm of the CAPM. Recalling the notation \eqref{Eq : C tilde i rho}, we have
	\begin{equation*}
		C\,_{\widetilde{i}{\widetilde{\mu}}} = [\widetilde{R}_i, R_{\widetilde{\mu}}] = \int_0^{\cdot} \beta_i(t)d[R_{\widetilde{\mu}}, R_{\widetilde{\mu}}](t) + [N_i, R_{\widetilde{\mu}}] = \int_0^{\cdot} \beta_i(t) dC_{\widetilde{\mu}{\widetilde{\mu}}}(t),
	\end{equation*}
	which implies that $c\,_{\widetilde{i}\widetilde{\mu}} = \beta_i c_{\widetilde{\mu}\widetilde{\mu}}$ also hold $(\mathbb{P} \otimes O)$-a.e., for $i = 1, \cdots, N$. On $\{ c_{\widetilde{\mu}\widetilde{\mu}} > 0 \}$, it follows that $\beta_i = c\,_{\widetilde{i}\widetilde{\mu}}/c_{\widetilde{\mu}\widetilde{\mu}}$ for $i=1, \cdots, N$.	Moreover, since $\widetilde{R}_i - \int_0^{\cdot} \beta_i(t)dR_{\widetilde{\mu}}(t)$ is a local martingale, we obtain $\widetilde{A}_i = \int_0^{\cdot} \beta_i(t)dA_{\widetilde{\mu}}(t)$, and also $\widetilde{\alpha}_i = \beta_i \alpha_{\widetilde{\mu}}$ holds $(\mathbb{P}\otimes O)$-a.e. for $i = 1, \cdots, N$. As a consequence, the identities of \eqref{con : A-2}
	\begin{equation*}
		\widetilde{\alpha}_i = \frac{\alpha_{\widetilde{\mu}}}{c_{\widetilde{\mu}\widetilde{\mu}}}c\,_{\widetilde{i}\widetilde{\mu}} = bc\,_{\widetilde{i}\widetilde{\mu}}, \qquad \text{hold for} \quad i = 1, \cdots, N, \quad (\mathbb{P}\otimes O)-\text{a.e. on} \quad \{c_{\widetilde{\mu}\widetilde{\mu}}>0\},
	\end{equation*}
	with $b$ given as in \eqref{def : b}. Also, the $(\mathbb{P}\otimes O)$-a.e. identities $\widetilde{\alpha}_i = \beta_i \alpha_{\widetilde{\mu}}$, combined with $\alpha_{\widetilde{\mu}}=\widetilde{\mu}'\alpha$, lead to the condition (B). Finally, $b = \widetilde{\alpha}_i/c\,_{\widetilde{i}\widetilde{\mu}}$ on $\{ c_{\widetilde{\mu}\widetilde{\mu}} > 0, ~ c\,_{\widetilde{i}\widetilde{\mu}} \neq 0 \}$ implies that $|b|\bm{1}_{\{c_{\widetilde{\mu}\widetilde{\mu}}>0\}}|c\,_{\widetilde{i}\widetilde{\mu}}| \leq |\widetilde{\alpha}_i|$ hold for $i=1, \cdots, N$, and thus the condition \eqref{con : A-1}:
	\begin{equation*}
		\sum_{i=1}^N \int_0^T |b(t)|\bm{1}_{\{c_{\widetilde{\mu}\widetilde{\mu}}>0\}}|dC\,_{\widetilde{i}\widetilde{\mu}}(t)|
		\leq \sum_{i=1}^N \int_0^T |d\widetilde{A}_i(t)| < \infty, \qquad \text{for all} \quad T \geq 0.
	\end{equation*}
	
	\smallskip
	
	Conversely, suppose that the conditions (A) and (B) are valid. For $i = 1, \cdots, N$, defining $\beta_i$ via \eqref{def : beta}, we have
	\begin{equation*}
		\int_0^T |\beta_i(t)||dA_{\widetilde{\mu}}(t)| \leq \int_0^T |b(t)|\bm{1}_{\{ c_{\widetilde{\mu}\widetilde{\mu}}>0\}}|dC\,_{\widetilde{i}\widetilde{\mu}}(t)| + \int_0^T |d\widetilde{A}_i(t)| < \infty,
	\end{equation*}
	as well as
	\begin{equation*}
		\int_0^T |\beta_i(t)|^2 dC_{\widetilde{\mu}\widetilde{\mu}}(t)
		= \int_0^T \frac{|c\,_{\widetilde{i}\widetilde{\mu}}(t)|^2}{c_{\widetilde{\mu}\widetilde{\mu}}(t)} \bm{1}_{\{c_{\widetilde{\mu}\widetilde{\mu}}>0\}} dO(t)
		\leq \int_0^T c\,_{\widetilde{i}, \widetilde{i}}(t)dO(t) = \widetilde{C}_{i,i}(T) < \infty.
	\end{equation*}
	These inequalities imply that $\beta_i \in \pazocal{I}(R_{\widetilde{\mu}})$ for $i=1, \cdots, N$. Furthermore, recalling the semimartingale decomposition \eqref{Eq : decomposition of tilde R}, we observe that
	\begin{align}
		\int_0^{\cdot} \beta_i(t) dR_{\widetilde{\mu}}(t)							\nonumber
		&= \int_0^{\cdot} \beta_i(t) \widetilde{\mu}'(t) d\widetilde{A}(t) + \int_0^{\cdot} \beta_i(t) \widetilde{\mu}'(t) d\widetilde{M}(t)
		\\
		&= \int_0^{\cdot} \beta_i(t) \bm{1}_{\{c_{\widetilde{\mu}\widetilde{\mu}}(t)>0\}} b(t) dC_{\widetilde{\mu}{\widetilde{\mu}}}(t)
		+ \int_0^{\cdot} \beta_i(t) \bm{1}_{\{c_{\widetilde{\mu}\widetilde{\mu}}(t)=0\}} \widetilde{\mu}'(t) d\widetilde{A}(t)															\label{Eq : beta integrals}
		\\
		&~~~+ \int_0^{\cdot} \beta_i(t) \widetilde{\mu}'(t) d\widetilde{M}(t),			\nonumber
	\end{align}
	from \eqref{def : b}. The first two integrals on the right hand side of \eqref{Eq : beta integrals} can be expressed as
	\begin{equation*}
		\int_0^{\cdot} \beta_i(t) \bm{1}_{\{c_{\widetilde{\mu}\widetilde{\mu}}(t)>0\}} b(t) dC_{\widetilde{\mu}{\widetilde{\mu}}}(t)
		= \int_0^{\cdot} \bm{1}_{\{c_{\widetilde{\mu}\widetilde{\mu}}(t)>0\}} b(t) dC\,_{\widetilde{i}{\widetilde{\mu}}}(t)
		= \int_0^{\cdot} \bm{1}_{\{c_{\widetilde{\mu}\widetilde{\mu}}(t)>0\}} d\widetilde{A}_i(t),
	\end{equation*}
	and
	\begin{equation*}
		\int_0^{\cdot} \beta_i(t) \bm{1}_{\{c_{\widetilde{\mu}\widetilde{\mu}}(t)=0\}} \widetilde{\mu}'(t) d\widetilde{A}(t)
		= \int_0^{\cdot} \bm{1}_{\{c_{\widetilde{\mu}\widetilde{\mu}}(t)=0\}} d\widetilde{A}_i(t),
	\end{equation*}
	for $i=1, \cdots, N$, on account of \eqref{def : beta}. Thus, we obtain
	\begin{equation*}
		\int_0^{\cdot} \beta_i(t) dR_{\widetilde{\mu}}(t)
		= \widetilde{A}_i+\int_0^{\cdot} \beta_i(t) \widetilde{\mu}'(t) d\widetilde{M}(t)
		= \widetilde{R}_i - \int_0^{\cdot} \big(e^i - \beta_i(t) \widetilde{\mu}(t) \big)' d\widetilde{M}(t) =: \widetilde{R}_i - N_i,
	\end{equation*}
	which is \eqref{def : CAPM}, where we define $N_i = \int_0^{\cdot} \big(e^i - \beta_i(t) \widetilde{\mu}(t) \big)' d\widetilde{M}(t)$ for $i = 1, \cdots, N$. We observe that the identities $(e^i-\beta_i\widetilde{\mu})'\widetilde{c}\widetilde{\mu} = c\,_{\widetilde{i}\widetilde{\mu}}-\beta_ic_{\widetilde{\mu}\widetilde{\mu}} = 0$ hold on the set $\{ c_{\widetilde{\mu}\widetilde{\mu}} > 0 \}$ from the definition \eqref{def : beta}, as well as on the set $\{ c_{\widetilde{\mu}\widetilde{\mu}} = 0 \}$ since $c\,_{\widetilde{i}\widetilde{\mu}} = 0$ holds there. Finally, we obtain
	\begin{equation*}
		[N_i, R_{\widetilde{\mu}}] = \int_0^{\cdot} \big(e^i - \beta_i(t) \widetilde{\mu}(t) \big)' \widetilde{c}(t)\widetilde{\mu}(t) dO(t) \equiv 0, \qquad i = 1, \cdots, N,
	\end{equation*}
	which shows that the top $n$ market is in the realm of the CAPM.
\end{proof}

\bigskip

\subsection{Functional generation of portfolios}

Functionally generated portfolios were first introduced by \cite{F_generating}. Given a function $G: \Delta_+^{N-1} \rightarrow (0, \infty)$ of class $C^2$ with the notation
\begin{equation}		\label{def : simplex}
	\Delta_{+}^{N-1} := \Delta^{N-1} \cap \mathbb{R}^N_+ = \Big\{ (x_1, \cdots, x_N) \in \mathbb{R}^N ~\big|~ x_i \geq 0 ~~ \text{for} ~~ i = 1, \cdots, N, ~~ \sum_{i=1}^N x_i = 1 \Big\},
\end{equation}
we can generate a portfolio $\pi^G$ from $G$, depending on the vector of market weights $\mu$. The formula (11.2) of \cite{FK_survey}, colloquially known as the `master formula', gives a simply way to compare the relative wealth process of $\pi^G$ with respect to the `market', namely, the market portfolio $\mu$ (see, Chapter III of \cite{FK_survey} for an overview).

\smallskip

In what follows, we present a new way to generate portfolios from a function having the market portfolio among the top $n$ stocks $\widetilde{\mu}$ in \eqref{def : top n market portfolio} as its input. We also derive a new `master formula' to compare the wealth of the so-generated portfolio, relative to $\widetilde{\mu}$, the market portfolio among the top $n$ stocks.

\smallskip

For any stock portfolio $\pi \in \pazocal{I}(R) \cap \Delta^{N-1}$, we have $\pi_0 \equiv 0$~(no investing in the money market), thus
\begin{equation}		\label{Eq : relative wealth from prop}
	\frac{dX^{\widetilde{\mu}}_{\pi}(t)}{X^{\widetilde{\mu}}_{\pi}(t)}
	= \sum_{i=1}^N \pi_i(t) dR^{\widetilde{\mu}}_i(t)
	= \sum_{i=1}^N \frac{\pi_i(t)}{S^{\widetilde{\mu}}_i(t)}dS^{\widetilde{\mu}}_i(t),
\end{equation}
from \eqref{Eq : dynamics of relative wealth}. Here, we recall from \eqref{def : denominated stock price} that $S^{\widetilde{\mu}}$ is the vector of stock prices denominated by the wealth process $X^{\widetilde{\mu}}$ of the market portfolio among the top $n$ stocks.

\smallskip

We note at this point, that the market portfolio $\mu$ in Example~\ref{ex : market portfolio}, has a very nice property: the denominated stock price $S^{\mu}_i$ has a simple representation, namely $S^{\mu}_i(\cdot) = \Sigma(0)\mu_i(\cdot)$, for $i = 1, \cdots, N$. Thus, if we used $\mu$ instead of $\widetilde{\mu}$ in deriving \eqref{Eq : relative wealth from prop}, the last integrator would be rewritten as $dS^{\mu}_i(t) = \Sigma(0)d\mu_i(t)$. However, unlike $\mu$, the components of $\widetilde{\mu}$ in \eqref{def : top n market portfolio} do not admit such a simple representation. For this reason, we will use the denominated stock price $S^{\widetilde{\mu}}_i(t)$ as integrators, and will let the generating function $G$ depend on $S^{\widetilde{\mu}}_i(t)$ instead of $\widetilde{\mu}_i(t)$ in what follows. 

\smallskip

For a given function $G : (0, \infty)^N \rightarrow (0, \infty)$ of class $C^2$, we want to write the relative-log wealth as
\begin{equation}
	\log X^{\widetilde{\mu}}_{\pi}(t) = \log \bigg( \frac{G\big(S^{\widetilde{\mu}}(t)\big)}{G\big(S^{\widetilde{\mu}}(0)\big)} \bigg) + J^{\widetilde{\mu}}_{\pi}(t), \qquad \text{for any} \quad t \geq 0,
\end{equation}
for some function $J^{\widetilde{\mu}}_{\pi}(\cdot)$ of finite variation. In order to find $J^{\widetilde{\mu}}_{\pi}(\cdot)$, we apply It\^o's rule, to obtain
\begin{equation}		\label{Eq : relative wealth from Ito}
	\frac{dX^{\widetilde{\mu}}_{\pi}(t)}{X^{\widetilde{\mu}}_{\pi}(t)} = d J^{\widetilde{\mu}}_{\pi}(t) + \sum_{i=1}^N \frac{D_iG\big(S^{\widetilde{\mu}}(t)\big)}{G\big(S^{\widetilde{\mu}}(t)\big)}dS_i^{\widetilde{\mu}}(t) + \frac{1}{2}\sum_{i=1}^N \sum_{j=1}^N \frac{D_{i, j}^2 G\big(S^{\widetilde{\mu}}(t)\big)}{G\big(S^{\widetilde{\mu}}(t)\big)}d[ S^{\widetilde{\mu}}_i, S^{\widetilde{\mu}}_j](t).
\end{equation}
Comparing the two equations \eqref{Eq : relative wealth from prop} and \eqref{Eq : relative wealth from Ito}, suppose we can find a portfolio $\pi$ such that
\begin{equation}		\label{Eq : pi condition}
	\sum_{i=1}^N \frac{\pi_i(t)}{S^{\widetilde{\mu}}_i(t)}dS^{\widetilde{\mu}}_i(t) = \sum_{i=1}^N \frac{D_iG\big(S^{\widetilde{\mu}}(t)\big)}{G\big(S^{\widetilde{\mu}}(t)\big)}dS_i^{\widetilde{\mu}}(t),
\end{equation}
holds, then we have
\begin{equation*}
	d J^{\widetilde{\mu}}_{\pi}(t) = -\frac{1}{2}\sum_{i=1}^N \sum_{j=1}^N \frac{D_{i, j}^2 G\big(S^{\widetilde{\mu}}(t)\big)}{G\big(S^{\widetilde{\mu}}(t)\big)}d[ S^{\widetilde{\mu}}_i, S^{\widetilde{\mu}}_j](t).
\end{equation*}
Now, a candidate portfolio $\pi$ satisfying \eqref{Eq : pi condition}, is given as
\begin{equation*}
	\pi_i (t) = S^{\widetilde{\mu}}_i(t)\frac{D_iG\big(S^{\widetilde{\mu}}(t)\big)}{G\big(S^{\widetilde{\mu}}(t)\big)}, \qquad i = 1, \cdots, N;
\end{equation*}
but it need not belong to $\Delta^{N-1}$. Instead, we set
\begin{equation}		\label{def : FGP}
	\pi_i(t) :=
	S^{\widetilde{\mu}}_i(t)\frac{D_iG\big(S^{\widetilde{\mu}}(t)\big)}{G\big(S^{\widetilde{\mu}}(t)\big)}
	+\widetilde{\mu}_i(t) - \widetilde{\mu}_i(t)\sum_{j=1}^N S^{\widetilde{\mu}}_j(t)\frac{D_jG\big(S^{\widetilde{\mu}}(t)\big)}{G\big(S^{\widetilde{\mu}}(t)\big)}, \qquad i = 1, \cdots, N,
\end{equation}
then it is easy to show that $\pi \in \Delta^{N-1}$. To check the condition \eqref{Eq : pi condition}, we note
\begin{equation*}
	\sum_{i=1}^N \frac{\pi_i(t)}{S^{\widetilde{\mu}}_i(t)}dS^{\widetilde{\mu}}_i(t)
	= \sum_{i=1}^N \frac{D_iG\big(S^{\widetilde{\mu}}(t)\big)}{G\big(S^{\widetilde{\mu}}(t)\big)}dS_i^{\widetilde{\mu}}(t)
	+ \Big\{ 1-\sum_{j=1}^N S^{\widetilde{\mu}}_j(t)\frac{D_jG\big(S^{\widetilde{\mu}}(t)\big)}{G\big(S^{\widetilde{\mu}}(t)\big)}\Big\} \sum_{i=1}^N \frac{\widetilde{\mu}_i(t)}{S^{\widetilde{\mu}}_i(t)}dS^{\widetilde{\mu}}_i(t),
\end{equation*}
and the last term vanishes because
\begin{equation*}
	\sum_{i=1}^N \frac{\widetilde{\mu}_i(t)}{S^{\widetilde{\mu}}_i(t)}dS^{\widetilde{\mu}}_i(t) = \sum_{i=1}^N \widetilde{\mu}_i(t)dR^{\widetilde{\mu}}_i(t) = dR^{\widetilde{\mu}}_{\widetilde{\mu}}(t) = 0.
\end{equation*}
Here, $R^{\widetilde{\mu}}_{\widetilde{\mu}} = \pazocal{L}(X^{\widetilde{\mu}}_{\widetilde{\mu}}) = \pazocal{L}(1) \equiv 0$, from Proposition~\ref{prop : relative wealth}.

\smallskip

The construction described above can be formulated as the following definition and proposition.

\smallskip

\begin{defn} [Functionally generated portfolio] \label{Def : FGP}
	Let $G : (0, \infty)^N \rightarrow (0, \infty)$ be a twice continuously differentiable function. Then, the vector $\pi^G \equiv \pi = (\pi_1, \cdots, \pi_N)$ defined as in \eqref{def : FGP} is called the \textit{stock portfolio generated by the function $G$ via the market portfolio among the top $n$ stocks}.
\end{defn}

\smallskip

\begin{prop} [Master Formula] \label{prop : master formula}
	For the stock portfolio $\pi^G$ generated by $G$ with the market portfolio among the top $n$ stocks, we have the decomposition
	\begin{equation}		\label{Eq : master formula}
		\log \bigg(\frac{X_{\pi^G}}{X_{\widetilde{\mu}}}\bigg) = \log \bigg(\frac{G \big(S^{\widetilde{\mu}}\big)}{G \big(S^{\widetilde{\mu}}(0)\big)}\bigg) -\frac{1}{2}\sum_{i=1}^N \sum_{j=1}^N \int_0^{\cdot}\frac{D_{i, j}^2 G\big(S^{\widetilde{\mu}}(t)\big)}{G\big(S^{\widetilde{\mu}}(t)\big)}d[ S^{\widetilde{\mu}}_i, S^{\widetilde{\mu}}_j](t).
	\end{equation}
\end{prop}

\smallskip

The above arguments, leading to Definition~\ref{Def : FGP} and Proposition~\ref{prop : master formula}, have two weaknesses. First, the functionally-generated stock portfolio $\pi^G$ in \eqref{def : FGP} is not generally a portfolio among the top $n$ stocks, i.e., it can fail to belong to $\pazocal{T}(n)$. Thus, the Master formula \eqref{Eq : master formula} compares a portfolio $\pi^G$ which is not a portfolio among the top $n$ stocks, with $\widetilde{\mu}$, which is a portfolio among the top $n$ stocks. We will fix this issue by restricting the class of generating functions $G$ in the next subsection.

\smallskip

Secondly, when we construct a portfolio via \eqref{def : FGP} or use the Master formula \eqref{Eq : master formula}, we need to know at each time $t \geq 0$ the entire history of the process $X_{\widetilde{\mu}}$, up to time $t$, because these equations require the values of the vector $S^{\widetilde{\mu}}=S/X_{\widetilde{\mu}}$. This issue is unfortunately inevitable in the open market, because of its own nature of $\widetilde{\mu}$; as it is composed of the top $n$ stocks, we need to keep track of the ranks of $N$ stocks all the time, whereas computing the wealth $X_{\mu}$ generated by the market portfolio $\mu$ only requires \textit{current} stock prices (and the stock prices at time $t=0$), from its simple representation \eqref{Eq : wealth of mu}. Though we cannot resolve this second issue, we will give a representation of $X_{\widetilde{\mu}}$ in the following subsection.

\bigskip

\subsection{Functionally generated portfolios using ranks}			\label{subsec : FGP by ranks}

Recalling the rank notation in Definition~\ref{Def : price process by rank}, we define the random permutation process $p_k(t)$ of $\{1, \cdots, N\}$ such that for $k = 1, \cdots, N$, 
\begin{equation}			\label{def : permutation p}
	S_{p_k(t)}(t) = S_{(k)}(t),
\end{equation}
\begin{equation*}
	p_k(t) < p_{k+1}(t) \quad \text{if} \quad S_{(k)}(t) = S_{(k+1)}(t).
\end{equation*}
$p_k(t)$ represents the index name of the stock at rank $k$ at time $t$, breaking ties with the lexicographic rule, so it is the inverse permutation of $u_i(t)$, introduced in \eqref{def : u}:
$u_i(t) = k \Longleftrightarrow p_k(t)= i$, for all $t \geq 0$.

\smallskip

For any continuous semimartingale $Y$, we denote the local time accumulated at the origin by $Y(\cdot)$ up to time $t \geq 0$ by $L^Y(t)$;
\begin{equation*}
	L^Y(t) := \frac{1}{2} \bigg(|Y(t)|-|Y(0)|-\int_0^t \text{sign}\big(Y(s)\big)dY(s) \bigg), \quad \text{where} \quad \textrm{sign}(x) = 2 \times \bm{1}_{(0, \infty)}(x)-1.
\end{equation*}
Then, $L^{S_{(k)}-S_{(\ell)}}(t)$ can be interpreted as the `collision local time' accumulated up to time $t$, whenever the $k$-th and $\ell$-th ranked processes of $S$ collide. In order to simplify the local time terms throughout this section, we introduce the following definition which prohibits the accumulation of local times of `triple collisions' between the stock prices.

\smallskip

\begin{defn}
	The components of the price vector $S = (S_1, \cdots, S_N)$ in Definition~\ref{Def : price process by name} are called \textit{pathwise mutually nondegenerate}, if
	\begin{enumerate} [(i)]
		\item the set $\{ t : S_i(t) = S_j(t) \}$ has Lebesgue measure zero, $\mathbb{P}$-a.e., for all $i \neq j$; and if
		\item $L^{S_{(k)}-S_{(\ell)}}(t) \equiv 0$ holds $\mathbb{P}$-a.e., for $|k-\ell| \geq 2$.
	\end{enumerate}
\end{defn}

\smallskip

\begin{prop}
	Suppose that the components of the price vector $S$ are pathwise mutually nondegenerate. Then, with the notation \eqref{def : S tilde}, the wealth process $X_{\widetilde{\mu}}$ of $\widetilde{\mu}$ admits the representation
	\begin{equation}		\label{Eq : wealth of tilde mu}
		X_{\widetilde{\mu}}(\cdot) = \frac{\widetilde{\Sigma}(\cdot)}{\widetilde{\Sigma}(0)} \exp \bigg( -\frac{1}{2} \int_0^{\cdot} \frac{1}{\widetilde{\Sigma}(t)}dL^{S_{(n)}-S_{(n+1)}}(t) \bigg).
	\end{equation}
\end{prop}

\smallskip

\begin{proof}
	From Proposition~\ref{prop : relative wealth} and the fact that $\widetilde{\mu}$ is a stock portfolio, we have
	\begin{equation*}
		X_{\widetilde{\mu}}(\cdot) 
		= \pazocal{E} \Big(\int_0^{\cdot} \sum_{i=1}^N \widetilde{\mu}_i(t)dR_i(t) \Big)
		= \pazocal{E} \Big(\int_0^{\cdot} \sum_{i=1}^N \frac{\widetilde{S}_i(t)}{\widetilde{\Sigma}(t)S_i(t)} dS_i(t) \Big)
		= \pazocal{E} \Big(\int_0^{\cdot} \sum_{i=1}^N \sum_{k=1}^n \frac{\bm{1}_{\{u_i(t) = k\}}}{\widetilde{\Sigma}(t)} dS_i(t) \Big).
	\end{equation*}
	On the other hand, from Proposition~4.1.11 of \cite{Fe}, we have
	\begin{equation*}
		\sum_{i=1}^N \bm{1}_{\{u_i(t) = k\}} dS_i(t) = dS_{(k)}(t) - \frac{1}{2} dL^{S_{(k)}-S_{(k+1)}}(t) + \frac{1}{2} dL^{S_{(k-1)}-S_{(k)}}(t),
	\end{equation*}
	for $k = 1, \cdots, N$ and $t \geq 0$, with the conventions $L^{S_{(0)}-S_{(1)}} \equiv 0$ and $L^{S_{(N)}-S_{(N+1)}} \equiv 0$. Thus, we obtain
	\begin{align*}
		X_{\widetilde{\mu}}(\cdot)\
		&=\pazocal{E} \Big(\int_0^{\cdot} \sum_{k=1}^n \frac{dS_{(k)}(t)}{\widetilde{\Sigma}(t)} + \frac{1}{2} \int_0^{\cdot} \sum_{k=1}^n \frac{dL^{S_{(k-1)}-S_{(k)}}(t)}{\widetilde{\Sigma}(t)} - \frac{1}{2} \int_0^{\cdot} \sum_{k=1}^n \frac{dL^{S_{(k)}-S_{(k+1)}}(t)}{\widetilde{\Sigma}(t)} \Big) 
		\\
		&=\pazocal{E} \Big(\int_0^{\cdot} \frac{d\widetilde{\Sigma}(t)}{\widetilde{\Sigma}(t)} - \frac{1}{2} \int_0^{\cdot} \frac{dL^{S_{(n)}-S_{(n+1)}}(t)}{\widetilde{\Sigma}(t)} \Big)
		= \frac{\widetilde{\Sigma}(\cdot)}{\widetilde{\Sigma}(0)} \exp \bigg( -\frac{1}{2} \int_0^{\cdot} \frac{dL^{S_{(n)}-S_{(n+1)}}(t)}{\widetilde{\Sigma}(t)} \bigg).
	\end{align*}
\end{proof}

\smallskip

The exponential term of \eqref{Eq : wealth of tilde mu} shows the `leakage', the effect caused by stocks which cross over from the top $n$ league to the bottom. Due to this effect, we need to keep track of the collision local time $L^{S_{(n)}-S_{(n+1)}}$ in order to compute $X_{\widetilde{\mu}}$, as we pointed out at the end of the previous subsection.

\smallskip

We next present Fernholz's original method of constructing rank-dependent portfolios from generating functions. We write $\mu_{(k)}$ to represent the $k$-th ranked market weight among $\mu_1, \cdots, \mu_N$ for $k=1, \cdots, N$, and introduce the vector $\bm{\mu} = (\mu_{(1)}, \cdots, \mu_{(N)})$ with components $\mu_{(k)} = S_{(k)}/\Sigma$, $k=1, \cdots, N$, as in \eqref{Eq : rank}, \eqref{def : Sigma}. The following result is based on Theorem~4.2.1 of \cite{Fe}.

\smallskip

\begin{thm}	[Functionally generated portfolios using ranked market weights]		\label{Thm : FGP rank}
	Suppose that the price vector $S$ is pathwise mutually nondegenerate. Let $p_k(\cdot)$, $k = 1, \cdots, N$ be the random permutation process defined by \eqref{def : permutation p} and let $\textbf{G}$ be a function defined on a neighborhood $U$ of $\Delta^{N-1}_+$. Suppose that there exists a positive $C^2$ function $G$ such that for $(x_1, \cdots, x_N) \in U$,
	\begin{equation}		\label{Eq : bold G}
		\textbf{G}(x_1, \cdots, x_N) = G(x_{(1)}, \cdots, x_{(N)}).
	\end{equation}
	Then $\textbf{G}$ generates the stock portfolio $\pi^{\textbf{G}}$ such that for $k=1, \cdots, N$,
	\begin{equation}		\label{def : pi^G}
		\pi^{\textbf{G}}_{p_k(t)}(t) = \bigg( \frac{D_kG\big(\bm{\mu}(t)\big)}{G\big(\bm{\mu}(t)\big)} + 1 - \sum_{\ell=1}^N \mu_{(\ell)} \frac{D_{\ell}G\big(\bm{\mu}(t)\big)}{G\big(\bm{\mu}(t)\big)} \bigg) \mu_{(k)}(t), \qquad \text{for} \quad t \geq 0.
	\end{equation}
	The log-relative wealth process of $\pi^{\textbf{G}}$ with respect to the market portfolio $\mu$, can be expressed via the `master formula' :
	\begin{align*}
		\log \bigg( \frac{X_{\pi^{\textbf{G}}}}{X_{\mu}} \bigg) = \log \bigg( \frac{G(\bm{\mu})}{G\big(\bm{\mu}(0)\big)} \bigg) &- \frac{1}{2} \int_0^{\cdot} \sum_{k=1}^{N-1} \bigg( \frac{\pi^{\textbf{G}}_{p_k(t)}(t)}{\mu_{(k)}(t)}-\frac{\pi^{\textbf{G}}_{p_{k+1}(t)}(t)}{\mu_{(k+1)}(t)} \bigg) dL^{\mu_{(k)}-\mu_{(k+1)}}(t)
		\\
		&-\frac{1}{2}\int_0^{\cdot} \sum_{k=1}^N \sum_{\ell=1}^N \frac{D^2_{k, \ell} G\big(\bm{\mu}(t)\big)}{G\big(\bm{\mu}(t)\big)} d[\mu_{(k)}, \mu_{(\ell)}](t).
	\end{align*}
\end{thm}

\smallskip

The portfolio $\pi^{\textbf{G}}$ generated via the recipe \eqref{def : pi^G} is easily checked to be a stock portfolio, i.e., $\pi^{\textbf{G}} \in \pazocal{I}(R) \cap \Delta^{N-1}$; however, it is not generally a portfolio among the top $n$ stocks. In order to make it a portfolio among the top $n$ stocks, we need to impose two conditions on the function $G$ in Theorem~\ref{Thm : FGP rank}:
\begin{enumerate} [(A)]
	\item $G$ is `balanced', i.e., satisfies the identity
	\begin{equation}		\label{con : balance}
		G(x_1, \cdots, x_N) = \sum_{j=1}^N x_jD_jG(x_1, \cdots, x_N), \qquad \text{for any} \quad x \in U,
	\end{equation}
	\item $G(x)$ depends only on the first $n$ components of $x$.
\end{enumerate}
If the condition (A) is satisfied, then the portfolio $\pi^{\textbf{G}}$ of \eqref{def : pi^G} has a simpler representation as
\begin{equation}		\label{Eq : pi^G balanced}
	\pi^{\textbf{G}}_{p_k(t)}(t) =  \frac{D_kG\big(\bm{\mu}(t)\big)}{G\big(\bm{\mu}(t)\big)} \mu_{(k)}(t), \qquad \text{for} \quad t \geq 0.
\end{equation}
Moreover, if the condition (B) holds as well, then $D_kG\big(\bm{\mu}\big) = 0$ for $k > n$, thus $\pi^{\textbf{G}}_{p_k(t)}(t) = 0$ for $k > n$. This means that the portfolio $\pi^{\textbf{G}}$ does not invest in the $i=p_k(t)$-th stock at time $t$, if the rank $k$ of this $i$-th stock is bigger than $n$ at time $t$.

\smallskip

\begin{defn}	[Admissible generating function in open market]	\label{Def : admissible GF}
	We call a function $G$ in Theorem~\ref{Thm : FGP rank} an \textit{admissible generating function of market consisting of the top $n$ stocks}, if it satisfies conditions (A) and (B) above. 
\end{defn}

\smallskip

\begin{cor}			\label{cor : FGP rank}
	If $G$ in Theorem~\ref{Thm : FGP rank} is an admissible generating function of market consisting of top $n$ stocks, then $\textbf{G}$ generates the stock portfolio among the top $n$ stocks $\pi^{\textbf{G}} \in \pazocal{I}(R) \cap \pazocal{T}(n) \cap \Delta^{N-1}$, given as \eqref{Eq : pi^G balanced} for $k = 1, \cdots, N$. In this case, we have the master formula
	\begin{align}
	\log \bigg( \frac{X_{\pi^{\textbf{G}}}}{X_{\mu}} \bigg) 
	= \log \bigg( \frac{G(\bm{\mu})}{G\big(\bm{\mu}(0)\big)} \bigg) &- \frac{1}{2} \int_0^{\cdot} \sum_{k=1}^{n} \bigg( \frac{D_kG\big(\bm{\mu}(t)\big)}{G\big(\bm{\mu}(t)\big)}-\frac{D_{k+1}G\big(\bm{\mu}(t)\big)}{G\big(\bm{\mu}(t)\big)} \bigg) dL^{\mu_{(k)}-\mu_{(k+1)}}(t)			\nonumber
	\\
	&-\frac{1}{2}\int_0^{\cdot} \sum_{k=1}^n \sum_{\ell=1}^n \frac{D^2_{k, \ell} G\big(\bm{\mu}(t)\big)}{G\big(\bm{\mu}(t)\big)} d[\mu_{(k)}, \mu_{(\ell)}](t).		\label{Eq : master formula balanced}
	\end{align}
\end{cor}

\bigskip

\begin{example} [Balanced functions]
	By solving the partial differential equation of \eqref{con : balance}, a balanced function $G$ can be shown to be homogeneous of degree $1$, i.e, the identity
	\begin{equation}
	G(ax) = aG(x)
	\end{equation}
	holds for any $x \in U$ and $a > 0$. From this simple characterization of balanced functions, we illustrate three types of balanced functions here:
	\begin{enumerate}[(i)]
		\item $G(x) = \frac{1}{c_1 + \cdots + c_N}\sum_{i=1}^{N} c_i x_i$,
		\item $G(x) = \big(\prod_{i=1}^{N} x_i \big)^{1/N}$,
		\item $G(x) = \big(\sum_{i=1}^{N} x_i^p \big)^{\frac{1}{p}}$.
	\end{enumerate}
	These functions are closely related to `three Pythagorean means'; (i) and (ii) are just the weighted-arithmetic and geometric means of the components of $x$, and (iii) becomes the harmonic mean when $p=-1$. A plethora of examples of these types can be found in the literature. The ``capitalization-weighted portfolios'' of large and small stocks from Example 6.2, Example 6.3 of \cite{Karatzas:Ruf:2017}, or from Example 4.3.2 of \cite{Fe} are special cases of (i). The ``equal-weighted portfolio'', which holds equal weights across all assets, in Section 4.3 of \cite{Karatzas:Ruf:2017}, is generated by (ii). The portfolio generated by (iii) for $0 < p < 1$ is called ``diversity-weighted portfolio'', and is discussed in detail in Example 3.4.4 and Section 6.2 of \cite{Fe}. Diversity-weighted portfolios with negative parameter $p<0$ in (iii) are the main subject of \cite{Vervuurt:Karatzas:2015}.
	
	\smallskip
	
	We can slightly generalize and make these functions satisfy conditions (A) and (B) as well:
	\begin{enumerate}[(i')]
		\item $G(x) = \sum_{i=1}^{n} c_i x_i$,
		\item $G(x) = \prod_{i=1}^{n} x_i^{c_i}, \quad \text{with} \quad \sum_{i=1}^{n}c_i=1$,
		\item $G(x) = \Big(\sum_{i=1}^{n} x_i^p \Big)^{\frac{1}{p}}$,
	\end{enumerate}
	for some constants $c_i$'s and $p$.
\end{example}

\bigskip

The following example further devolops Example 4.3.2 of \cite{Fe}, and shows that the top $n$ market portfolio $\widetilde{\mu}$, defined in \eqref{def : top n market portfolio}, can be generated functionally.

\smallskip

\begin{example} [Top $n$ market portfolio]
	Consider the function 
	\begin{equation*}
		\textbf{G}(x) = G(x_{(1)}, \cdots, x_{(n)}) = \sum_{k=1}^n x_{(k)}
	\end{equation*}
	satisfying the conditions (A) and (B) above. Corollary~\ref{cor : FGP rank} implies that $\textbf{G}$ generates the portfolio
	\begin{equation*}
		\pi^{\textbf{G}}_{p_k(\cdot)}(\cdot) = \frac{\mu_{(k)}(\cdot)}{\mu_{(1)}(\cdot)+\cdots+\mu_{(n)}(\cdot)} \bm{1}_{\{ k \leq n \}}
		= \frac{S_{(k)}(\cdot)}{S_{(1)}(\cdot)+\cdots+S_{(n)}(\cdot)}\bm{1}_{\{ k \leq n \}}.
	\end{equation*}
	This coincides with the top $n$ market portfolio $\widetilde{\mu}$, because
	\begin{equation*}
		\frac{S_{(k)}(\cdot)\bm{1}_{\{ k \leq n \}}}{S_{(1)}(\cdot)+\cdots+S_{(n)}(\cdot)} = \frac{S_{p_k(\cdot)}(\cdot)\bm{1}_{\{ k \leq n \}}}{\widetilde{\Sigma}(\cdot)}
	= \widetilde{\mu}_{p_k(\cdot)}(\cdot),
	\end{equation*}
	holds for $k = 1, \cdots, N$, from \eqref{def : top n market portfolio}. The master formula \eqref{Eq : master formula balanced} is then
	\begin{equation}		\label{Eq : mu vs tilde mu}
		\log \bigg( \frac{X_{\widetilde{\mu}}}{X_{\mu}} \bigg) = \log \bigg( \frac{\mu_{(1)}(\cdot)+\cdots+\mu_{(n)}(\cdot)}{\mu_{(1)}(0)+\cdots+\mu_{(n)}(0)} \bigg)- \frac{1}{2} \int_0^{\cdot} \frac{dL^{\mu_{(n)}-\mu_{(n+1)}}(t)}{\mu_{(1)}(t)+\cdots+\mu_{(n)}(t)}.
	\end{equation}
\end{example}

\smallskip

In Corollary~\ref{cor : FGP rank}, the portfolio $\pi^{\textbf{G}}$ is indeed a stock portfolio among the top $n$ stocks; but the master formula \eqref{Eq : master formula balanced} compares its performance with the market portfolio $\mu$, which is not a portfolio among the top $n$ stocks. In the open market setting, since we only consider portfolios among the top $n$ stocks, it is more appropriate to compare a portfolio's performance with respect to $\widetilde{\mu}$, rather than $\mu$. This can be done by combining \eqref{Eq : master formula balanced} and \eqref{Eq : mu vs tilde mu}.

\smallskip

\begin{cor}	[Master formula in top $n$ market]		\label{cor : FGP in open market}
	For a functionally generated portfolio $\pi^{\textbf{G}}$ as in Corollary~\ref{cor : FGP rank}, the master formula, which compares the log-relative wealth of $\pi^{\textbf{G}}$ to that generated by $\widetilde{\mu}$, the top $n$ market, is given as
	\begin{align}
		\log \bigg( \frac{X_{\pi^{\textbf{G}}}}{X_{\widetilde{\mu}}} \bigg)
		&= \log \bigg( \frac{G(\bm{\mu})}{G\big(\bm{\mu}(0)\big)} \bigg) - \log \bigg( \frac{\mu_{(1)}(\cdot)+\cdots+\mu_{(n)}(\cdot)}{\mu_{(1)}(0)+\cdots+\mu_{(n)}(0)} \bigg)		\label{Eq : master formula for tilde mu}
		\\
		&~~~-\frac{1}{2} \int_0^{\cdot} \sum_{k=1}^{n} \bigg( \frac{D_kG\big(\bm{\mu}(t)\big)}{G\big(\bm{\mu}(t)\big)}-\frac{D_{k+1}G\big(\bm{\mu}(t)\big)}{G\big(\bm{\mu}(t)\big)} \bigg) dL^{\mu_{(k)}-\mu_{(k+1)}}(t)		\nonumber
		\\
		&~~~+ \frac{1}{2} \int_0^{\cdot} \frac{dL^{\mu_{(n)}-\mu_{(n+1)}}(t)}{\mu_{(1)}(t)+\cdots+\mu_{(n)}(t)}
		-\frac{1}{2}\int_0^{\cdot} \sum_{k=1}^n \sum_{\ell=1}^n \frac{D^2_{k, \ell} G\big(\bm{\mu}(t)\big)}{G\big(\bm{\mu}(t)\big)} d[\mu_{(k)}, \mu_{(\ell)}](t).		\nonumber
	\end{align}
\end{cor}

\smallskip

We call this formula of \eqref{Eq : master formula for tilde mu}, the `master formula for the top $n$ market' to distinguish it from the formula of \eqref{Eq : master formula balanced}, which we call the `master formula in the entire market'.

\bigskip

\begin{example} [Diversity-weighted portfolio]
	Consider a function
	\begin{equation*}
		\textbf{G}(x) = G(x_{(1)}, \cdots, x_{(N)}) = \bigg(\sum_{k=1}^n x_{(k)}^p\bigg)^{1/p}	
	\end{equation*}
	with a fixed constant $p \in (0, 1)$. Corollary~\ref{cor : FGP rank} implies that $\textbf{G}$ generates the ``diversity-weighted portfolio''
	\begin{equation*}
	\pi^{\textbf{G}}_{p_k(\cdot)}(\cdot) = \frac{\mu_{(k)}^p(\cdot)}{\mu_{(1)}^p(\cdot)+\cdots+\mu_{(n)}^p(\cdot)} \bm{1}_{\{ k \leq n \}}, \qquad k = 1, \cdots, N.
	\end{equation*}
	The master formula in the top $n$ market in \eqref{Eq : master formula for tilde mu} is then given as
	\begin{align}
	\log \bigg( \frac{X_{\pi^{\textbf{G}}}}{X_{\widetilde{\mu}}} \bigg)
	&= \frac{1}{p} \log \bigg( \frac{\mu_{(1)}^p(\cdot)+\cdots+\mu_{(n)}^p(\cdot)}{\mu_{(1)}^p(0)+\cdots+\mu_{(n)}^p(0)} \bigg)- \log \bigg( \frac{\mu_{(1)}(\cdot)+\cdots+\mu_{(n)}(\cdot)}{\mu_{(1)}(0)+\cdots+\mu_{(n)}(0)} \bigg)		\nonumber
	\\
	&-\frac{1}{2} \int_0^{\cdot} \frac{\mu_{(n)}^{p-1}(t)}{\mu_{(1)}^p(t)+\cdots+\mu_{(n)}^p(t)} dL^{\mu_{(n)}-\mu_{(n+1)}}(t)		\label{vanishing local time terms}
	+ \frac{1}{2} \int_0^{\cdot} \frac{dL^{\mu_{(n)}-\mu_{(n+1)}}(t)}{\mu_{(1)}(t)+\cdots+\mu_{(n)}(t)}
	\\
	&-\frac{1-p}{2}\int_0^{\cdot} \sum_{k=1}^n \sum_{\ell=1}^n \frac{\mu_{(k)}^{p-1}(t)\mu_{(\ell)}^{p-1}(t)}{\big(\mu_{(1)}^p(t)+\cdots+\mu_{(n)}^p(t)\big)^2} d[\mu_{(k)}, \mu_{(\ell)}](t)		\nonumber
	\\
	&+\frac{1-p}{2}\int_0^{\cdot} \sum_{k=1}^n \frac{\mu_{(k)}^{p-2}(t)}{\mu_{(1)}^p(t)+\cdots+\mu_{(n)}^p(t)} d[\mu_{(k)}, \mu_{(k)}](t).		\nonumber
	\end{align}
	Here, in the first integral of \eqref{vanishing local time terms}, we use the fact that the local time process $L^{\mu_{(k)}-\mu_{(k+1)}}(\cdot)$ is flat off the set $\{s \geq 0: \mu_{(k)}(s) = \mu_{(k+1)}(s) \}$ for $k=1, \cdots, n-1$.
\end{example}

\bigskip

\subsection{Universal portfolio}
 
In this subsection, we explore Cover's universal portfolio theory in open markets. This portfolio was first introduced by \cite{Cover_1991} in discrete time, and its extension to continuous time was developed by \cite{Jamshidian_1992}. More recent work under the setting of Stochastic Portfolio Theory can be found in \cite{Cuchiero_Scha_Wong}.

\smallskip

Recalling the notation $\Delta^{N-1}_+$ from \eqref{def : simplex}, we need first the following notation
\begin{equation}		\label{def : Delta N-1, n}
	\Delta^{N-1, n}_+ := \big\{ x \in \mathbb{R}^N ~\Big|~ x_k \geq 0 ~~ \text{for} ~~ k = 1, \cdots, N, ~~ \sum_{k=1}^n x_k = 1, ~~ x_{n+1} = \cdots = x_N = 0 \big\}
\end{equation}
throughout this subsection. Since we are only allowed to invest in the top $n$ stocks in an open market, the notion of Cover's `constant rebalanced portfolio' needs to be amended, as follows.

\smallskip

\begin{defn} [constant rebalanced portfolio by rank]	\label{Def : constant rebalanced portfolio}
	If a stock portfolio $\pi \in \pazocal{I}(R) \cap \pazocal{T}(n) \cap \Delta^{N-1}$ among the top $n$ stocks satisfies 
	\begin{equation}	\label{Eq : constant rebalanced}
		\pi_{p_k(t)}(t) = \xi_k \quad \text{for} \quad t \geq 0, \quad k=1, \cdots, N
	\end{equation}
	with some $\xi = (\xi_1, \cdots, \xi_N) \in \Delta^{N-1, n}_+$, we call $\pi$ a \textit{constant rebalanced portfolio among the top $n$ stocks by rank}. This portfolio re-balances at all times to maintain a constant proportion $\xi_k$ of current wealth invested in the $k$-th ranked stock, for $k \leq n$. We denote the collection of constant rebalanced portfolios among the top $n$ stocks by $\pazocal{CR}^n$.
\end{defn}

\smallskip

\begin{prop}		\label{prop : constant rebalanced is FG}
	Every constant rebalanced portfolio among the top $n$ stocks by rank is functionally generated.
\end{prop}

\smallskip

\begin{proof}
	For a fixed $\xi \in \Delta^{N-1, n}_+$, consider a function 
	\begin{equation}		\label{Eq : CR generating function}
		\textbf{G}(x) = G(x_{(1)}, \cdots, x_{(N)}) = \prod_{k=1}^n x_{(k)}^{\xi_k}.
	\end{equation}
	It is easy to check that $G$ is an admissible generating function of market consisting of the top $n$ stocks, and it generates the portfolio via the recipe \eqref{Eq : pi^G balanced}:
	\begin{equation*}
		\pi^{\textbf{G}}_{p_k(t)}(t) =  \xi_k, \quad \text{for} \quad t \geq 0, \quad k=1, \cdots, N.
	\end{equation*}
	Since $\xi$ is chosen arbitrarily from $\Delta^{N-1, n}_+$, the claim follows.
\end{proof}

\smallskip

Thanks to Proposition~\ref{prop : constant rebalanced is FG}, for every $\xi \in \Delta^{N-1, n}_+$ there exists a corresponding portfolio $\pi \in \pazocal{CR}^n$ as in Definition~\ref{Def : constant rebalanced portfolio}, and we write $X_{\xi}(t)$ to represent the wealth process of $\pi$ at time $t$ in the manner of \eqref{def : wealth}, namely,
\begin{equation}			\label{def : wealth of xi}
	X_{\xi}(t) \equiv X_{\pi}(t) 
	= \pazocal{E}\bigg( \int_0^{t} \sum_{i=1}^N \pi_i(s) dR_i(s) \bigg)
	= \pazocal{E}\bigg( \sum_{k=1}^n \xi_k \int_0^{t} \sum_{i=1}^N \mathbbm{1}_{\{ u_i(s) = k \}}dR_i(s) \bigg) \qquad \text{for} \quad t \geq 0.
\end{equation}

\bigskip

For $T > 0$ fixed, we define
\begin{equation}			\label{def : best retrospectively chosen portfolio}
	X^*(T) := \sup_{\pi \in \pazocal{CR}^n} X_{\pi}(T) = \sup_{\xi \in \Delta^{N-1, n}_+} X_{\xi}(T).
\end{equation}
This $X^*(T)$ represents the maximal wealth at time $T$, achievable over all constant rebalanced portfolios among the top $n$ stocks by rank. We show in the following that a $\pazocal{F}(T)$-measurable random vector of weights $\pi^*(T) \equiv \xi^*$ exists, which attains the supremum in \eqref{def : best retrospectively chosen portfolio}, namely, that $X^*(T) = X_{\pi^*(T)}(T) = X_{\xi^*}(T)$ holds.

\smallskip

\begin{lem}		\label{lem : continuity}
	For a fixed $T > 0$, the mapping $\Delta^{N-1, n}_+ \ni \xi \mapsto X_{\xi}(T) \in \mathbb{R}$ is continuous.
\end{lem}

\begin{proof}
	For $\xi, \zeta \in \Delta^{N-1, n}_+$, we have
	\begin{align}
		&~~~\log X_{\xi}(T) - \log X_{\zeta}(T)
		= \log \frac{X_{\xi}(T)}{X_{\widetilde{\mu}}(T)} - \log \frac{X_{\zeta}(T)}{X_{\widetilde{\mu}}(T)}						\nonumber
		\\
		&= \sum_{k=1}^n \log \bigg(\frac{\mu_{(k)}(T)}{\mu_{(k)}(0)}\bigg) ^{(\xi_k-\zeta_k)}
					\label{Eq : showing continuity}
		\\
		&~~~- \frac{1}{2} \int_0^{T} \sum_{k=1}^{n-1} \bigg( \frac{\xi_k-\zeta_k}{\mu_{(k)}(t)}-\frac{\xi_{k+1}-\zeta_{k+1}}{\mu_{(k+1)}(t)} \bigg) dL^{\mu_{(k)}-\mu_{(k+1)}}(t)
		- \frac{1}{2} \int_0^{T} \frac{\xi_n-\zeta_n}{\mu_{(n)}(t)} dL^{\mu_{(n)}-\mu_{(n+1)}}(t)									\nonumber
		\\
		&~~~-\frac{1}{2}\int_0^{T} \sum_{k=1}^n \sum_{\ell=1}^n \frac{\xi_k\xi_{\ell}-\zeta_k\zeta_{\ell}}{\mu_{(k)}(t)\mu_{(\ell)}(t)} d[\mu_{(k)}, \mu_{(\ell)}](t)
		+\frac{1}{2}\int_0^{T} \sum_{k=1}^n \frac{\xi_k\zeta_k}{\mu_{(k)}^2(t)} d[\mu_{(k)}, \mu_{(k)}](t).													\nonumber
	\end{align}
	In the last equality, we used the master formula \eqref{Eq : master formula for tilde mu} twice, and applied it to the functions of the form \eqref{Eq : CR generating function} for $\xi$ and $\zeta$, respectively. Since the functions $\mu_{(k)}(\cdot),~ L^{\mu_{(k)}-\mu_{(k+1)}}(\cdot),~ [\mu_{(k)}, \mu_{(\ell)}](\cdot)$ for $1 \leq k, \ell \leq n$ on the right-hand side \eqref{Eq : showing continuity} are all continuous, they are bounded on the compact interval $[0, T]$. Thus, we obtain the estimate $\big|\log X_{\xi}(T) - \log X_{\zeta}(T)\big| \leq ||\xi-\zeta||K_T$ for some positive constant $K_T$, which depends on $\min_{0 \leq t \leq T}\mu_{(k)}(t), ~L^{\mu_{(k)}-\mu_{(k+1)}}(T),~ [\mu_{(k)}, \mu_{(k)}](T)$ for $k = 1, \cdots, n$, and this proves the continuity.
\end{proof}

\smallskip

\begin{defn} [Best retrospectively chosen vector of weights]
	\label{Def : best retrospective portfolio}
	The continuity shown in Lemma~\ref{lem : continuity} shows that there exists a vector $\xi^* \equiv \pi^*(T) \in \Delta^{N-1, n}_+$ which attains the supremum in \eqref{def : best retrospectively chosen portfolio} for a fixed $T \in (0, \infty)$. We call this $\pazocal{F}(T)$-measurable, $\Delta^{N-1, n}_+$-valued random variable $\pi^*(T)$ the \textit{best retrospectively chosen vector of weights among the top $n$ stocks for the given $T \in (0, \infty)$}.
\end{defn}

\smallskip

Even though $\pi^*(T)$ meant to outperform all constant rebalanced portfolios among the top $n$ stocks by rank at $T>0$, constructing it requires the knowledge of stock prices over the entire interval $[0, T]$, that is, ahead of time. \cite{Cover_1991} introduced a remarkable way to construct a portfolio, called ``universal portfolio'', depending only on past stock prices, whose long-run performance is almost as good as that of the best retrospectively chosen vector of weights. Cover's idea of building the universal portfolio, was to determine its weights by averaging the performances of all constant portfolio weights, at any time $t \geq 0$.

\smallskip

\begin{defn} [universal portfolio]			\label{Def : universal portfolio}
	With the notation $\Delta^{N-1, n}_+$ of \eqref{def : Delta N-1, n}, the portfolio $\hat{\pi}$, defined as 
	\begin{equation}		\label{def : universal portfolio}
		\hat{\pi}_{p_k(t)}(t) := \frac{\int_{\Delta^{N-1, n}_+} \xi_k X_{\xi}(t) d\xi}{\int_{\Delta^{N-1, n}_+} X_{\xi}(t) d\xi} \qquad \text{for} \quad t \geq 0, \quad k = 1, \cdots, N,
	\end{equation}
	is called \textit{universal portfolio among the top $n$ stocks}.
\end{defn}

\smallskip

From the notation $\Delta^{N-1, n}_+$, we have $\hat{\pi}_{p_k(t)}(t) = 0$ for all $t \geq 0$ for $k > n$; i.e., $\hat{\pi}$ invests only in the top $n$ stocks, thus it belongs to $\pazocal{I}(R) \cap \pazocal{T}(n) \cap \Delta^{N-1}$, the collections of stock portfolios among the top $n$ stocks. We next compute the wealth of the universal portfolio.

\smallskip

\begin{prop}
	The wealth process $X_{\hat{\pi}}$ is given as
	\begin{equation}		\label{Eq : universal wealth}
		X_{\hat{\pi}}(t) = \frac{\int_{\Delta^{N-1, n}_+} X_{\xi}(t) d\xi}{\int_{\Delta^{N-1, n}_+} d\xi}, \qquad \text{for} \quad t \geq 0.
	\end{equation}
\end{prop}

\begin{proof}
	Let $Z(t)$ denote the right-hand side of \eqref{Eq : universal wealth}. We have
	\begin{align*}
		\frac{dZ(t)}{Z(t)}
		= \frac{\int_{\Delta^{N-1, n}_+} dX_{\xi}(t) d\xi}{\int_{\Delta^{N-1, n}_+} X_{\xi}(t) d\xi} 
		&= \frac{\int_{\Delta^{N-1, n}_+} X_{\xi}(t) \sum_{i=1}^N \sum_{k=1}^n \xi_k \mathbbm{1}_{\{ u_i(t) = k \}} dR_i(t) d\xi}{\int_{\Delta^{N-1, n}_+} X_{\xi}(t) d\xi}
		\\
		&= \sum_{i=1}^N \sum_{k=1}^n \hat{\pi}_{p_k(t)}(t) \mathbbm{1}_{\{ u_i(t) = k \}} dR_i(t)
		= \sum_{i=1}^N \hat{\pi}_i(t) dR_i(t)
		= \frac{dX_{\hat{\pi}}(t)}{X_{\hat{\pi}}(t)}.
	\end{align*}
	Here, the second, third and last equalities are from \eqref{def : wealth of xi}, \eqref{def : universal portfolio}, and \eqref{def : wealth}, respectively. Since $X_{\hat{\pi}}(0) = Z(0) = 1$, the result follows.
\end{proof}

\bigskip

We are now ready to compare the long-run performance of the universal portfolio with the best retrospectively chosen vector of weights.

\smallskip

\begin{thm}
	Suppose that the portfolio $\mu$, defined in \eqref{def : market portfolio}, satisfies
	\begin{equation}
		\mu_{(1)}(t) \geq \cdots \geq \mu_{(n)}(t) \geq \delta, \quad \text{for all} \quad t \geq 0 \quad \text{for some} \quad \delta > 0,
	\end{equation}
	\begin{equation}		\label{con : asymptotic}
		\limsup_{T \rightarrow \infty} \frac{1}{T}[\mu_{(k)}, \mu_{(k)}](T) < \infty, \qquad \limsup_{T \rightarrow \infty} \frac{1}{T}L^{\mu_{(k)}-\mu_{(k+1)}}(T) < \infty, \quad \text{for} \quad k = 1, \cdots, n.
	\end{equation}
	Then, the best retrospectively chosen vector of weights and the universal portfolio have the same asymptotic growth rate; that is,
	\begin{equation}		\label{Eq : same asymptotic growth rate}
		\lim_{T \rightarrow \infty} \frac{1}{T}\Big(\log X_{\pi^*(T)}(T) - \log X_{\hat{\pi}}(T) \Big) = 0,
	\end{equation}
	where $\pi^*(T)$ and $\hat{\pi}$ are as in Definitions~\ref{Def : best retrospective portfolio} and \ref{Def : universal portfolio}, respectively.
\end{thm}

\smallskip

\begin{proof}
	Since $X_{\pi^*(T)}(T) \geq X_{\xi}(T)$ holds for every $\xi \in \Delta^{N-1, n}_+$ for every $T \geq 0$, the inequality ``$\geq$'' of \eqref{Eq : same asymptotic growth rate} is obvious from \eqref{Eq : universal wealth}.
	
	\smallskip
	
	We now show the reverse inequality. Let $\xi^* \in \Delta^{N-1, n}_+$ be the corresponding vector of weights $\pi^*(T)$ as in Definition~\ref{Def : best retrospective portfolio}. For any $\xi \in \Delta^{N-1, n}_+$ satisfying $||\xi^* - \xi|| \leq \eta$ for some $\eta > 0$, we have the estimate
	\begin{align*}
		\frac{1}{T}\Big(\log X_{\xi}(T) - \log X_{\xi^*}(T) \Big) 
		&\geq -\frac{\eta}{T}\Big(\frac{a_n}{\delta}\max_{1 \leq k \leq n} L^{\mu_{(k)}-\mu_{(k+1)}}(T) + \frac{b_n}{\delta^2}\max_{1 \leq k \leq n} [\mu_{(k)}, \mu_{(k)}](T) \Big)
		\\
		&=: - \frac{\eta }{T}K_T,
	\end{align*} 
	in the same manner as in the proof of Lemma~\ref{lem : continuity}, for some positive constants $a_n$ and $b_n$ depending on $n$. Due to the condition \eqref{con : asymptotic}, we can take $\eta$ sufficiently small such that $\frac{\eta}{T}K_T \leq \epsilon$ holds for every $T \geq 1$, for any given $\epsilon > 0$. To summarize, for any given $\epsilon > 0$, there exists $\eta > 0$ such that 
	\begin{equation}		\label{Eq : estimate in the ball}
		\frac{1}{T}\big(\log X_{\xi}(T) - \log X_{\xi^*}(T) \big) \geq -\epsilon
	\end{equation}
	holds for every $\xi \in B(\xi^*, \eta)$ and for every $T \geq 1$. Here, $B(\xi^*, \eta)$ is the intersection of $\Delta^{N-1, n}_+$ and $|| \cdot ||$-ball in $\mathbb{R}^N$ centered at $\xi^*$ with radius $\eta$. We denote $V_{B(\xi^*, \eta)}$ and $V_{\Delta^{N-1, n}_+}$ the volume of $B(\xi^*, \eta)$ and the volume of the subset $\Delta^{N-1, n}_+$ of $\mathbb{R}^N$, respectively.
	
	\smallskip
	
	From \eqref{Eq : universal wealth} and Jensen's inequality, we have
	\begin{align*}
		\bigg( \frac{X_{\hat{\pi}}(T)}{X_{\pi^*(T)}(T)} \bigg)^{\frac{1}{T}}
		= \bigg( \frac{\int_{\Delta^{N-1, n}_+} X_{\xi}(T)d\xi}{X_{\xi^*}(T) ~ V_{\Delta^{N-1, n}_+}} \bigg)^{\frac{1}{T}}
		&\geq \bigg( \frac{\int_{B(\xi^*, \eta)} X_{\xi}(T)d\xi}{X_{\xi^*}(T) ~ V_{\Delta^{N-1, n}_+}} \bigg)^{\frac{1}{T}}
		\\
		& \geq \frac{\big(V_{B(\xi^*, \eta)}\big)^{\frac{1}{T}-1} \int_{B(\xi^*, \eta)} X_{\xi}(T)^{\frac{1}{T}} d\xi}{\big(X_{\xi^*}(T)\big)^{\frac{1}{T}} ~ \big(V_{\Delta^{N-1, n}_+}\big)^{\frac{1}{T}}}
		\\
		&= \frac{\big(V_{B(\xi^*, \eta)}\big)^{\frac{1}{T}-1}}{\big(V_{\Delta^{N-1, n}_+}\big)^{\frac{1}{T}}} \int_{B(\xi^*, \eta)} \bigg( \frac{X_{\xi}(T)}{X_{\xi^*}(T)}\bigg)^{\frac{1}{T}} d\xi
		\geq \bigg( \frac{V_{B(\xi^*, \eta)}}{V_{\Delta^{N-1, n}_+}} \bigg)^{\frac{1}{T}} e^{-\epsilon},
	\end{align*}
	where the last inequality is from \eqref{Eq : estimate in the ball}. Taking logarithms, then letting $T \rightarrow \infty$ for any given $\epsilon > 0$, the desired inequality follows.
\end{proof}

\bigskip

\section{Conclusion}
\label{sec: conclusion}

Most of the results in Section~\ref{sec: main}, including the main Theorem~\ref{thm : main result}, a foundational result of equity market structure and of the study of arbitrage in open markets, can be formulated quite simply in terms of the local characteristics $\widetilde{\alpha}$ and $\widetilde{c}$ of the open market, defined in \eqref{def : tilde alpha}, \eqref{def : tilde c}. In particular, the supermartingale num\'eraire portfolio $\rho$ in the top $n$ open market, if it exists, should satisfy the equation $\widetilde{\alpha} = \widetilde{c}\rho$ of \eqref{Eq : equiv statement (3)}. From this equation, we were able to conclude that the supermartingale num\'eraire portfolio $\rho$ in the open market takes the form of $\rho = D\widetilde{c}^{\dagger}\widetilde{\alpha}$. Here, multiplying by the diagonal matrix $D$ of \eqref{def : diagonal matrix D} makes the portfolio invest only in the top $n$ stocks, while maintaining its supermartingale num\'eraire property.

\smallskip

However, as foretold in the introductory part of Section~\ref{sec: stock portfolio and FGP}, we cannot use this technique to deal with stock portfolios; multiplying by $D$ a stock portfolio in order to make it invest only in the top $n$ stocks, destroys its self-financing property. For example, a unit vector $\pi := e^1 = (1, 0, \cdots, 0)$ is a stock portfolio which invests all capital into the first stock, but $D\pi$ is not a stock portfolio as it invests all wealth into the money market whenever the first stock fails to belong to the top $n$ market. Thus, for stock portfolios in open markets, a different approach is offered. Fernholz's functional generation of stock portfolios with ranked market weights, under the extra conditions (A) and (B) of Definition~\ref{Def : admissible GF}, provides a systematic way to construct stock portfolios that invest only in the top $n$ open market. This approach also yields the `master formula' in Corollary~\eqref{cor : FGP in open market}, which allows comparing these portfolios with the top $n$ market portfolio, $\widetilde{\mu}$. As an application of this formula, we could prove that Cover's result on the universal portfolio is also valid in open markets. 

\smallskip

Nonetheless, there are a lot of limitations when considering stock portfolios in open markets. First, the balance condition \eqref{con : balance} significantly restricts the class of generating functions in open markets. Moreover, the local time terms which appear on the right-hand side of the master formula \eqref{Eq : master formula for tilde mu}, make it very difficult to find stock portfolios in open markets which outperform $\widetilde{\mu}$. These difficulties are an inevitable price to pay for dealing with stock portfolios in open markets.



\newpage

\bibliography{aa_bib}

\begin{thebibliography}{}

\bibitem[Banner and Ghomrasni, 2008]{Banner:Ghomrasni}
Banner, A.~D. and Ghomrasni, R. (2008).
\newblock Local times of ranked continuous semimartingales.
\newblock {\em Stochastic Process. Appl.}, 118(7):1244--1253.

\bibitem[Cover, 1991]{Cover_1991}
Cover, T. (1991).
\newblock Universal portfolios.
\newblock {\em Mathematical Finance}, 1(1):1--29.

\bibitem[Cuchiero et~al., 2019]{Cuchiero_Scha_Wong}
Cuchiero, C., Schachermayer, W., and Wong, T.-K.~L. (2019).
\newblock Cover's universal portfolio, stochastic portfolio theory, and the
  num{\'e}raire portfolio.
\newblock {\em Mathematical Finance}, 29(3):773--803.

\bibitem[Fernholz, 2002]{Fe}
Fernholz, E.~R. (2002).
\newblock {\em Stochastic Portfolio Theory}, volume~48 of {\em Applications of
  Mathematics (New York)}.
\newblock Springer-Verlag, New York.
\newblock Stochastic Modelling and Applied Probability.

\bibitem[Fernholz, 1999]{F_generating}
Fernholz, R. (1999).
\newblock Portfolio generating functions.
\newblock In Avellaneda, M., editor, {\em Quantitative Analysis in Financial
  Markets}. World Scientific.

\bibitem[Fernholz and Karatzas, 2009]{FK_survey}
Fernholz, R. and Karatzas, I. (2009).
\newblock Stochastic {P}ortfolio {T}heory: an overview.
\newblock In Bensoussan, A., editor, {\em Handbook of Numerical Analysis},
  volume Mathematical Modeling and Numerical Methods in Finance. Elsevier.

\bibitem[Jamshidian, 1992]{Jamshidian_1992}
Jamshidian, F. (1992).
\newblock Asymptotically optimal portfolios.
\newblock {\em Mathematical Finance}, 2(2):131--150.

\bibitem[Karatzas and Kardaras, 2020]{KK2}
Karatzas, I. and Kardaras, C. (2020).
\newblock {\em Portfolio Theory and Arbitrage}.
\newblock To be published.

\bibitem[Karatzas and Ruf, 2017]{Karatzas:Ruf:2017}
Karatzas, I. and Ruf, J. (2017).
\newblock Trading strategies generated by {L}yapunov functions.
\newblock {\em Finance and Stochastics}, 21(3):753--787.

\bibitem[Karatzas and Shreve, 1991]{KS1}
Karatzas, I. and Shreve, S.~E. (1991).
\newblock {\em Brownian Motion and Stochastic Calculus}, volume 113 of {\em
  Graduate Texts in Mathematics}.
\newblock Springer-Verlag, New York, second edition.

\bibitem[Larsen and {\v{Z}}itkovi{\'c}, 2007]{Lar:Zit:2007}
Larsen, K. and {\v{Z}}itkovi{\'c}, G. (2007).
\newblock Stability of utility-maximization in incomplete markets.
\newblock {\em Stochastic Process. Appl.}, 117(11):1642--1662.

\bibitem[Schweizer, 1995]{Schweizer:1995}
Schweizer, M. (1995).
\newblock On the minimal martingale measure and the {F}{\"o}llmer-{S}chweizer
  decomposition.
\newblock {\em Stochastic Analysis and Applications}, 13(5):573--599.

\bibitem[Vervuurt and Karatzas, 2015]{Vervuurt:Karatzas:2015}
Vervuurt, A. and Karatzas, I. (2015).
\newblock Diversity-weighted portfolios with negative parameter.
\newblock {\em Ann. Finance}, 11(3-4):411--432.

\end{thebibliography}
\bibliographystyle{apalike}

\end{document}